\documentclass[pra,aps,10pt,twocolumn,nofootinbib,superscriptaddress,longbibliography]{revtex4-2}
\usepackage{algorithm2e}
\usepackage{amsfonts}
\usepackage{amsmath}
\usepackage{amssymb}
\usepackage{bm}
\usepackage{mathtools}
\usepackage{blindtext}
\usepackage{array}
\usepackage{multirow}
\usepackage{subcaption}
\usepackage{diagbox}
\usepackage{graphicx}
\usepackage[ colorlinks = true, linkcolor = blue, urlcolor=blue, citecolor = red, anchorcolor = green,
]{hyperref}

\usepackage{caption}
\providecommand{\U}[1]{\protect\rule{.1in}{.1in}}
\newtheorem{theorem}{Theorem}

\newtheorem{corollary}[theorem]{Corollary}

\newtheorem{definition}{Definition}

\newtheorem{lemma}[theorem]{Lemma}

\newtheorem{remark}{Remark}

\newenvironment{proof}[1][Proof]{\noindent\textbf{#1.} }{\ \rule{0.5em}{0.5em}}
\numberwithin{theorem}{section}
\numberwithin{proposition}{section}
\numberwithin{definition}{section}
\numberwithin{example}{section}
\renewcommand{\arraystretch}{1.2}
\newcommand{\im}[0]{\mathrm{i}}

\newcommand{\ket}[1]{|#1\rangle}
\newcommand{\braketop}[3]{\ensuremath{\left\langle #1 \right| #2 \left| #3 \right\rangle}}
\newcommand{\matele}[3]{\langle #1 \vert #2 \vert #3 \rangle}
\newcommand{\outerprod}[2]{\vert #1 \rangle\!\langle #2 \vert}

\newcommand{\btheta}{\bm{\theta}}
\newcommand{\len}[0]{\text{len}}

\newcommand{\half}{\frac{1}{2}}
\newcommand{\hw}{$\hbar\omega$}
\newcommand{\sof}[0]{\vert f \vert}
\newcommand{\sofbar}[0]{\vert \bar{f} \vert}
\newcommand{\summ}[1]{\sum_{i=0}^{2^{#1}-1}}
\newcommand{\dopr}[2]{(-1)^{#1 \cdot #2}}

\newcolumntype{S}{>{\centering}p{0.1\textwidth}}
\newcolumntype{s}{>{\centering}p{0.06\textwidth}}
\newcolumntype{M}[1]{>{\centering\arraybackslash}m{#1}}
\newcolumntype{P}[1]{>{\centering\arraybackslash}p{#1}}

\allowdisplaybreaks

\begin{document}

\title[]{Neutron-nucleus dynamics simulations for quantum computers}

\author{Soorya Rethinasamy}
\thanks{Corresponding author}
\email{sooryarethin@gmail.com}
\affiliation{School of Applied and Engineering Physics, Cornell University, Ithaca, New York 14850, USA}
\affiliation{Department of Physics and Astronomy, Louisiana State University, Baton Rouge, Louisiana 70803, USA}
\affiliation{Hearne Institute for Theoretical Physics and Center for Computation and Technology, Louisiana State University, Baton Rouge, Louisiana 70803, USA}

\author{Ethan Guo}
\affiliation{Department of Physics and Astronomy, Louisiana State University, Baton Rouge, Louisiana 70803, USA}
\affiliation{Hearne Institute for Theoretical Physics and Center for Computation and Technology, Louisiana State University, Baton Rouge, Louisiana 70803, USA}

\author{Alexander Wei}
\affiliation{Department of Physics and Astronomy, Louisiana State University, Baton Rouge, Louisiana 70803, USA}
\affiliation{Hearne Institute for Theoretical Physics and Center for Computation and Technology, Louisiana State University, Baton Rouge, Louisiana 70803, USA}

\author{Mark M. Wilde}
\affiliation{School of Electrical and Computer Engineering, Cornell University, Ithaca, New York 14850, USA}
\affiliation{Department of Physics and Astronomy, Louisiana State University, Baton Rouge, Louisiana 70803, USA}
\affiliation{Hearne Institute for Theoretical Physics and Center for Computation and Technology, Louisiana State University, Baton Rouge, Louisiana 70803, USA}

\author{Kristina D. Launey}
\affiliation{Department of Physics and Astronomy, Louisiana State University, Baton Rouge, Louisiana 70803, USA}

\begin{abstract}
With a view toward addressing the explosive growth in the computational demands of nuclear structure and reactions modeling, we develop a novel quantum algorithm for neutron-nucleus simulations with general potentials, which provides acceptable bound-state energies even in the presence of noise, through the noise-resilient training method. In particular, the algorithm can now solve for any band-diagonal to full Hamiltonian matrices, as needed to accommodate a general central potential. While we illustrate the approach for exponential Gaussian-like potentials and \textit{ab initio} inter-cluster potentials (optical potentials), it can also accommodate the complete form of the chiral effective-field-theory nucleon-nucleon potentials used in \textit{ab initio} nuclear calculations. In this study, we provide a comprehensive analysis for the efficacy of this approach for three different qubit encodings, including the one-hot, binary, and Gray encodings, in terms of the number of Pauli strings and commuting sets involved. We also discuss the advantages of the algorithm for Hamiltonians of various band-diagonal widths, especially critical for potentials of perturbative nature, leading to a drastically reduced runtime of quantum simulations. We prove that the Gray encoding allows for an efficient scaling of the model-space size $N$ (or number of basis states used) and is more resource efficient for band-diagonal Hamiltonians having bandwidth up to $N$. We introduce a new commutativity scheme called distance-grouped commutativity (DGC) and compare its performance with the well-known qubit-commutativity (QC) scheme. We lay out the explicit grouping of Pauli strings and the diagonalizing unitary under the DGC scheme, and we prove that it outperforms the QC scheme, at the cost of a more complex diagonalizing unitary. Lastly, we provide first solutions of the neutron-alpha dynamics from quantum simulations suitable for noisy intermediate-scale quantum  processors, using an optical potential rooted in first principles, as well as a study of the bound-state physics in neutron-Carbon systems, along with a comparison of the efficacy of the one-hot and Gray encodings.
\end{abstract}

\date{\today}
\startpage{1}
\endpage{10}
\maketitle

\tableofcontents

\section{Introduction}
The atomic nucleus is a quantum many-body system made of nucleons that are subject to residual strong forces that have no analytical solution. For an $A$-particle system, the nuclear problem needs to be solved numerically in the infinite-dimensional Hilbert space of $A$ particles with Hamiltonians that admit state-of-the-art nucleon-nucleon (NN) forces, often three-nucleon (3N), and even four-nucleon (4N) forces. This leads to the so-called scale explosion problem in nuclear structure calculations, i.e., the explosive growth in computational resource demands with increasing number of particles and size of the spaces in which they reside. Major progress in the development of high-precision inter-nucleon interactions \cite{BedaqueVKolck02,EpelbaumNGKMW02,EntemM03,Epelbaum06,PhysRevLett.115.122301} along with the utilization of high-performance computing resources have tremendously advanced nuclear science explorations. This has placed \textit{ab initio} (or from first principles) large-scale simulations at the frontier of physics, including, for example, accurate theoretical predictions for light muonic atoms \cite{Ji_PRL_2013}, scattering calculations of interest to astrophysics and energy applications \cite{ElhatisariLRE15,HupinQN19, LauneyMD_ARNPS21}, as well as input to high-precision beta-decay measurements that probe physics beyond the standard model \cite{PhysRevLett.128.202503,PhysRevLett.128.202502}. 

The situation is even more complicated when one needs accurate descriptions of nuclear reactions -- the dynamics of several nuclei (reaction fragments) that interact, -- especially in regions of the nuclear chart where  experiments are currently infeasible. A general approach to reactions, especially suitable for heavier nuclear systems, is based on identifying few-body degrees, typically the  reaction fragments (or clusters) involved in the reaction, and reduce the many-body problem to a few-body technique \cite{thompsonn09}. This reduction results in  effective interactions (often referred to as optical potentials) between the clusters. Here again, the demand in classical computational resources grows exponentially with the number of reaction fragments and the range of their interaction.

With a view toward addressing such challenges in the long term by harnessing the advantages of quantum computing -- as demonstrated for various low- and high-energy nuclear physics problems (e.g., see \cite{PhysRevLett.120.210501,PhysRevLett.127.040505,PhysRevA.105.022440,PhysRevC.105.064308,PhysRevA.106.062435,KGL+22,PhysRevLett.130.221003,TurroLLNL23,Watson23,Davoudi23,PerezObiol23,BS24}) -- in this paper, we start with the simplest case of two clusters, one of which is a neutron. We provide, for the first time, solutions of the neutron-nucleus dynamics from quantum simulations suitable for the far-term error-corrected regime as well as for the noisy intermediate-scale quantum (NISQ) processors coupled with the noise-resilient (NR) training method \cite{Sharma_2020}: this is illustrated for the bound-state physics of the neutron-alpha (n-$^4$He) optical potential rooted in first principles \cite{BurrowsL21}, as well as for the lowest $\half^+$ energy in Carbon isotopes calculated through the n+$^{10}$C, n+$^{12}$C, and n+$^{14}$C dynamics. We note that, in distinction to terminology often used in quantum computing, ``dynamics simulation'' or ``simulation of (nuclear) dynamics'' refer here to the problem of modeling the nuclear multi-cluster system using a specific inter-cluster potential. 

The present method utilizes the Variational Quantum Eigensolver (VQE) \cite{Peruzzo2014,Cerezo2021,bharti2021noisy} and is based on the pioneering nuclear simulations of the deuteron on quantum computers, where the potential is given only by a single matrix element~\cite{PhysRevLett.120.210501}. In our study, we design a novel quantum algorithm for a two-cluster system and a very general potential, which allows for versatile applications, including widely used exponential potentials in reaction calculations (see, e.g., Ref.~\cite{DescouvemontB10}) and most importantly, optical potentials derived \textit{ab initio}. We note here that the choice of VQE can be replaced with any other technique, for example, tensor network warm starting, etc.

In this paper, we provide a generalized, extensible, and strong mathematical formulation of three mappings to qubits: one-hot, binary, and Gray encodings (see, e.g.,~\cite{Sawaya20}), explore their relative advantages, and illustrate these for simulations with the above-mentioned potentials. Going beyond the scope of earlier work that explored specific properties of mappings based on simulations only (e.g., see~\cite{SA21}), we provide mathematical proofs of scaling for these three encodings. Furthermore, the techniques that we use to prove these results are applicable for general encodings. Based on this, we show that the Gray encoding (introduced to nuclear calculations in Ref.~\cite{PhysRevA.103.042405}) allows for an efficient scaling of the model-space size $N$ (or number of the basis states used) and is more resource efficient not only for tridiagonal Hamiltonians ($K=1$), as suggested in Ref.~\cite{Sawaya20}, but also for band-diagonal Hamiltonians for $K<N/2$, where $2K+1$ is the bandwidth of the Hamiltonian. Interestingly, we show that for bandwidths larger than $N$, more off-diagonals can be added, if needed for an increased accuracy, without increasing the complexity of the problem. Another  outcome of this study relates to the efficacy of measurements, which is of key importance to obtaining acceptable outcomes on quantum devices. In particular, we introduce a new commutativity scheme called distance-grouped commutativity (DGC), which is especially useful for band-diagonal matrices.  We compare its performance with the well-known qubit-commutativity (QC) scheme. We lay out the explicit grouping of Pauli strings and the diagonalizing unitary under the DGC scheme. We show that the DGC scheme outperforms the QC scheme, at the cost of a more complex diagonalizing unitary. We note here that the diagonalizing unitary turns out to be the $\operatorname{GHZ}$ preparation unitary (See Lemma~\ref{lem:diagUnitaryDGC} for the mathematical definition), and these have been used for  diagonalization in quantum simulations \cite{SCLW22, FCP+23}.

We note that, in this study, the quantum simulations are reported for the lowest bound states of two clusters, for which a manageable number of qubits can be currently used (three or four qubits). Ultimately, the algorithm presented here can underpin multi-cluster dynamics simulations at low energies, such as those relevant to astrophysical studies, and can be utilized for weakly-bound states (e.g., for n+$^{16}$C and n+$^{18}$C) and even for isolated low-lying resonances that require solutions in much larger model spaces (larger $N$)\footnote{Applications of the present quantum algorithm are shown here for the use of harmonic oscillator single-particle basis states, the same square-integrable basis utilized in Refs.~\cite{QuaglioniN09,DreyfussLESBDD20,BurrowsL21,PhysRevLett.128.202503}, while the asymptotics are recovered in an R-matrix technique \cite{thompsonn09}. Alternatively, one can use the quantum algorithm for a square-nonintegrable basis, such as the eigenstates of the Woods-Saxon potential as utilized, e.g., in Ref.~\cite{mercennemp19}.}. Knowledge about the bound-state physics is key, e.g., to the description of deuteron break-up reactions, such as (d,p) and (d,n) reactions, for which standard distorted-wave Born approximation methods rely on the physics of the bound states of the proton-nucleus and neutron-nucleus systems. As another important implication, exploring trade-offs of band-diagonal and full Hamiltonian matrices for different encodings is critical for the simplest two-cluster system of two nucleons; namely, this allows quantum simulations of nuclear structure to handle the complete form of the chiral nucleon-nucleon (NN) potentials (e.g., see \cite{EntemM03,Epelbaum06,PhysRevLett.115.122301}), which are in turn key to \textit{ab initio} large-scale simulations of light, medium-mass, and even selected heavy nuclei.

We provide solutions for various n$+$C systems using an exponential potential, and for n$+\alpha$ based on the \textit{ab initio}  optical potential derived in Ref.~\cite{BurrowsL21}. We find that the quantum simulations are successful in finding the lowest $\frac{1}{2}^+$ bound-state energies. We develop a warm-start algorithm, inspired by perturbation theory, that is particularly useful for band-diagonal Hamiltonians. In this method, we first simulate the system for a simpler, leading-order, Hamiltonian, and use the endpoint of the simulation as the start for the full-scale simulation. We find that this method allows for a quicker convergence to the true energy value.

Our paper is intended to serve interdisciplinary research at the intersection of nuclear physics and quantum information science, and in some cases includes details that may be well known in one of the fields but are pedagogical for researchers of the other field: our aim here is to provide a complete framework for the problem at hand. Our paper is organized as follows. In Sec.~\ref{sec:ProbDesc}, we introduce the nuclear problem of solving the neutron-nucleus dynamics and its Hamiltonian. In Sec.~\ref{sec:encodings}, we discuss different encoding methods of mapping the given Hamiltonian to a form that can be simulated on a quantum computer. The different encoding methods we discuss include the one-hot encoding, binary encoding, and the Gray encoding. In Sec.~\ref{sec:VQE}, we briefly explain the variational principle, for completeness of presentation. Since the variational principle depends on the choice of a trial state, called ansatz, we delineate the different ansatz choices for the different encodings. In Sec.~\ref{sec:tradeoffs}, we analyze various advantages between the different encodings considered, including the number of Pauli terms and the number of commuting sets, for a most general local potential and its band-diagonal approximation. As part of our trade-off analysis for the different encodings of Hamiltonians, we introduce the new DGC measurement scheme to group Pauli strings into commuting operator sets. For this scheme we provide the explicit diagonalizing unitary and an analysis of the number of commuting sets. In Sec.~\ref{sec:discussions}, we provide quantum simulations for different encodings and nuclear systems, including comparisons of the different commuting sets. We discuss the results and challenges for weakly bound states.

\section{Problem description}
\label{sec:ProbDesc}

A many-body ``configuration interaction" (CI) method (often called the shell model in nuclear physics~\cite{BrussardG77,Shavitt98,BarrettNV13}) solves the many-body Schr\"odinger equation for $A$ particles:
\begin{equation}
H \Psi(\vec r_1, \vec r_2, \ldots, \vec r_A) = E \Psi(\vec r_1, \vec r_2, \ldots, \vec r_A),
\label{ShrEqn}
\end{equation}
for which the interaction and basis configurations are adopted as follows.
The intrinsic non-relativistic nuclear and Coulomb interaction Hamiltonian is defined as
\begin{equation}
H = T_{\rm rel} + V_{\rm NN}  + V_{\rm 3N} + \cdots + V_{\rm Coulomb}, 
\label{intH}
\end{equation}
where $T_{\rm rel} =\frac{1}{A}\sum_{i<j}\frac{(\vec p_i - \vec p_j)^2}{2m_{\rm N}}$ is the relative kinetic energy ($m_{\rm N}$ is the nucleon mass), $V_{\rm NN}=\sum_{i<j}^A (V_{\rm NN})_{ij}$  is the nucleon-nucleon (NN) interaction  (and possibly, $V_{\rm 3N}=\sum_{i<j<k}^A (V_{\rm NNN})_{ijk}$, $V_{\rm 4N}$, $\ldots$ interactions), and $V_{\rm Coulomb}$ is the Coulomb interaction between the protons.  The Hamiltonian may also include higher-order electromagnetic interactions, such as magnetic dipole-dipole terms.  

A complete orthonormal basis $\{\psi_i\}_i$ is adopted, such that the expansion  $\Psi(\vec r_1, \vec r_2, \ldots, \vec r_A)$ in terms of  unknown coefficients $c_k$,
$\Psi(\vec r_1, \vec r_2, \ldots, \vec r_A) = \sum_{k} c_k \psi_k(\vec r_1, \vec r_2, \ldots, \vec r_A)$,
 renders Eq.~\eqref{ShrEqn} into a matrix eigenvalue equation:
\begin{equation}
\label{eq:ham_eig_eqn}
\sum_{k'} H_{k k'} c_{k'} = E c_k,
\end{equation}
where the many-particle Hamiltonian matrix elements 
$H_{k k'} = \langle \psi_k | H | \psi_{k'} \rangle$ are in general complex and are calculated for the given interaction Eq.~\eqref{intH}. Typically, the basis is a finite set of antisymmetrized products of  single-particle states (Slater determinants), referred to as a ``model space". In this study, we use the single-particle states  of a three-dimensional spherical harmonic oscillator (HO), 
$\phi_{n_r (\ell \half)jmt_z}(\vec r)$, where $n_r$ is the radial quantum number, the orbital angular momentum $\ell$ and spin-$\frac{1}{2}$ are coupled to the total angular momentum $j$, and $t_z$ distinguishes between protons and neutrons (we use the convention of HO wavefunctions that are positive at infinity). Such a basis allows for preservation of translational invariance of the nuclear self-bound system and provides solutions in terms of single-particle wave functions that are analytically known. With larger model spaces utilized in the  shell-model theory, the  eigensolutions converge to the exact ones.

To describe the neutron-nucleus (NA) dynamics, e.g., for n-$\alpha$, where $\alpha$ (with $A=4$ particles) is in its ground state $\ket{\Psi^{(A=4)}_0}$ with energy $E^{(4)}_0$, one can deduce an effective non-local interaction $\tilde V(r,r')$ between the neutron and the four-body system using the Green's function approach \cite{BurrowsL21}, where $r$ (and $r'$) is the relative distance between the two clusters before (and after) scattering.  This is based on solutions of the five-body system, that is, all states $\ket{\Psi^{(5)}_k}$ with their energy $E^{(5)}_k$, along with their single-particle overlaps $u_k(\vec r) = \braketop{\Psi^{(5)}_{k}}{a^\dagger_{\vec r}}{\Psi^{(4)}_{0}}$, where $a^\dagger_{\vec r}$ creates a single particle at distance ${\vec r}$. This $\tilde V(r,r')$ potential, which can be readily derived in the \textit{ab initio} framework (see Ref.~\cite{BurrowsL21} for n-$\alpha$), can be rendered into an equivalent local form $V(r)$ according to Eq.~(31) of Ref.~\cite{PhysRevC.95.024315}, 
\begin{equation}
    V(r) u(r) = \int dr'\, r'^2\, \tilde V(r,r')\, u(r'),
    \label{eq:locV}
\end{equation}
where $u(r)$ is in units of fm$^{-3/2}$ and $V(r)$ is in units of MeV. We note that for bound states and resonances for which only the elastic channel is open, the potentials are real, except on the poles, which can be numerically avoided through the principal value theorem as shown in Ref.~\cite{BurrowsL21}; while in this study, we do not need the imaginary part of the $\tilde V(r,r')$ optical potential for the bound-state simulations at hand, the generalization to complex matrices is straightforward and feasible.
The potential energy $V(r)$ enters into the two-body Schr\"odinger equation, as described next.

Let $\vec{r}_1$ and $\vec{p}_1$ be the position and momentum vector of the nucleus, and let $\vec{r}_2$ and $\vec{p}_2$ be the position and momentum vector of the neutron (or any reaction fragment). Thus, the Hamiltonian of the $A+1$ nuclear system is given by
\begin{equation}
    \tilde{H} = \frac{p_1^2}{2m_1} + \frac{p_2^2}{2m_2} + V(|\vec{r_1} - \vec{r_2}|),
\end{equation}
where $m_1$ is the mass of the nucleus and $m_2$ is the mass of the neutron. Since we are interested in the relative motion of the projectile (the neutron) relative to the target (nucleus), we make a transformation to the center-of-mass coordinate, $\vec{R} = \frac{m_1 \vec{r_1} + m_2 \vec{r_2}}{m_1 + m_2}$, and relative coordinate,  $
    \vec{r} = \vec{r_1} - \vec{r_2}$: 
\begin{equation}
    \tilde{H} = \frac{P^2}{2M} + \frac{p^2}{2\mu} + V(r), {\,\rm with }\,H=\frac{p^2}{2\mu} + V(r),
\end{equation}
where $M = m_1 + m_2$ is the total mass, $\mu = \frac{m_1 m_2}{m_1 + m_2}$ is the reduced mass (usually reported in terms of the nucleon mass $m_{\rm N}$ and the mass numbers of the target $A$ and projectile $a$  as $\mu = \frac{A a}{A + a}m_{\rm N}$), and we use that 
$    \vec{P} = \vec{p_1} + \vec{p_2}$ and $   \vec{p}= \frac{m_2\vec{p_1} - m_1\vec{p_2}}{m_1 + m_2}$.
The term $\frac{P^2}{2M}$ represents the kinetic energy of the centre of mass. Since $[P, p] = 0$ and $[P, r] = 0$, this term can be dealt with independently. Thus, in the center-of-mass reference frame, we need to solve Eq.~\eqref{eq:ham_eig_eqn}, where $T = \frac{p^2}{2\mu}$ is the relative kinetic energy and $V(r)$ is the potential energy or the effective neutron-nucleus interaction. We focus on the $^2S_{\frac{1}{2}}$ partial wave (following the notation $^{2s+1}\ell_J$): $\alpha$ (or an even-even Carbon isotope) is in a $0^+$ ground state, the relative orbital angular momentum is $\ell=0$, the spin of the neutron is $s=1/2$, yielding total angular momentum $J=1/2$ (and since the projectile is a neutron, there is no Coulomb interaction). For this channel (and any positive-parity channel), the most general central potential can be expressed as 
\begin{equation}
    \label{eq:Vgen}
    V(r) =\sum_{k=0}^\infty v_k r^{2k},
\end{equation}
where the coefficients $v_k$ are taken to be real for all $k$ in this work, since we are interested in the bound-state physics. 
To represent this Hamiltonian on a quantum computer, we use a discrete-variable representation in the harmonic-oscillator basis, for which the radial wave functions are known analytically. In this basis, the Hamiltonian is an infinite-dimensional matrix. However, in order to perform simulations, we truncate the matrix to an $N \times N$ matrix, called $H_N$ (retaining the notations of Ref.~\cite{PhysRevLett.120.210501}). As the size $N$ of the matrix  increases, the approximation to the true Hamiltonian becomes more accurate. Expanding in this basis,
\begin{equation}
    \label{eq:H_matrix_rep}
    H_N \coloneqq \sum\limits_{n_r, n_r'=0}^{N-1} \langle n_r' \vert T+V \vert n_r \rangle \outerprod{n_r'}{n_r},
\end{equation}
where $n_r$ denotes a relative harmonic oscillator radial node number. The matrix elements of $T$ for every $\ell$ are given by
\begin{multline}
    \matele{n_r' \ell}{T}{n_r \ell} \coloneqq \frac{\hbar \omega}{2} \Bigg[ \bigg(2n_r + \ell + \frac{3}{2}\bigg) \delta_{n_r',n_r} \\
    \qquad - \sqrt{n_r\bigg(n_r+ \ell + \frac{1}{2}\bigg)}\delta_{n_r',n_r-1} \\
    \qquad - \sqrt{\bigg(n_r+1\bigg)\bigg(n_r+ \ell + \frac{3}{2}\bigg)}\delta_{n_r',n_r+1} \Bigg],
\end{multline}
where in Ref.~\cite{PhysRevLett.120.210501} and in this paper we use $\ell=0$ (hence, we will omit the $\ell$ notation henceforth).
The matrix representation of $T$ is tridiagonal; i.e., the only non-zero elements are the main diagonal and one diagonal each above and below it. 

The matrix elements of $V$ are given as
\begin{equation}
    \matele{n_r'}{V}{n_r} = \sum\limits_{k=0}^\infty v_k \matele{n_r'}{r^{2k}}{n_r}.
\end{equation}
To calculate the matrix elements of $r^{2k}$, we use a recursive approach,
\begin{equation}
    \matele{n_r'}{r^{2k}}{n_r} = \sum\limits_{n_r'' = 0}^\infty \matele{n_r'}{r^{2k-2}}{n_r''} \matele{n_r''}{r^2}{n_r},
\end{equation}
with the base case for every $\ell$ being 
\begin{multline}
    \matele{n_r'\ell}{r^2}{n_r \ell} \coloneqq b_s^2 \Bigg[ \bigg(2n_r + \ell+ \frac{3}{2}\bigg)\delta_{n_r',n_r} \\
    \qquad + \sqrt{n_r\bigg(n_r+ \ell + \frac{1}{2}\bigg)}\delta_{n_r',n_r-1} \\
    \qquad + \sqrt{\bigg(n_r+1\bigg)\bigg(n_r+\ell +\frac{3}{2}\bigg)}\delta_{n_r',n_r+1}\Bigg],
\end{multline}
where the oscillator length is defined in terms of the reduced mass and \hw. Also,  $b_s\coloneqq \sqrt{ \frac{\hbar }{\mu \omega} }$ (and $\ell=0$ is used in this study).

A schematic of the matrices of this problem description is given in Fig.~\ref{fig:TruncatedPotential}b, as compared to the one used in Ref.~\cite{PhysRevLett.120.210501} with a single matrix element (Fig.~\ref{fig:TruncatedPotential}a).

\begin{figure}
    \includegraphics[width=0.8\columnwidth]{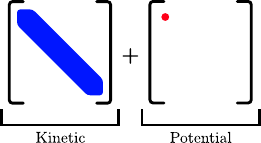}\\ (a)
\\
    \includegraphics[width=0.8\columnwidth]{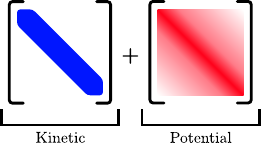}\\ (b)
\\
    \includegraphics[width=0.8\columnwidth]{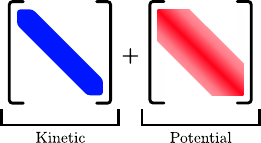}\\ (c)
\\
    \caption{The kinetic and potential energy matrices: (a) for the contact potential used in Refs.~\cite{PhysRevLett.120.210501,PhysRevA.103.042405}, (b) for the complete potential $V(r)$, and (c) for the truncated potential $V_K(r)=\sum_{k=0}^K v_k r^{2k}$  used in this work.
    }
    \label{fig:TruncatedPotential}
\end{figure}

In many cases, it is advantageous to approximate, to a very good degree, the neutron-nucleus potential by an exponential form:
\begin{equation}
    V(r) \approx V_{\rm E}(r) = V_0\exp\!\left(-c(r/b_s)^2 \right).
    \label{eq:Vexp}
\end{equation}
Expanding the potential as a Taylor series around $r=0$ leads to
\begin{align}
    V_{\rm E}(r) &=  V_0 \sum\limits_{k=0}^\infty \frac{(-1)^k c^k}{k!} \left( \frac{r}{b_s} \right)^{2k}. 
\end{align}

Finally, for the quantum simulations, we set an upper truncation parameter $K$ for the number of terms in the expansion. Therefore, the potential is given by
\begin{equation}
    \matele{n_r'}{V_K}{n_r} = \sum\limits_{k=0}^K v_k \matele{n_r'}{r^{2k}}{n_r},
    \label{eq:Vcut}
\end{equation}
with $v_k=V_0 (-1)^k c^k/ (k!b_s^{2k})$ in the case of the exponential approximation $V_{\rm E}(r)$ of Eq.~\eqref{eq:Vexp}.
We note that the matrix representation of $r^2$ is tridiagonal and in general, the matrix representation of $r^{2K}$ is $(2K+1)$-diagonal. 
A schematic of the truncated matrices of this problem description suitable for quantum simulations is given in Fig.~\ref{fig:TruncatedPotential}c.

The quantum computational simulations of nuclear systems of Ref.~\cite{PhysRevLett.120.210501,PhysRevA.103.042405} have used a contact potential for $\ell=0$, or $\braketop{n_r'}{V}{n_r}=V_0\delta_{n_r,n_r'}\delta_{n_r,0}$ (Fig.~\ref{fig:TruncatedPotential}a). In this work, the Hamiltonian is generalized for  
any band-diagonal to full matrix and can now accommodate a general central potential, such as  exponential Gaussian-like potentials using Eq.~\eqref{eq:Vexp} (e.g., \cite{DescouvemontB10}), along with \textit{ab initio} inter-cluster potentials (as those derived in Ref.~\cite{BurrowsL21}) and the central part of chiral NN potentials using Eq.~\eqref{eq:Vgen}. In general, chiral NN potentials, such as in Refs.~\cite{EntemM03,Epelbaum06,PhysRevLett.115.122301}, used in \textit{ab initio} nuclear calculations require, in addition, spin and isospin degrees of freedom $\alpha=n_r(\ell \half)jmt_z$, as described above. Including the additional spin-isospin quantum numbers leads to a larger set of basis states, but nonetheless, in the present framework this is straightforward by generalizing Eq.~\eqref{eq:H_matrix_rep} to 
\begin{equation}
    H = \sum\limits_{\alpha, \alpha'} \langle \alpha' \vert T+V \vert \alpha \rangle \outerprod{\alpha'}{\alpha},  
\end{equation}
where $\langle \alpha' \vert V \vert \alpha \rangle$ are known matrix elements for any chiral NN potential and the enumerating index $\alpha$ replaces $n_r$ in the mappings to qubits discussed next. Importantly, physically relevant potentials with a comparatively soft core or widely used potentials renormalized using, e.g., the Similarity Renormalization Group (SRG) technique~\cite{BognerFP07}, are band-diagonal as a result of the decoupling of low- and high-momentum configurations. Hence, the advantages found in this study as a function of the bandwidth $2K+1$ directly generalize to the complete form of the chiral NN potentials and their routinely used band-diagonal structure.

\section{Mapping onto a quantum computer
\label{sec:encodings}}

In the current form, the Hamiltonian from Eq.~\eqref{eq:H_matrix_rep} is given in terms of number operators of the form $\outerprod{m}{m}$ and step-$i$ ladder operators of the form $\outerprod{m}{m-i}$ and $\outerprod{m}{m+i}$ for $i \in \{1, \ldots, K\}$ (using the notation of Ref.~\cite{PhysRevA.103.042405}). To find the bound-state energy of this Hamiltonian, we first need to map these operators into Pauli strings. 

Mappings from operators in the Fock space to Pauli strings are called \textit{encodings}. Multiple encodings can be found in existing literature, for example, the one-hot encoding (OH), the Bravyi-Kitaev encoding, the Verstraete-Cirac encoding, the binary encoding, and the Gray encoding (see, e.g.,~\cite{Sawaya20,PhysRevA.103.042405,VC05}. In addition, there has been substantial work on using qudits (higher dimensional systems) rather than qubits, which has been shown to provide some advantage in specific scenarios \cite{IRS23, VEN24, BRSK+21}. In this work, we focus on three qubit encodings - OH, binary, and Gray - and analyze various properties and trade-offs. 

To provide an illustrative example, for each encoding above, we show the explicit Pauli terms for the $N=4, K=2$ Hamiltonian with the exponential potential in Eq.~\eqref{eq:Vexp} with parameters
\begin{equation}
       V_0 = -2.79\, \textrm{MeV},\, 
       c = 0.05 ,\,
    \hbar \omega = 15.95\, \textrm{MeV},
    \label{eq:eg}
\end{equation}
for the case of n+$^{16}$C with reduced mass $\mu=\frac{16}{17}m_{\textrm{N}}$ (here we use $m_{\textrm{N}}=938.272029$ MeV for both protons and neutrons). The Hamiltonians for the different encodings can be found in their respective sections below. 

We note that the development for the ladder operators and step-$1$ operators follows directly from \cite{PhysRevA.103.042405}. We generalize these results to step-$k$ operators, for $k>1$, in this section. 

\subsubsection{One-hot encoding}

For fixed $N$, the one-hot encoding maps the number and ladder operators to Pauli strings on $N$ qubits. The Fock basis states are mapped as follows:
\begin{equation}
    \ket{m} \rightarrow \ket{q_0 q_1 \ldots q_{N-1}},
\end{equation}
for $m \in \{0, \ldots, N-1\}$, where $q_m = 1$ and all other bits are zero. For example, for $N=4$, the states are mapped as
\begin{eqnarray}
 \begin{tabular}{ c c c }
\label{tab:OH-basis-states}
    $\ket{0}$ & $\to$ & $\ket{1000}$, \\ 
    $\ket{1}$ & $\to$ & $\ket{0100}$, \\ 
    $\ket{2}$ & $\to$ & $\ket{0010}$, \\ 
    $\ket{3}$ & $\to$ & $\ket{0001}$.
\end{tabular}  
\end{eqnarray}

Next, we define the number operators and ladder operators. The number operator $\outerprod{m}{m}$ maps $\ket{m}$ to itself and maps all other basis states to $0$:
\begin{equation}
    \left(\outerprod{m}{m}\right) \ket{m'} = \delta_{m, m'} \ket{m}.
\end{equation}
Thus, in the encoded basis, the number operators are mapped to
\begin{equation}
    \outerprod{m}{m} \rightarrow \outerprod{1}{1}_{m} \coloneqq \frac{1}{2}(I_{m} - Z_{m}),
\end{equation}
where the subscript $m$ indicates the qubit on which the operator acts, that is, e.g., $Z_2=I \otimes I \otimes Z \otimes I \otimes \dots \otimes I$. 

Next, the action of the ladder operator $\outerprod{m}{m \pm i}$ on the number states is as follows:
\begin{equation}
    \left(\outerprod{m}{m \pm i} \right) \ket{m'} = \delta_{m \pm i, m'}\ket{m}.
\end{equation}
In this encoded basis, this action involves flipping the $(m \pm i)$ qubit to $0$ and the $m$ qubit to 1:
\begin{align}
    &\outerprod{m}{m \pm i} \rightarrow \outerprod{1}{0}_{m} \otimes \outerprod{0}{1}_{m \pm i} \nonumber \\
    &= \frac{1}{2}(X_{m} -\mathrm{i} Y_{m}) \otimes \frac{1}{2}(X_{m \pm i} + \im Y_{m \pm i}) \nonumber \\
    &= \frac{1}{4}(X_{m} X_{m \pm i} + \im X_{m} Y_{m \pm i} - \im Y_{m}X_{m \pm i} + Y_{m}Y_{m \pm i}).
\end{align}
We note that these ladder operators always occur in pairs due to the Hermiticity of the Hamiltonian. Thus, 
\begin{equation}
    \outerprod{m}{m \pm i} + \outerprod{m \pm i}{m} = \frac{1}{2}(X_{m} X_{m \pm i} + Y_{m}Y_{m \pm i}).
\end{equation}
The rest of the operators can be constructed similarly. The list of all operators for $N=4$ can be found in Appendix~\ref{app:list_operators}. 

Substituting the above defined number and ladder operators into Eq.~\eqref{eq:H_matrix_rep}, we find the representation of $H_{N, K}$ in the one-hot encoding to be
\begin{multline}
    \label{eq:HN_OneHot}
    H_{N, K} = \frac{1}{2} \sum\limits_{m=0}^{N-1} \langle m \vert H \vert m \rangle (I_{m} - Z_{m})\\
    + \frac{1}{2} \sum\limits_{m=0}^{N-1} \sum\limits_{k=1}^{\max(K, 1)} \langle m+k \vert H \vert m \rangle (X_{m} X_{m+k} + Y_{m} Y_{m+k}).
\end{multline}
The upper limit of the $k$ summation is $\max(K, 1)$ due to the fact that even if $K = 0$, the first off-diagonal term in the Hamiltonian is non-zero as a result of the kinetic energy having two off-diagonal terms. Note that, to preserve the size of the matrix to be $N$, the summation over $k$ terminates if $m + k > N - 1$. Thus, in the encoded one-hot basis, if $K>1$, the resulting matrix is $(2K+1)$-diagonal.  

For the example under consideration in Eq.~\eqref{eq:eg}, the encoded Hamiltonian is given by
\begin{multline}
\label{eq:example_OH_enc}
    H_{4, 2} = 67.117\ IIII - 4.674\ IIIZ - 12.751\ IIZI \\
    - 20.812\ IZII - 28.880\ ZIII -4.814\ IIXX \\ 
    - 4.814\ IIYY -8.801\ IXXI - 8.801\ IYYI \\
    -12.772\ XXII - 12.772\ YYII -0.004\ IXIX \\
    - 0.004\ IYIY -0.014\ XIXI - 0.014\ YIYI,
\end{multline}
where we have truncated the coefficients to three decimal places.

\begin{remark}
    The Jordan--Wigner transformation is a mapping from fermionic operators to Pauli operators of the form:
    \begin{align}
        a^\dagger_{m} &\rightarrow \frac{1}{2} \left[ \prod_{j=0}^{m-1} Z_j \right] (X_{m} - \im Y_{m}) \\
        a_{m} &\rightarrow \frac{1}{2} \left[ \prod_{j=0}^{{m}-1} Z_j \right] (X_{m} + \im Y_{m}).
    \end{align}

    We note that the Jordan--Wigner transformation and the one-hot encoding do not map a fermionic operator to the same Pauli string. The Jordan--Wigner transformation results in the following mapping:
    \begin{equation}
        a^\dagger_{m} a_{{m}+i} + a^\dagger_{{m}+i} a_{m}  
        \rightarrow \frac{1}{2} \left( X_{m} \overline{Z} X_{{m}+i} + Y_{m} \overline{Z} Y_{{m}+i} \right),
        \label{eq:JW-map}
    \end{equation}
    where $\overline{Z} \equiv Z_{{m}+1} \otimes \cdots \otimes Z_{{m}+i-1}$. On the other hand, the one-hot encoding results in the map:
    \begin{equation}
        \outerprod{{m}}{{m}+i} + \outerprod{{m}+i}{{m}}  
        \rightarrow \frac{1}{2} \left( X_{m} X_{{m}+i} + Y_{m} Y_{{m}+i} \right).
        \label{eq:OH-map}
    \end{equation}

    While the operators in Eqs.~\eqref{eq:JW-map} and \eqref{eq:OH-map} in general act differently, their action is identical on the set of encoded basis states given in Eq.~\eqref{tab:OH-basis-states}. Thus, if we restrict the states to be superpositions of the encoded basis states only, these operators can be used interchangeably.
\end{remark}

\subsubsection{Binary encoding}

For fixed $N$, the binary encoding maps the number and ladder operators to Pauli strings on $n = \lceil \log_2(N) \rceil$ qubits. For simplicity, we restrict $N$ to be a power of two, in which case the encoded basis set is exactly of size $n = \operatorname{log}_2(N)$. The Fock basis states are mapped as follows:
\begin{equation}
    \ket{m} \rightarrow \ket{q_0 q_1 \ldots q_{n-1}},
\end{equation}
for $m \in \{0, \ldots, N-1\}$, where the bitstring $q = q_0 q_1 \ldots q_{n-1}$ is the binary representation of $m$, denoted by $b_{m}$, on $n$ qubits. For the case of $N=8$, the encoded basis consists of three qubits and is given by
\begin{center}
\begin{tabular}{ c c c }
    $\ket{0}$ & $\to$ & $\ket{000}$, \\ 
    $\ket{1}$ & $\to$ & $\ket{001}$, \\ 
    $\ket{2}$ & $\to$ & $\ket{010}$, \\ 
    $\ket{3}$ & $\to$ & $\ket{011}$, \\
    $\ket{4}$ & $\to$ & $\ket{100}$, \\ 
    $\ket{5}$ & $\to$ & $\ket{101}$, \\ 
    $\ket{6}$ & $\to$ & $\ket{110}$, \\ 
    $\ket{7}$ & $\to$ & $\ket{111}$.
\end{tabular}
\end{center}

Thus, the binary basis $\mathcal{B}_n$, on $n = \operatorname{log}_2(N)$ qubits, is a list of $N$ basis elements:
\begin{equation}
    \mathcal{B}_n = (b_0, b_1, \ldots, b_{2^n - 1}),
\end{equation}
where $b_{m}$ is the binary representation of the integer $m$. 

Next, let us consider how the number and ladder operators map. To this end, let us establish some notation for the following  operators:
\begin{equation}
    \outerprod{m+k}{m} \to B^k_{m}.
\end{equation}
The number operators $\outerprod{m}{m}$ map to $B^0_{m}$, which are defined as
\begin{align}
    \label{eq:BE-B0}
    \outerprod{m}{m} \to B^0_{m} &\coloneqq \outerprod{b_{m}}{b_{m}} \notag \\
    &= P^{b_{m, 0}}_0 \otimes P^{b_{m, 1}}_1 \otimes \cdots \otimes P^{b_{m, n-1}}_{n-1} \notag \\
    &= \bigotimes\limits_{i=0}^{n-1} P^{b_{m, i}}_i,
\end{align}
where $b_{m, i}$ denotes bit $i$ of $b_{m}$ and the operators $P^0$ and $P^1$ are defined as 
\begin{equation}
\begin{aligned}
    \label{eq:Proj01}
    P^0 &\coloneqq  \outerprod{0}{0} = \frac{1}{2} (I + Z) ,  \\
    P^1 & \coloneqq \outerprod{1}{1} = \frac{1}{2} (I - Z).
\end{aligned}
\end{equation}
For example, the number operator $\outerprod{6}{6}$ is mapped to 
\begin{align}
    \outerprod{6}{6} \to B^0_6 & = P^1_0 \otimes P^1_1 \otimes P^0_2 \nonumber \\
    &= \outerprod{1}{1}_0 \otimes \outerprod{1}{1}_1 \otimes \outerprod{0}{0}_2.
\end{align}

In a similar fashion, the step-$1$ ladder operators are mapped as follows:
\begin{equation}
    \outerprod{m+1}{m} \to B^1_m \coloneqq \bigotimes\limits_{i=0}^{n-1} \outerprod{b_{m+1, i}}{b_{m, i}}.
\end{equation}
For example, the ladder operator
\begin{equation}
    \outerprod{4}{3} \to B^1_3 = \outerprod{1}{0}_0 \otimes \outerprod{0}{1}_1 \otimes \outerprod{0}{1}_2,
\end{equation}
as $\ket{3}$ and $\ket{4}$ are mapped to $\ket{011}$ and $\ket{100}$, respectively. 

 Step-$k$ ladder operators for $k>1$ are defined recursively in terms of step-$1$ ladder operators:
\begin{equation}
    \outerprod{m+k}{m} \to B^k_m \coloneqq B^1_{m+k-1} B^{k-1}_m.
\end{equation}
The list of all operators for $N=4$ can be found in Appendix~\ref{app:list_operators}.

Thus, substituting for the number and ladder operators in Eq.~\eqref{eq:H_matrix_rep}, we find the representation of $H_{N, K}$ in the binary encoding to be
\begin{multline}
\label{eq:HN_Binary}
    H_{N, K} = \sum\limits_{m=0}^{N-1} \langle m \vert H \vert m \rangle B^0_m \\
    + \sum\limits_{m=0}^{N-1} \sum\limits_{k=1}^{\max(1, K)} \langle m+k \vert H \vert m \rangle (B^k_m + (B^k_m)^\dagger).
\end{multline}
Note that the summation over $i$ terminates if $m+k > N - 1$. 

For the example under consideration Eq.~\eqref{eq:eg},  the encoded Hamiltonian is given by
\begin{multline}
\label{eq:example_bin_enc}
    H_{4, 2} = 33.556\ II - 8.073\ ZI - 16.134\ IZ 
    - 0.004\ ZZ \\
    - 0.014\ IX + 7.959\ XZ - 0.006 ZX\\
    - 17.586\ XI - 8.801\ XX - 8.801\ YY.
\end{multline}

\subsubsection{Gray encoding}
For a fixed $N$, the Gray encoding maps the number and ladder operators to Pauli strings on $n = \lceil \log_2(N) \rceil$ qubits. For simplicity, we restrict $N$ to be a power of two, in which case the encoded basis set is exactly of size~$n = \operatorname{log}_2(N)$. We first define the Gray basis on $n$ bits, $\mathcal{G}_n$, as a list of $2^n$ basis elements:
\begin{equation}
    \mathcal{G}_n = (g_0, g_1, \ldots, g_{2^n - 1}),
\end{equation}
where each $g_i = (g_{i,0}, g_{i,1}, \ldots g_{i,L-1})$ is a bitstring of length~$n$. The only characteristic of a Gray encoding is that each bitstring entry $g_i$ differs from its neighbor at a single bit. Thus, for a given $n$, there are multiple possible Gray codes. In this work, drawing inspiration from Ref.~\cite{PhysRevA.103.042405}, we use a binary reflective Gray code on $n$ bits. Such a code is defined recursively as follows:
\begin{equation}
    \mathcal{G}_n = (\mathcal{G}_{n-1} \cdot 0, \overline{\mathcal{G}_{n-1}} \cdot 1),
\end{equation}
where $\overline{\mathcal{G}_n}$ is the Gray code on $n$ bits with the entries in reverse order, and $X \cdot y$ is the list of entries of $X$ with $y$ appended at the end. For example, given that $\mathcal{G}_2 = (00, 10, 11, 01)$, we can construct $\mathcal{G}_3$ as follows:
\begin{align}
    \mathcal{G}_3 &= (\mathcal{G}_2 \cdot 0, \overline{\mathcal{G}_2} \cdot 1) \notag \\
    &= ((00, 10, 11, 01) \cdot 0, \overline{(00, 10, 11, 01)} \cdot 1) \notag \\
    &= (000, 100, 110, 010, 011, 111, 101, 001).
\end{align}

Thus, for a fixed $N$, the Fock basis states are mapped to the corresponding entry in a Gray basis $\mathcal{G}_n$, where $n = \operatorname{log}_2(N)$:
\begin{equation}
    \ket{m} \rightarrow \ket{q_0 q_1 \ldots q_{n-1}},
\end{equation}
for $m \in \{0, \ldots, N-1\}$, where the bitstring $q = q_0 q_1 \cdots q_{n-1}$ is the $m$th entry in the Gray basis $g_m$. For example, for $N=8$, the encoded basis is made of three qubits:
\begin{center}
\begin{tabular}{ c c c }
    $\ket{0}$ & $\to$ & $\ket{000}$, \\ 
    $\ket{1}$ & $\to$ & $\ket{100}$, \\ 
    $\ket{2}$ & $\to$ & $\ket{110}$, \\ 
    $\ket{3}$ & $\to$ & $\ket{010}$, \\ 
    $\ket{4}$ & $\to$ & $\ket{011}$, \\ 
    $\ket{5}$ & $\to$ & $\ket{111}$, \\ 
    $\ket{6}$ & $\to$ & $\ket{101}$, \\ 
    $\ket{7}$ & $\to$ & $\ket{001}$. \\ 
\end{tabular}
\end{center}

Next, let us consider how the number and ladder operators map. To this end, let us define the following operators:
\begin{equation}
    \outerprod{m+k}{m} \to G^k_m.
\end{equation}
The number operator $\outerprod{m}{m}$ is mapped to $G^0_m$, which is defined as
\begin{equation}
    \label{eq:GE-L0}
    \outerprod{m}{m} \to G^0_m \coloneqq \bigotimes\limits_{i=0}^{n-1} P^{g_{m, i}}_i.
\end{equation}
To define the step-1 ladder operators, we use a similar construction as in the previous section:
\begin{equation}
    \outerprod{m+1}{m} \rightarrow G^1_m \coloneqq \bigotimes\limits_{i=0}^{n-1} \outerprod{g_{m+1, i}}{g_{m, i}}.
\end{equation}
For example, the ladder operator connecting the basis elements $\ket{2}$ and $\ket{3}$ is given by
\begin{equation}
    \outerprod{3}{2} \to G^1_2 = \outerprod{0}{1}_0 \otimes P^1_1 \otimes P^0_2,
\end{equation}
since the $\ket{3}$ and $\ket{2}$ basis elements differ on the first bit and the other two bits are in the state $\ket{10}$. However, since the Hamiltonian is Hermitian, both the ladder-1 up and down operators are scaled with the same coefficient. Thus, 
\begin{equation}
    G^1_2 + (G^1_2)^\dagger = X_0 \otimes P^1_1 \otimes P^0_2.
\end{equation}

We then go on to define the step-$i$ ladder operators recursively with the step-1 ladder operators being the base case:
\begin{align}
    &\outerprod{m+i}{m} \to G^i_m \nonumber \\
    &\coloneqq (\outerprod{m+i}{m+i-1})(\outerprod{m+i-1}{m}), 
\end{align}
where the ladder operators in the first parenthesis are defined as above. Next, we add an extra term that leads to the recursive definition needed.
For example,
\begin{align}
    \outerprod{4}{1} &= (\outerprod{4}{3})(\outerprod{3}{1}) \nonumber \\
    &= (\outerprod{4}{3} + \outerprod{3}{4})(\outerprod{3}{1}) \nonumber \\
    &= (P^0_0 \otimes P^1_1 \otimes X_2)(\outerprod{3}{1}),
\end{align}
where $\outerprod{3}{1}$ is defined similarly. The list of all operators for $N=4$ can be found in Appendix~\ref{app:list_operators}. 

Thus, substituting for the number and ladder operators in Eq.~\eqref{eq:H_matrix_rep}, we find the representation of $H_N$ in a Gray encoding with the potential truncation parameter set to~$K$:
\begin{multline}
\label{eq:HN_Gray}
   H_{N, K} = \sum\limits_{m=0}^{N-1} \langle m \vert H \vert m \rangle G^0_m \\
    + \sum\limits_{m=0}^{N-1} \sum\limits_{k=1}^{\max(K, 1)} \langle m+k \vert H \vert m \rangle (G^k_m + (G^k_m)^\dagger).
\end{multline}
Note that the summation over $k$ terminates if $m+k > N - 1$. 

For the example under consideration Eq.~\eqref{eq:eg},  the encoded Hamiltonian is given by
\begin{multline}
\label{eq:example_gray_enc}
    H_{4, 2} = 33.556\ II - 16.133\ ZI - 0.004\ IZ 
    - 8.073\ ZZ \\ - 17.586\ IX + 7.959\ ZX + 8.801 XZ \\
    - 8.801\ XI - 0.014\ XX - 0.006\ YY.
\end{multline}

\section{Lowest-state energy computation by variational quantum eigensolver \label{sec:VQE}}

\subsection{Variational principle}

For completeness, we summarize the variational principle that underpins the Variational Quantum Eigensolver (VQE). Given a Hamiltonian, the minimum energy eigenstate is called the ground state $\ket{\psi_g}$ and its energy $E_g$ is called the ground-state energy\footnote{
We note that for a model space restricted to a given total angular momentum $J$ (spin) and parity $\pi$ of the nucleus, the minimum energy eigenstate provides the lowest state for the given $J^\pi$ and coincides with the ground state only if the ground state has the same spin-parity. For example, in the present study of Carbon isotopes, the lowest ${\frac{1}{2}}^+$ state is the ground state only for $^{15}$C (the composite system for n+$^{14}$C) and $^{19}$C (the composite system for n+$^{18}$C), while the other Carbon isotopes under consideration have a ground state of different spin-parity.
}. More precisely, given a Hamiltonian $H$, the ground-state energy and the ground state are defined as:
\begin{equation}
\begin{aligned}
    E_g &\coloneqq  \min\limits_{\ket{\psi}}\ \langle \psi \vert H \vert \psi \rangle,  \\
    \ket{\psi_g} &\coloneqq  \underset{\ket{\psi}}{\operatorname{argmin}}\ \langle \psi \vert H \vert \psi \rangle.
\end{aligned}
\end{equation}

To estimate the ground-state energy (lowest $J^\pi$ state energy), we use a parameterized quantum circuit to attempt achieve the minimum. We define a set of unitaries $\{U(\btheta)\}_{\btheta}$ called an ansatz, such that 
\begin{align}
    E_{\btheta} &= \langle 0 \vert U^\dagger(\btheta) H U(\btheta) \ket{0} \nonumber \\
    &= \langle \psi(\btheta) \vert H \ket{\psi(\btheta)} \geq \langle \psi_g \vert H \ket{\psi_g} = E_g.
\end{align}
Since the inequality holds for every value $\btheta$, minimizing over every $\btheta$ provides an upper bound on the true ground-state energy:
\begin{equation}
    E_g \leq \min\limits_{\btheta} E_{\btheta},
\end{equation}
where the equality is achieved if the chosen ansatz can express every input state $\ket{\psi}$ for some parameter values. 

In this work, we use the Hamiltonian $H_{N, K}$ described in Sec.~\ref{sec:ProbDesc}, with truncation parameter $N$ for the size of the matrix and $K$ for the potential (referred to as hyperparameters). The choice of the potential is either a general central potential deduced \textit{ab initio} Eq.~\eqref{eq:Vgen} or an exponential potential Eq.~\eqref{eq:Vexp}.

\subsection{Ansatz description}
\label{sec:ansatz}

The choice of ansatz depends on the encoding scheme used. The structure of the ansatz chosen should ideally be able to express all possible combinations of the encoded basis states. A general hardware-efficient ansatz can be used \cite{HEA2017}, but the symmetry and structure of the Hamiltonian can influence and direct the ansatz definition.

Since the Hamiltonian is purely real, the eigenstates must be fully real. In conjunction with the fact that Hermitian matrices have real eigenvalues, this means that the Hamiltonian is diagonalized by an orthogonal transformation, and not a general unitary transformation. Thus, the ansatz unitary generates a real superposition of the basis states.

\subsubsection{One-hot ansatz}

For a given $N$, the ansatz choice for the one-hot encoding creates a real-coefficient superposition of the basis states. A pure state for an $N$-qubit one-hot basis with real coefficients can be expressed using generalized spherical coordinates. For example, with $N=4$,
\begin{multline}
    \label{eq:OneHotSuperPos}
    \ket{\psi (\btheta)} = \cos{\theta_1}\ket{0001} +  \sin{\theta_1}\cos{\theta_2}\ket{0010} + \\ \qquad \sin{\theta_1}\sin{\theta_2}\cos{\theta_3}\ket{0100} + \sin{\theta_1}\sin{\theta_2}\sin{\theta_3}\ket{1000}.
\end{multline}
Thus, for a truncation parameter $N$, the state is parameterized by $N-1$ parameters. The encoded state can be generated recursively using $R_y$ rotation gates and $\operatorname{CNOT}$ gates, as seen in Fig.~\ref{fig:OneHotAnsatz}.

\begin{figure}
    \includegraphics[width=0.8\columnwidth]{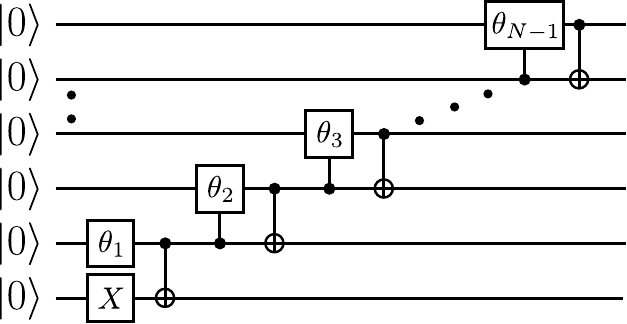}
    \caption{Recursive circuit ansatz to generate the superposition of the one-hot basis states. The input to the circuit is $\ket{0}^{\otimes N}$, and $\theta_i$ denotes the $R_y(2\theta_i)=\exp{(-\im \theta_i Y_i)}$ rotation gate.}
    \label{fig:OneHotAnsatz}
\end{figure}

\subsubsection{Binary and Gray ansatz}
In these encodings, we make use of the entire space spanned by the encoded basis states. This enables the use of a general hardware-efficient ansatz \cite{HEA2017}. However, from the argument above, we restrict this ansatz to create real superpositions of the basis states. This can be done with $R_y$ rotation gates and $\operatorname{CNOT}$ as the entangling gate. The ansatz is build up using multiple layers, each having a set of $Y$-rotations and entangling gates (see Fig.~\ref{fig:BR_GR_Ansatz}). The number of layers $L$ is a hyperparameter that needs to be chosen such that the ansatz is expressive enough without increasing the number of parameters too much. 

\begin{figure}
\includegraphics[width=0.5\columnwidth]{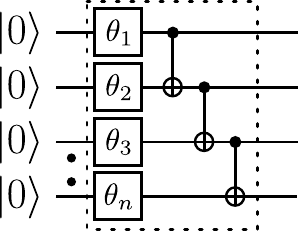}
    \caption{Circuit ansatz to generate a parameterized real superposition of all basis states. The input to the circuit is $\ket{0}^{\otimes n}$, where $n = \operatorname{log}_2(N)$. Each layer (marked with dotted lines) is repeated $L$ times. Thus, the total number of parameters is $nL$.}
    \label{fig:BR_GR_Ansatz}
\end{figure}

\section{Encoding techniques and trade-offs \label{sec:tradeoffs}}

In this section, we explore the various trade-offs between the encoding techniques. Important parameters for any simulation include the number of Pauli terms in the encoded Hamiltonian, the number of commuting sets, etc. A comprehensive analysis for a contact potential, or $\braketop{n_r'}{V}{n_r} = V_0\delta_{n_r, n_r'} \delta_{n_r,0}$ (see Fig.~\ref{fig:TruncatedPotential}a), can be found in Ref.~\cite{PhysRevA.103.042405}. We generalize these results for any band-diagonal to full Hamiltonian matrix needed to accommodate a general central potential, including exponential Gaussian-like potentials, \textit{ab initio} inter-cluster potentials, and the central part of any chiral NN potential for \textit{ab initio} nuclear calculations. Furthermore, we provide new insights and we discuss open research directions proposed in Ref.~\cite{PhysRevA.103.042405}.

We note that, in the present study, the Hamiltonian for $K=0$ (diagonal potential) is tridiagonal due to the kinetic energy term. As a result, the entries for $K=1$ are used for $K=0$. On the other hand, the most general results for a Hamiltonian matrix of $2K+1$ bandwidth (that permits a diagonal matrix) are summarized in Table~\ref{tab:details_OHencodings} for the one-hot encoding and in Table~\ref{tab:details_BGencodings} for the binary and Gray encodings. 

\begin{figure}
\includegraphics[width=\columnwidth]{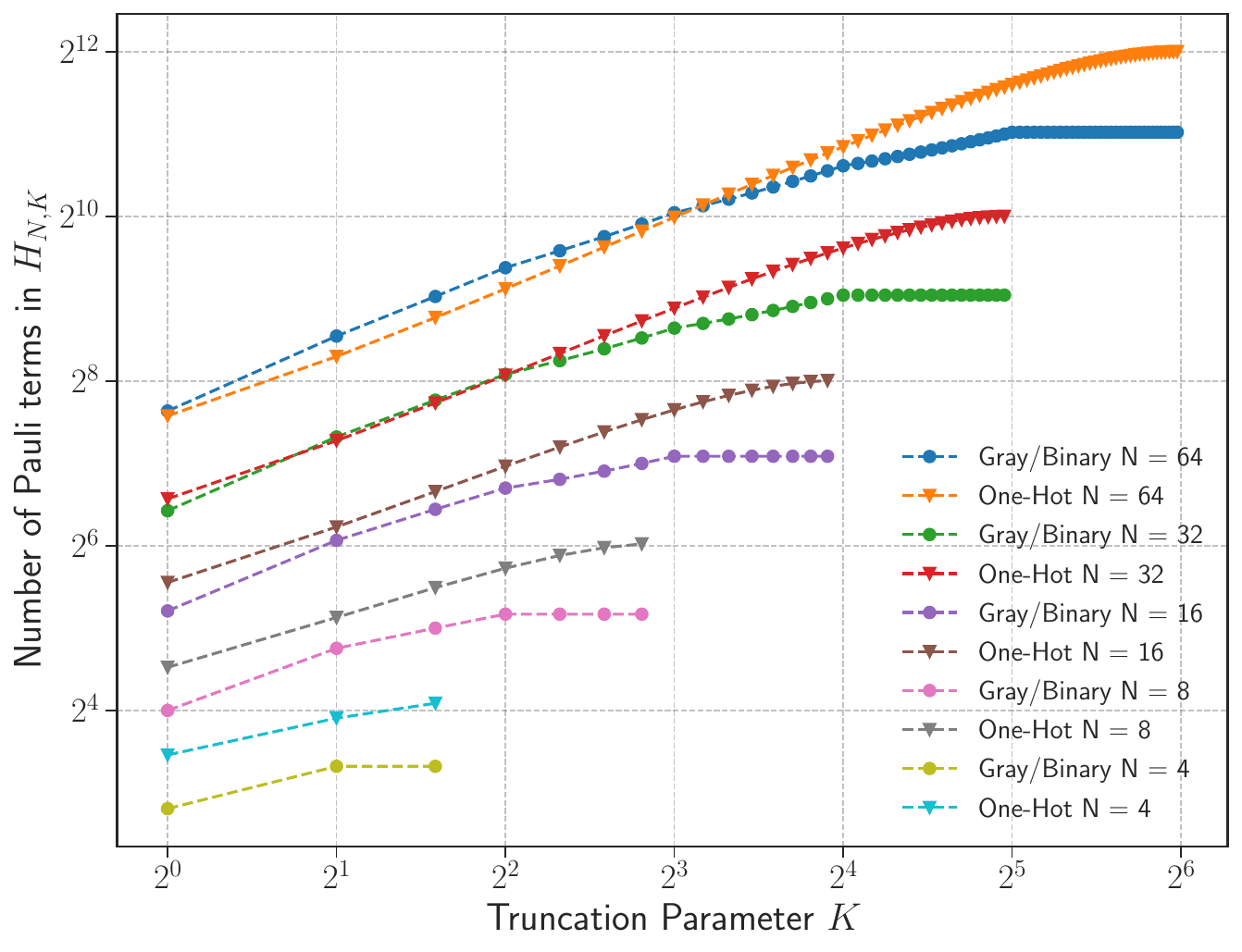}
    \caption{Number of Pauli terms for one-hot, binary, and Gray encodings for a general potential of the form $V_K(r) = \sum_{k=0}^K v_k r^{2k}$ and a general (tridiagonal to full) Hamiltonian matrix. Note that in the one-hot encoding, each term acts on $N$ qubits, while in the binary and Gray encodings, each term acts on $n= \operatorname{log}_2(N)$ qubits. }
    \label{fig:NumPauliH_NK}
\end{figure}

\subsection{Number of Pauli terms}

\subsubsection{One-hot encoding}

As seen in Eq.~\eqref{eq:HN_OneHot}, the encoded Hamiltonian for truncation parameters $N, K$ is given by 
\begin{multline}
    H_{N, K} = \frac{1}{2} \sum\limits_{m=0}^{N-1} \langle m \vert H \vert m \rangle (I_m - Z_m)\\
    + \frac{1}{2} \sum\limits_{m=0}^{N-1} \sum\limits_{k=1}^K \langle m+k \vert H \vert m \rangle (X_m X_{m+k} + Y_m Y_{m+k}),
\end{multline}
with the $k$ summation terminating when $m+k > N-1$. The above equation is true for $K \geq 1$; the case for $K=0$ is the same as $K=1$, since the kinetic energy term is tridiagonal. This results from the truncation of the matrix to size~$N \times N$.

In the first summation, the identity operator ($I^{\otimes N}$) and the $N$ individual $Z$ operators give a total of $N+1$ Pauli terms (e.g., for three qubits, the terms are $III$, $ZII$, $IZI$, and $IIZ$). For the second term, we split the $m$ summation into two parts: one with $m < N-K$ and $m \geq N-K$. In the first part, each $k$ summation goes from $1$ to $K$, contributing $2K$ terms. Thus, this part adds a total of $2K(N-K)$ terms. In the second part, the $k$ summation is truncated before it reaches $K$. This contributes a total of $2[(K-1) + (K-2) + \ldots + 1]$. Thus, the total of the number of Pauli terms is
\begin{equation}
    \left\vert H_{N, K} \right\vert = 1 + N + 2NK - K(K+1),
\end{equation}
as depicted in Fig.~\ref{fig:NumPauliH_NK}.

For the example of $N=4$ and $K=2$, we see that the number of terms is $15$, confirmed by \eqref{eq:example_OH_enc}. Note that each of these terms acts on $N$ qubits. We note that for $K>N/2$, we see that the number of Pauli terms is $\mathcal{O}(N^2)$.

\renewcommand{\arraystretch}{1.75}
\begin{table}[]
\centering
\begin{tabular}{|M{2cm}|M{6cm}|}
\hline
\multirow{2}{*}{$H_{N,K}$} & One-hot \\ 
\cline{2-2}
& $0 \leq K \leq N$ \\
\hline\hline

\multirow{2}{2cm}{\centering Qubits}  
& \multirow{2}{6cm}{\centering $N$} \\
& \\
\hline

\multirow{2}{2cm}{\centering Pauli Terms} 
& \multirow{2}{6cm}{\centering $1 + N + 2NK - K(K+1)$} \\
& \\
\hline

\multirow{2}{2cm}{\centering QC Sets} 
& \multirow{2}{6cm}{\centering $3$} \\
& \\
\hline

\multirow{2}{2cm}{\centering Ansatz} 
& \multirow{2}{6cm}{\centering $2$ one-qubit gates + $(2N-3)$ two-qubit gates} \\
& \\
\hline
\end{tabular}
\caption{Number of qubits, Pauli terms, and qubit-wise commuting sets for the one-hot code, for a general Hamiltonian matrix of $2K+1$ bandwidth. In the present study, the Hamiltonian for $K=0$ (diagonal potential) is tridiagonal because of the kinetic energy term, in which case the entries for $K=1$ should be used. More details can be found in Sec.~\ref{sec:discussions}.}
\label{tab:details_OHencodings}
\end{table}
\renewcommand{\arraystretch}{2.05}

\subsubsection{Binary and Gray encoding}

As seen in Eqs.~\eqref{eq:HN_Binary} and \eqref{eq:HN_Gray}, the encoded Hamiltonian for truncation parameters $N, K$ is given by
\begin{multline}
    H_{N, K} = \sum\limits_{m=0}^{N-1} \langle m \vert H \vert m \rangle L^0_m \\
    + \sum\limits_{m=0}^{N-1} \sum\limits_{k=1}^{K} \langle m+k \vert H \vert m \rangle (L^k_m + (L^k_m)^\dagger),
\end{multline}
where $L^i_j = B^i_j$ for the binary encoding and $L^i_j = G^i_j$ for the Gray encoding. The above equation is true for $K \geq 1$; the case for $K=0$ is the same as $K=1$, since the kinetic energy term is tridiagonal.

The number of Pauli terms in the Hamiltonian for both encodings is 
\begin{equation}
    \vert H(N, K) \vert 
        =  \begin{cases} 
            \operatorname{d}(n, 1) + n 2^{n-1} & K = 0 \\
            \operatorname{d}(n, K) + 2^{n-1} \sum\limits_{k=1}^K  \overline{n}_k & 1 \leq K \leq 2^{n-1} \\
            2^{n-1}(1+2^n) & K>2^{n-1},
        \end{cases}
\end{equation} 
as depicted in Fig.~\ref{fig:NumPauliH_NK}, where 
\begin{equation}
\label{eq:numDiagTerms}
    \operatorname{d}(n, K) = \sum_{m=0}^K \binom{n}{m},
\end{equation}
and $\overline{n}_k \coloneqq n - \lceil \operatorname{log}_2(k) \rceil $.
The proof for this can be found in Lemma~\ref{lem:NumPauliGrayBinary}. For the example of $N=4, K=2$, we see that the number of terms is $10$, confirmed by Eqs.~\eqref{eq:example_bin_enc} and \eqref{eq:example_gray_enc}. For $K>N/2$, we see that the number of Pauli terms is $\mathcal{O}(N^2)$.

To summarize, we show that the Gray and binary codes have the same number of Pauli terms for all $N$ and $K$. Furthermore, the number of Pauli terms saturates above $K = N/2$. The one-hot encoding does not saturate, and at $K = 2^{n-1}$, the one-hot encoding always has more Pauli terms than the Gray or binary encoding. For the general potential in consideration \eqref{eq:Vgen}, Fig.~\ref{fig:NumPauliH_NK} plots the number of terms as a function of $K$ for different $N$ values.

\renewcommand{\arraystretch}{2.05}
\begin{table*}[]
\begin{tabular}{|M{2cm}|M{4cm}|M{2.5cm}|M{4cm}|M{2.5cm}|}
\hline
\multirow{2}{*}{$H_{N, K}$} & \multicolumn{2}{c|}{Binary} & \multicolumn{2}{c|}{Gray} \\ \cline{2-5} 
& $K \leq N/2$ & $K > N/2$ & $K \leq N/2$ & $K > N/2$ \\ \hline\hline

\multirow[c]{2}{=}[-0.05cm]{\centering Qubits} 
& \multicolumn{2}{M{6.5cm}|}{\multirow[c]{2}{6.5cm}[-0.05cm]{\centering $\operatorname{log}_2(N)$}} 
& \multicolumn{2}{M{6.5cm}|}{\multirow[c]{2}{6.5cm}[-0.05cm]{\centering $\operatorname{log}_2(N)$}} \\ 
& \multicolumn{2}{c|}{} & \multicolumn{2}{c|}{} \\
\hline

\multirow[c]{2}{=}[-0.05cm]{\centering Pauli Terms} 
& \multirow[c]{2}{4cm}[-0.00cm]{\centering $\operatorname{d}(n, K) + 2^{n-1} \sum\limits_{k=1}^K  \overline{n}_k$} 
& \multirow[c]{2}{=}[-0.05cm]{\centering $2^{n-1}(1+2^n)$} 
& \multirow[c]{2}{4cm}[-0.00cm]{\centering $\operatorname{d}(n, K) + 2^{n-1} \sum\limits_{k=1}^K  \overline{n}_k$} 
& \multirow[c]{2}{=}[-0.05cm]{\centering $2^{n-1}(1+2^n)$} \\
& & & & \\
\hline

\multirow[c]{2}{=}[-0.05cm]{\centering QC Sets} 
& \multirow[c]{2}{=}[-0.00cm]{\centering $1 + \sum\limits_{k=1}^K 2^{\vert b(\bar{k}) \vert}[ 1- 2^{-\overline{n}_k}]$} 
& \multirow[c]{2}{=}[-0.05cm]{\centering $\frac{1}{2} \left( 1 + 3^n \right)$} 
& \multirow[c]{2}{=}[-0.00cm]{\centering $1 + \sum\limits_{k=1}^K \overline{n}_k 2^{\vert g(k-1) \vert}$} 
& \multirow[c]{2}{=}[-0.05cm]{\centering $\frac{1}{2} \left( 1 + 3^n \right)$} \\
& & & & \\
\hline

\multirow[c]{2}{=}[-0.05cm]{\centering DGC Sets} 
& \multirow[c]{2}{=}[-0.00cm]{\centering $1 + \sum\limits_{k=1}^{K} \overline{n}_k$} 
& \multirow[c]{2}{=}[-0.00cm]{\centering $2^n$} 
& \multirow[c]{2}{=}[-0.00cm]{\centering $1 + \sum\limits_{k=1}^{K} \overline{n}_k$} 
& \multirow[c]{2}{=}[-0.00cm]{\centering $2^n$} \\ 
& & & & \\
\hline

\multirow[c]{2}{=}[-0.05cm]{\centering 2QG in Diag. Unitary} 
& \multirow[c]{2}{=}[-0.00cm]{\centering $\frac{1}{2}\sum\limits_{k=1}^{K} \overline{n}_k \left[ 2 \vert b(\overline{k}) \vert - 1 - \overline{n}_k \right]$} 
& \multirow[c]{2}{=}[-0.00cm]{\centering $1 + 2^{n-1}(n-2)$} 
& \multirow[c]{2}{=}[-0.00cm]{\centering $\sum\limits_{k=1}^K \overline{n}_k \vert g_{k-1} \vert $} 
& \multirow[c]{2}{=}[-0.00cm]{\centering $1 + 2^{n-1}(n-2)$} \\ 
& & & &\\
\hline

\multirow[c]{2}{=}[-0.05cm]{\centering Ansatz} & \multicolumn{2}{M{6.5cm}|}{\multirow[c]{2}{6.5cm}[-0.05cm]{\centering $n L$ one-qubit gates + $(n-1)L$ two-qubit gates}} & \multicolumn{2}{M{6.5cm}|}{\multirow[c]{2}{6.5cm}[-0.05cm]{\centering $n L$ one-qubit gates + $(n-1)L$ two-qubit gates}}   \\ 
& \multicolumn{2}{c|}{} & \multicolumn{2}{c|}{} \\
\hline

\end{tabular}
\caption{Number of qubit, Pauli terms, qubit-wise commuting sets and distance-grouped commuting sets for both the binary and Gray code, for a general Hamiltonian matrix of $2K+1$ bandwidth.  Here we use $N = 2^n$ and $\overline{n}_k \coloneqq n - \lceil \operatorname{log}_2(k) \rceil$. The quantity $\operatorname{d}(n, K)$ is defined in \eqref{eq:numDiagTerms}. For the ansatz, $L$ refers to the number of layers, which is chosen large enough beforehand. In the present study, the Hamiltonian for $K=0$ (diagonal potential) is tridiagonal because of the  kinetic energy term, in which case the entries for $K=1$ should be used. Notably, for $K>N/2$, there is a saturation in all quantities, implying that larger diagonal width (better approximation) can be handled without an increase in the complexity of the problem. More details can be found in Sec.~\ref{sec:discussions}.}
\label{tab:details_BGencodings}
\end{table*}
\renewcommand{\arraystretch}{1.2}

\subsection{Number of commuting sets}
\label{sec:numCG}
To measure any operator provided as a linear combination of Pauli strings, 
\begin{equation}
    O = \sum\limits_{i=1}^{N_P} a_i P_i,
\end{equation}
we can measure individual Pauli terms and sum the results, because
\begin{equation}
    \operatorname{Tr}[O\rho] = \sum\limits_{i=1}^{N_P} a_i \operatorname{Tr}[P_i \rho].
\end{equation}
Thus, for an operator consisting of $N_P$ terms, we estimate the measurement statistics for each of the $N_P$ terms. 

Most quantum computers allow for measurements in the computational basis alone, i.e., in the Pauli-$Z$ basis. If the Pauli term does not have Pauli-$Z$ on a particular qubit, the qubit must first be rotated before a computational basis measurement. For example, to measure $X \otimes Z$ on a two-qubit state $\rho$, 
\begin{align}
    \operatorname{Tr}[(X \otimes Z) \rho] &= \operatorname{Tr}[(H \otimes I) (Z \otimes Z)(H \otimes I) \rho] \notag \\
    &= \operatorname{Tr}[(Z \otimes Z)(H \otimes I) \rho (H \otimes I)].
\end{align}
Thus, measuring $X \otimes Z$ is equivalent to applying Hadamard on the first qubit and then measuring in the computational basis.

A method to reduce the number of measurements is to measure commuting observables in their common eigenbasis. The idea of reducing the measurement complexity has led to a large number of advances \cite{VYI20, YVI20, Huggins_2021}. Two commuting Pauli strings are guaranteed to have a common eigenbasis. Consider two Pauli strings~$A$ and $B$ such that $[A, B] = 0$, and denote the common eigenbasis as $\{\ket{\psi_i}\}_i$. Thus,
\begin{equation}
    A = \sum\limits_i a_i \outerprod{\psi_i}{\psi_i} , \qquad 
    B = \sum\limits_i b_i \outerprod{\psi_i}{\psi_i}.
\end{equation}
The elements $\ket{\psi_i}$ are related to the computational basis elements $\ket{i}$ by a unitary transformation. More concretely, 
\begin{equation}
    \ket{\psi_i} = U \ket{i}.
\end{equation}
Thus, to measure $A$ and $B$ simultaneously, we first apply the unitary $U^\dagger$ and then measure in the computational basis:
\begin{align}
    \operatorname{Tr}[A\rho] &= \sum\limits_i a_i \operatorname{Tr}[\outerprod{\psi_i}{\psi_i}\rho] \notag \\
    &= \sum\limits_i a_i \operatorname{Tr}[U\outerprod{i}{i}U^\dagger \rho] \notag \\
    &= \sum\limits_i a_i \operatorname{Tr}[\outerprod{i}{i} U^\dagger \rho U], 
\end{align}
and similarly for $B$:
\begin{equation}
\operatorname{Tr}[B\rho]=\sum\limits_i b_i \operatorname{Tr}[\outerprod{i}{i} U^\dagger \rho U].    
\end{equation}
Thus, we can recreate the measurement statistics of both $A$ and $B$ by measuring the state $U^\dagger \rho U$ in the computational basis. To measure multiple observables, each element needs to pair-wise commute with all other elements. However, splitting a set of observables into commuting sets and finding the common eigenbasis is non-trivial. Furthermore, finding the $U$ that rotates the computational basis into this common eigenbasis is also non-trivial. 

We now look at two alternate simpler strategies. The first strategy is to look at qubit-wise commutativity~(QC). Two Pauli strings qubit-wise commute if  the corresponding operators acting on each qubit commute. For example, $XIZ$ and $XZI$ qubit-wise commute since $[X, X] = [I, Z] = [Z, I] = 0$. Qubit-wise commutativity is a sufficient, but not necessary, condition for general commutativity. For example, $XX$ and $YY$ commute but do not qubit-wise commute. Thus, this strategy leads to a sub-optimal grouping of Pauli strings. However, the grouping itself can be done efficiently, and the unitary that rotates the computational basis into the common eigenbasis is always a tensor product of individual Pauli operators.

In addition, in this paper, we introduce a new strategy that is based on grouping Pauli strings in terms of the distance operators defined in Appendix~\ref{app:LemDef}. We refer to this grouping  as distance-grouped commutativity (DGC). While the precise structure is not relevant here, it leads to a more optimal set of Pauli operators as compared to the qubit-wise scheme. However, the diagonalizing unitary, while simple conceptually, is no longer a tensor-product of individual Pauli operators (see Appendix~\ref{app:LemDef} for further details). 

In this section, we analyze the number of QC and DGC sets that the Pauli terms can be split into, for all encodings. The number of sets as a function of the hyperparameter $K$ and for various model-space sizes $N$ is shown in Figs.~\ref{fig:NumCommutingH_NK}-\ref{fig:CompNumCommutingH_NK}.

\begin{figure}
\includegraphics[width=\columnwidth]{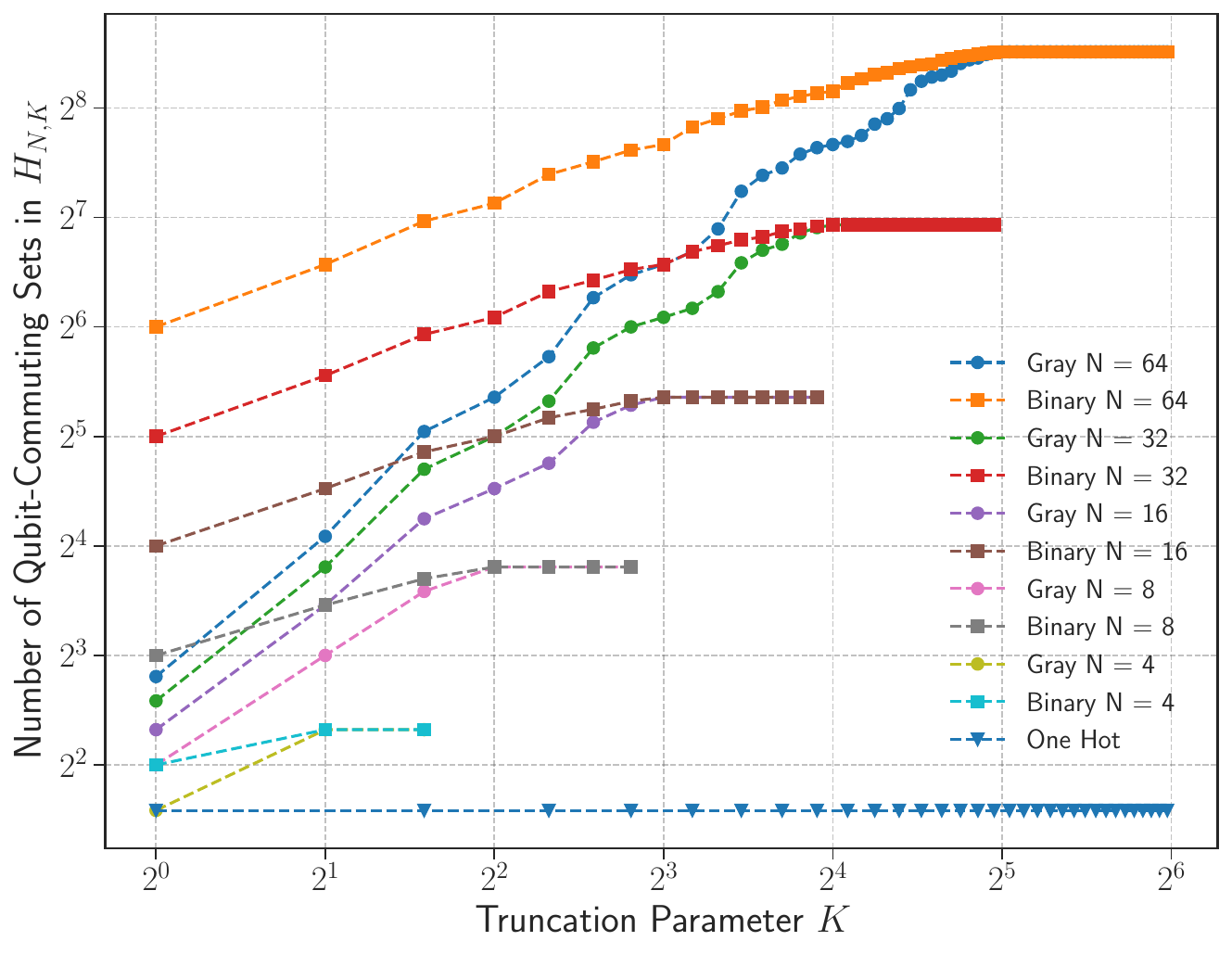}
    \caption{Number of qubit-wise commuting Pauli sets terms for one-hot, binary, and Gray encodings for a general potential of the form $V_K(r) = \sum_{k=0}^K v_k r^{2k}$ and a general (tridiagonal to full) Hamiltonian matrix. }
    \label{fig:NumCommutingH_NK}
\end{figure}

\begin{figure}
\includegraphics[width=\columnwidth]{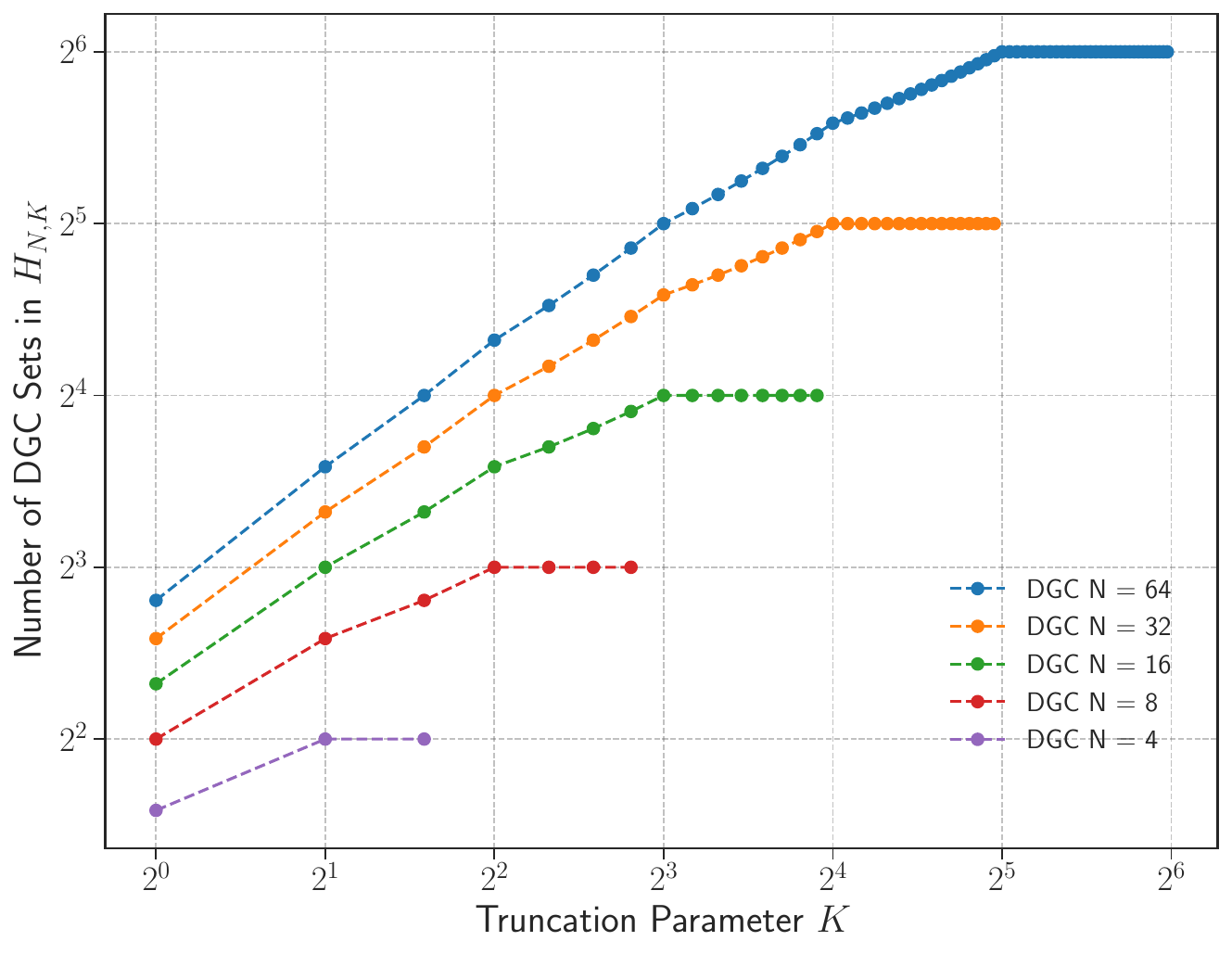}
    \caption{Number of DGC sets for the binary, and Gray encodings for a general potential of the form $V_K(r) = \sum_{k=0}^K v_k r^{2k}$ and a general (tridiagonal to full) Hamiltonian matrix.}
    \label{fig:NumDGCommutingH_NK}
\end{figure}

\begin{figure}
\includegraphics[width=\columnwidth]{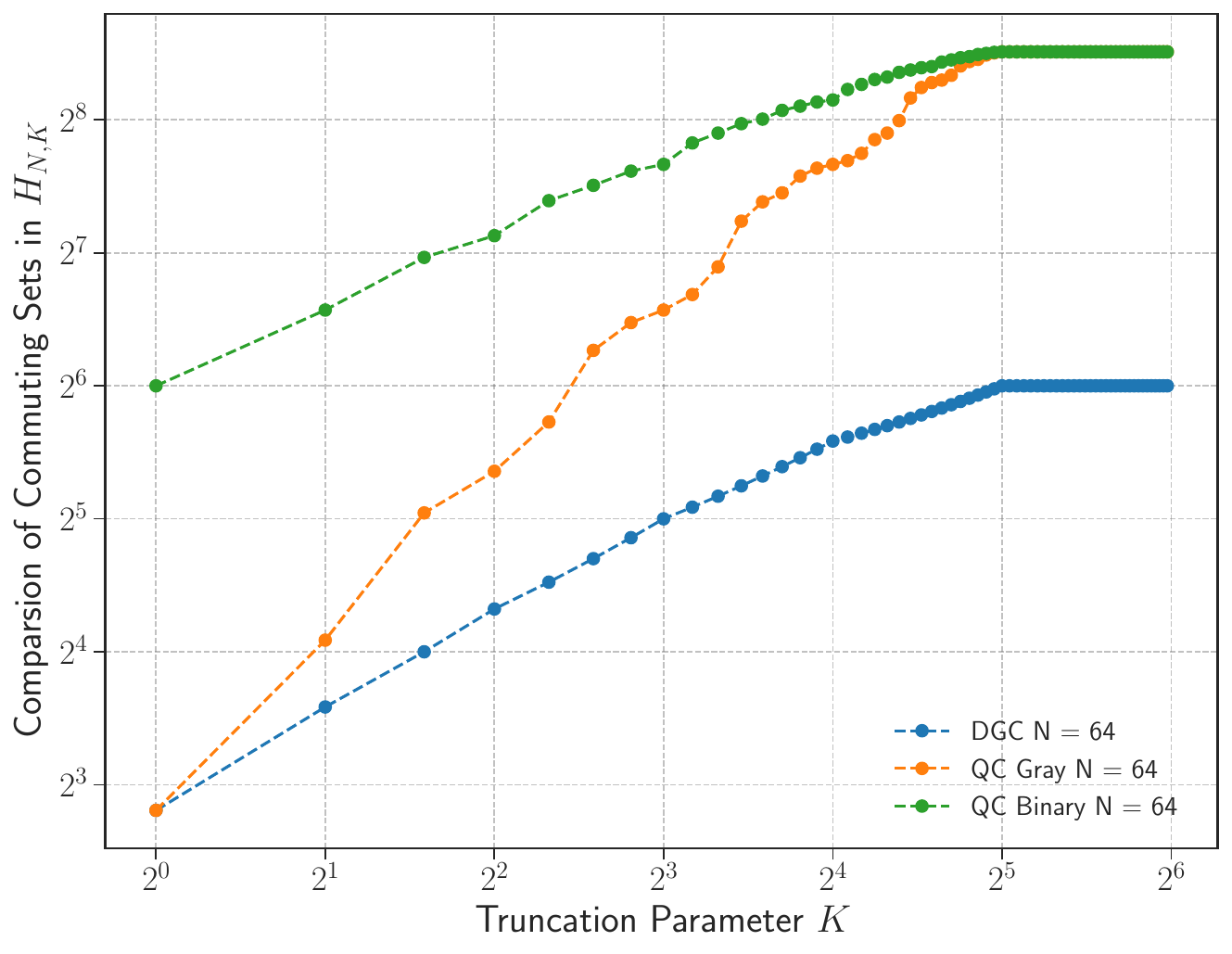}
    \caption{Comparing the two commutativity schemes for a general potential of the form $V_K(r) = \sum_{k=0}^K v_k r^{2k}$ and a general (tridiagonal to full) Hamiltonian matrix. We note that the DGC leads to fewer terms to be measured, at the cost of a more complex rotation gate into a common eigenbasis.}
\label{fig:CompNumCommutingH_NK}
\end{figure}

\subsubsection{One-hot encoding: Qubit-wise commutativity}

As seen in Eq.~\eqref{eq:HN_OneHot}, the first summation results from the number operators. These operators all qubit-wise commute, and their common eigenbasis is the computational basis.
Thus, the statistics of these operators can be inferred from a measurement of $Z^{\otimes N}$. Similarly, the ladder operators can be split into two sets -- one with only $X$ and $I$, and another with only $Y$ and $I$. Within each of these sets, all operators qubit-wise commute. Thus, the measurement statistics of all the operators can be inferred from the measurement of $\{X^{\otimes N}, Y^{\otimes N}\}$. Thus, the number of qubit-wise commuting sets is three.

\subsubsection{Binary encoding: Qubit-wise commutativity}

As stated in Lemma~\ref{lem:NumQCGBinary} in Appendix~\ref{app:LemDef}, the number of qubit-wise commuting sets in a binary encoding is given by
\begin{equation}
    \vert H(N, K) \vert_C
    =  \begin{cases} 
            2^n & K = 0 \\
            1 + \sum\limits_{k=1}^K 2^{\vert b_{\bar{k}} \vert}\left[ 1- 2^{-\overline{n}_k} \right] & 1 \leq K \leq 2^{n-1} \\
            \frac{1}{2} \left( 1 + 3^n \right) & K > 2^{n-1},
        \end{cases}
\end{equation}
where $\vert w \vert$ is the Hamming weight of the string $w$, the variable $\bar{k}$ is defined as $2^n - k$, $\overline{n}_k \coloneqq n - \lceil \operatorname{log}_2(k) \rceil $, and $b_i$ denotes the binary representation of $i$. We first notice that the number of qubit-wise commuting sets for $K > N/2$ is $\mathcal{O}(N^{\operatorname{log}_2(3)})$. 

For the example of $N=4, K=2$ from Eq.~\eqref{eq:example_bin_enc}, qubit-wise commutativity leads to the following sets:
\begin{align}
    &\{II, ZI, IZ, ZZ\}, \notag \\
    &\{IX, ZX\}, \notag \\
    &\{XI, XZ\}, \notag \\
    &\{XX\}, \notag \\
    &\{YY\}.
\end{align}

\subsubsection{Gray encoding: Qubit commutativity}

As stated in Lemma~\ref{lem:NumQCGGray} in Appendix~\ref{app:LemDef}, the number of qubit-wise commuting sets in a Gray encoding is given by
\begin{equation}
    \vert H(N, K) \vert_C
    =  \begin{cases} 
        1 + n & K = 0 \\
        1 + \sum\limits_{k=1}^K 2^{\vert g_{k-1} \vert} \overline{n}_k & 1 \leq K \leq 2^{n-1} \\
        \frac{1}{2} \left( 1 + 3^n \right) & K>2^{n-1},
        \end{cases}
\end{equation}
where $\vert w \vert$ is the Hamming weight of the string $w$, $\overline{n}_k \coloneqq n - \lceil \operatorname{log}_2(k) \rceil $, and $g_k$ denotes the $k$th entry in the Gray basis. We first notice that the number of sets for $K > N/2$ is $\mathcal{O}(N^{\operatorname{log}_2(3)})$. The set of Pauli strings for $K>N/2$ being exactly the same as the binary encoding, leads to the number of commuting sets being equal.
However, we note that for $K<N/2$ the Gray encoding leads to a lower number of qubit-wise commuting sets than the binary encoding. Indeed, the ordering of the computational basis elements in the Gray encoding favors low-weight Pauli strings for lower $K$, leading to a lower number of qubit-wise commuting sets.

For the example of $N=4, K=2$ from Eq.~\eqref{eq:example_gray_enc}, qubit-wise commutativity leads to the following sets:
\begin{align}
    &\{II, ZI, IZ, ZZ\}, \notag \\
    &\{IX, ZX\}, \notag \\
    &\{XI, XZ\}, \notag \\
    &\{XX\}, \notag \\
    &\{YY\}.
\end{align}
Since this example is for the case of $K=N/2$, as expected, the number of QC sets is the same for both binary and Gray encodings.

\subsubsection{Binary/Gray encoding: Distance-grouped commutativity}
As seen in Lemma~\ref{lem:NumDGCBinGray}, the number of distance-grouped commuting sets is given by
\begin{equation}
    \vert H(N, K) \vert_C 
    =  \begin{cases} 
            1+n & K = 0 \\
            1 + \sum\limits_{k=1}^{K} \overline{n}_k & 1 \leq  K \leq 2^{n-1} \\
            2^n & K > 2^{n-1},
        \end{cases}
\end{equation}
where $\overline{n}_k \coloneqq n - \lceil \operatorname{log}_2(k) \rceil $. We note that the number of DGC sets for $K > N/2$ is $\mathcal{O}(N)$. However, each measurement requires a more complex unitary transformation to a common eigenbasis as compared to the qubit-commutative sets. 

For the example of $N=4, K=2$ from Eqs.~\eqref{eq:example_gray_enc} and \eqref{eq:example_bin_enc}, distance-grouped commutativity leads to the following sets:
\begin{align}
    &\{II, ZI, IZ, ZZ\}, \notag \\
    &\{IX, ZX\}, \notag \\
    &\{XI, XZ\}, \notag \\
    &\{XX, YY\}.
\end{align}

While the number of DGC sets for the Gray and binary encoding are exactly the same for all $N$ and $K$, the complexity in the measurement procedure is not the same. We now quantify the complexity of the measurement scheme based on the number of two-qubit gates in the diagonalizing unitaries for both encodings.

As stated in Lemma~\ref{lem:2QG_Binary} in Appendix~\ref{app:LemDef}, the number of two-qubit gates in the diagonalizing unitary using the DGC scheme for the binary code is given by
\begin{align}
    &\vert H(N, K) \vert_{DU} \notag \\
    &=  \begin{cases} 
            0.5n(n-1) & K = 0 \\
            0.5\sum\limits_{k=1}^K \overline{n}_k \left[ 2 \vert b(\overline{k}) \vert - 1 - \overline{n}_k \right] & 1 \leq K \leq 2^{n-1} \\
            1 + 2^{n-1}(n-2) & K > 2^{n-1},
        \end{cases}
\end{align}
where $\overline{n}_k \coloneqq n - \lceil \operatorname{log}_2(k) \rceil $.

Similarly, from Lemma~\ref{lem:2QG_Gray} in Appendix~\ref{app:LemDef}, the number of two-qubit gates in the diagonalizing unitary using the DGC scheme for the Gray code is given by
\begin{equation}
    \vert H(N, K) \vert_{DU} = 
        \begin{cases} 
        0 & K = 0 \\
        \sum\limits_{k=1}^K \overline{n}_k g_{k-1} & 1 \leq K \leq 2^{n-1} \\
        1 + 2^{n-1}(n-2) & K > 2^{n-1},
        \end{cases}
\end{equation}
where $\overline{n}_k \coloneqq n - \lceil \operatorname{log}_2(k) \rceil $, $\bar{k}$ is defined as $2^n - k$, and $\vert b(k) \vert $ is Hamming weight of the binary representation of~$k$. A comparison of the two encoding schemes is given in Fig.~\ref{fig:CompNum2QG_H_NK}.

\begin{figure}
\includegraphics[width=\columnwidth]{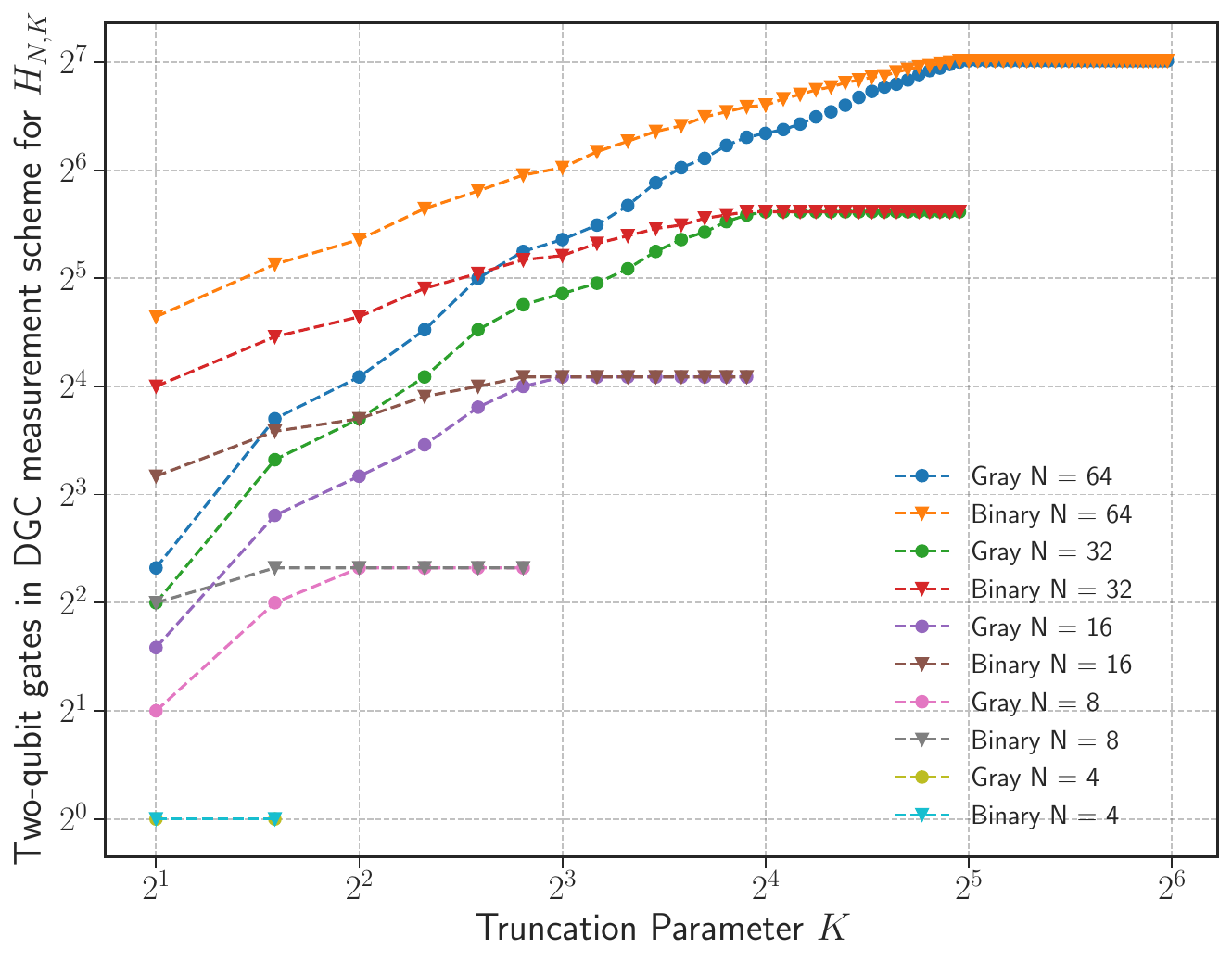}
    \caption{Comparing the number of two-qubit gates needed in the DGC measurement scheme for the binary and Gray codes for a general potential of the form $V_K(r) = \sum_{k=0}^K v_k r^{2k}$ and a general (tridiagonal to full) Hamiltonian matrix.}
    \label{fig:CompNum2QG_H_NK}
\end{figure}

To summarize, compared to the binary encoding, the Gray encoding has the same number of Pauli strings, and for a bandwidth of up to $N$, it has a lower number of QC sets, the same number of DGC sets but a less complex diagonalizing unitary. The advantage of the Gray encoding over other encodings comes from the fact that every basis entry differs from its neighbours on a single bit. This, coupled with the fact that the Hamiltonian must be Hermitian, guarantees $G^1_m + (G^1_m)^\dagger$ to be Pauli-$X$ on the single flipped qubit, and single-qubit projectors on the remaining qubits. This leads to the same number of Pauli strings as compared to the binary encoding, a lower number of QC sets, and the same number of DGC sets but with a lower number of two-qubit gates for low bandwidths. 

For example, using the Gray encoding for $N=8$, the step-$1$ ladder operator $G^1_1$ along with its Hermitian conjugate consists of the following Pauli strings:
\begin{align}
    G^1_1 + (G^1_1)^\dagger &= P^1_1 \otimes X_1 \otimes P^0_2 \notag \\
    &= IXI + IXZ - ZXI - ZXZ,
\end{align}
where the qubit numbers are omitted for clarity. The measurement statistics of all these Pauli strings, using both the QC and DGC scheme, can be inferred from the measurement statistics of $ZXZ$.

In comparison, the binary encoding for the same step-$1$ ladder operator $B^1_1$ along with its Hermitian conjugate consists of the following Pauli strings:
\begin{align}
    &B^1_1 + (B^1_1)^\dagger \notag \\
    &= P^1_0 \otimes \left[(\outerprod{1}{0}_1 \otimes \outerprod{0}{1}_2) + (\outerprod{0}{1}_1 \otimes \outerprod{1}{0}_2)\right], \notag \\
    &= IXX + IYY - ZXX - ZYY,
\end{align}
where the qubit numbers are omitted for clarity. In the QC scheme, we need to measure $ZXX$ and $ZYY$ to infer the measurement statistics of all the above operators. In the DGC scheme, we need to measure $Z_0 \otimes (U^{\operatorname{GHZ}}_{1, 2})^\dagger (Z_1 \otimes Z_2) U^{\operatorname{GHZ}}_{1, 2}$, where $U^{\operatorname{GHZ}}$ is defined in Lemma~\ref{lem:diagUnitaryDGC} in Appendix~\ref{app:LemDef}.

\begin{remark}
    As seen in Table~\ref{tab:details_OHencodings}, the number of commuting sets for the one-hot encoding is always three, independent of system size. On the other hand, for the binary and Gray encodings detailed in Table~\ref{tab:details_BGencodings}, we find that the number of qubit-wise commuting sets and distance-grouped commuting sets depends on $N, K$ and is strictly greater than three. While this gives the impression that the one-hot encoding leads to simpler measurements, we note that these measurements are on $N$ qubits, as opposed to $\operatorname{log}_2(N)$ qubits.
\end{remark}

\section{Quantum simulations and discussions \label{sec:discussions}}

In this section, we perform quantum simulations using the Gray and one-hot encodings, with and without noise effects, and we discuss the results. The aim is to illustrate if nuclear problems with band-diagonal Hamiltonian matrices can be solved on current quantum devices and to compare the two types of encodings. Specifically, we study the lowest $\half ^+$ bound state in the n+$^{10}$C, n+$^{12}$C, and n+$^{14}$C, using an exponential potential given in Eq.~\eqref{eq:Vexp}. We also perform quantum simulations for the lowest $\half^+$ orbit in the \textit{ab initio} deduced n-$\alpha$ local potential\footnote{We note that an optical potential, such as the one derived in the Green's function approach \cite{BurrowsL21}, provides a mean field that, for n+$\alpha$, yields a negative-energy $\half^+$ orbit, occupied by the protons and neutron of the $\alpha$ particle, and is associated with the physics of a hole in $^4$He, that is, with the ground state of  $^3$He. The next $\half^+$ eigensolution (in increasing energy) corresponds to a scattering state in $^5$He.}.  

\begin{figure}
    \includegraphics[width=0.75\columnwidth]{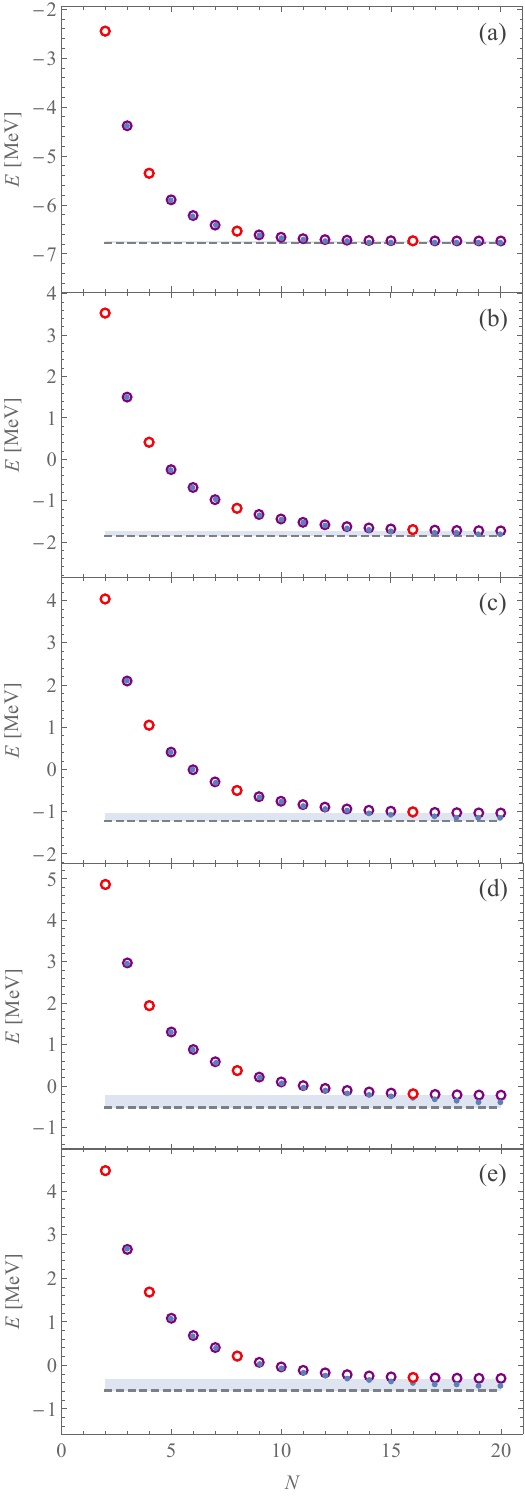}
    \caption{Energy of the lowest $\half^+$ state in (a) n+$^{10}$C, (b) n+$^{12}$C, (c) n+$^{14}$C, (d) n+$^{16}$C, and (e) n+$^{18}$C,  vs. the model-space size $N$ for $V_{\rm E}$ \eqref{eq:Vexp} (blue filled) and its $K=3$ approximation (purple open, for $N$ qubits; red open, for $n=\log_2 N$ qubits), as compared to experiment (dashed line) \cite{KELLEY201288C11,AJZENBERGSELOVE19911,exptC2023ENDSF}. Shaded area provides the error associated with the $K=3$ approximation (the small error band for n+$^{10}$C is not visible on the plot).}
    \label{fig:nC_EvsN}
\end{figure}

\subsection{Exponential potentials for neutron-Carbon dynamics}

For Carbon targets, we use the exponential potential of Eq.~\eqref{eq:Vexp} with parameters given in Table~\ref{tab:Cdata}. These parameterizations yield effective interactions that closely reproduce the experimental energy of the lowest $\half^+$ state in the composite n+C system (cf. the calculated energy~$E_{\rm th}$ and experimental energy $E_{\rm expt}$ in Table~\ref{tab:Cdata}). 

\begin{table}[th]
\centering
\begin{tabular}{c|c|c|c|c|c|c|}
    & $A$ &$E_{\rm expt}$ [MeV]& $E_{\rm th}$ [MeV] & $E_{\rm th}^{K=3}$ [MeV] &$\frac{1}{\hbar \omega}V_0$& $c^{-\half}$\\
    \hline
n+$^{10}$C&	$10$	&	$-6.78$	&	$-6.78$	&	$-6.74$	&	$	-0.650	$ & $	5.43	$	\\
n+$^{12}$C&	$12$	&	$-1.86$	&	$-1.86$	&	$-1.74$	&	$	-0.283	$ & $	5.35	$	\\
n+$^{14}$C&	$14$	&	$-1.22$	&	$-1.22$	&	$-1.04$	&	$	-0.242	$ & $	5.0	$	\\
n+$^{16}$C&	$16$	&	$-0.52$	&	$-0.52$	&	$-0.23$	&	$	-0.175	$ & $	4.7	$	\\
n+$^{18}$C&	$18$	&	$-0.58$	&	$-0.59$	&	$-0.30$	&	$	-0.192	$ & $	4.6	$	\\\end{tabular}
    \caption{Experimental energy for the lowest $\half^+$ state for each neutron-Carbon system, with the corresponding theoretical energy $E_{\rm th}$ of the exponential potential $V_{\rm E}$ \eqref{eq:Vexp} and its $K=3$ approximation $E_{\rm th}^{K=3}$. For each case, the parameters ($V_0$ and $c$) of the Hamiltonian are shown, with 
    $\hbar \omega =\frac{41}{(A+1)^{1/3}}$ MeV, where $A$ is  the mass of the target.}
    \label{tab:Cdata}
\end{table}
 
For Gaussian-like exponential potentials, the truncation parameter $K$ needs to be odd, since for even $K$ values, the potential curves downward with increasing distance and leads to spurious bound states.
Furthermore, the $K=3$ case provides a reasonable approximation, as illustrated by the $K=3$ error band in Fig.~\ref{fig:nC_EvsN}. This error is smaller for the lighter isotopes, where the $\half^+$ state is more deeply bound (larger binding energy) compared to the one for n+$^{16}$C and n+$^{18}$C  (see also $E_{\rm th}^{K=3}$ in Table~\ref{tab:Cdata}).
It is important to note that weakly bound states converge at a slower rate compared to states with larger binding energies. As shown in Fig.~\ref{fig:nC_EvsN}, the $N=8$ energy (or 3 qubits for the Gray encoding) closely agrees with the exact value for n+$^{10}$C, but larger model spaces are needed for the other systems. The model-space size requirements become even larger for resonances. This suggests that the use of the Gray encoding that can reach large model spaces with fewer qubits becomes advantageous for weakly bound states and resonances. To illustrate the advantages of the Gray encoding compared to the one-hot encoding, we perform quantum simulations for n+$^{14}$C using both encodings, which is detailed below.

\subsection{\textit{Ab initio} deduced local optical potential for n\texorpdfstring{$+\alpha$}{+alpha}}
\label{sec:ab_initio_potentials}

The n+$^4$He optical potential is calculated using the \textit{ab initio} symmetry-adapted no-core shell model with Green's function approach (SA-NCSM/GF) \cite{BurrowsL21}, with the chiral NNLO$_{\rm opt}$ nucleon-nucleon potential \cite{Ekstrom13}, at a center-of-mass energy  $E=0$ MeV, and for $15$ HO shells with $\hbar\omega=12$ and $16$ MeV  (see Fig.~\ref{fig:choosing-truncation-k-local}a). The choice for the basis parameters, the total number of HO shells and $\hbar\omega$, is based on a systematic study of large-scale calculations reported in Ref.~\cite{BurrowsL21}; namely, Ref.~\cite{BurrowsL21} has shown that for these parameters, n+$\alpha$ phase shifts are converged with respect to the model-space size (see Fig.~8 of Ref.~\cite{BurrowsL21}), leading to a parameter-free estimate for the total cross section that is shown to  reproduce the experiment (see Fig.~3 of Ref.~\cite{BurrowsL21}). From the \textit{ab initio} optical potential and using Eq.~\eqref{eq:locV}, one can calculate the local potential $V(r)$ for $\hbar\omega=12 $ MeV  (Fig.~\ref{fig:choosing-truncation-k-local}b):
\begin{eqnarray}
V(r)&=&-57.207 + 6.653 r^2 + 0.086 r^4 - 0.013 r^6  \\ \nonumber &-&  
 0.001 r^8 - 1.8 \times10^{-5} r^{10} + 2.3\times10^{-6} r^{12}  \\ \nonumber &+&  
 2.1\times10^{-7} r^{14} + 5.7\times10^{-9 } r^{16} - 3.6\times10^{-10} r^{18} \\ \nonumber &-&  
 4.3\times10^{-11} r^{20} - 1.5\times10^{-12} r^{22} \\ \nonumber &+& 5.0\times10^{-14} r^{24} + \mathcal{O}(r^{26}) \,\,\,\,\,\,\, (r \lesssim 3.5 {\rm fm}),
 \label{eq:Vhw12}
\end{eqnarray}
with the corresponding lowest eigenvalue of the original nonlocal potential $E_0=-18.85$ MeV.
\begin{figure}
(a) n+$^4$He \textit{ab initio} potential \\
\includegraphics[width=0.72\columnwidth]{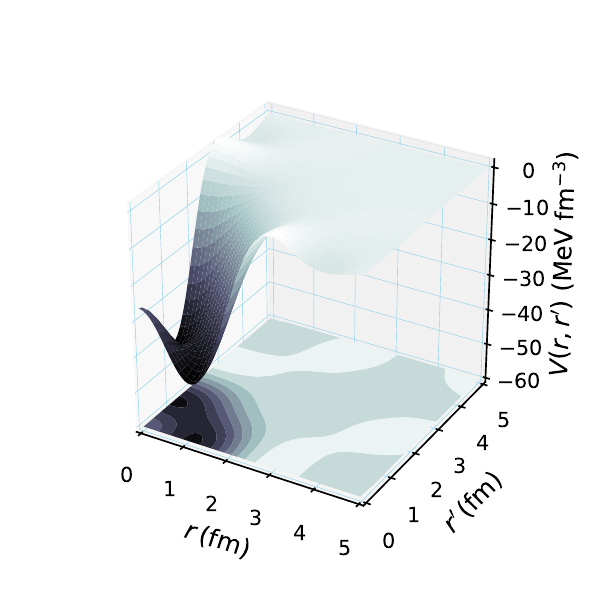}\\
    (b) $\hbar\omega=12$ MeV\\
    \includegraphics[width=0.92\columnwidth]{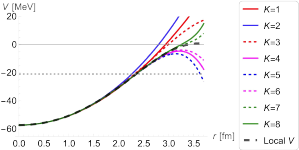}\\
    (c) $\hbar\omega=16$ MeV\\
    \includegraphics[width=0.92\columnwidth]{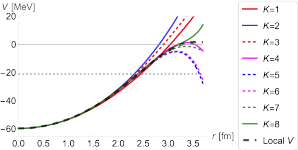}
    \caption{
    (a) \textit{Ab initio} nonlocal n+$^4$He potential as a function of the distance $r$ from the center-of-mass of the $\alpha$ particle for the $S_{\frac12}$ partial wave, calculated in the SA-NCSM/GF with $\hbar\omega=16$ MeV and 15 HO shells, for zero projectile energy (see also Fig.~4 in Ref.~\cite{BurrowsL21} for other energies and partial waves).
    (b) \& (c) Approximate potentials for different truncation parameters $K$, as compared to the \textit{ab initio} deduced local potential labeled as ``Local $V$" (black, solid) for (b) $\hbar\omega=12$ MeV and (c) $\hbar\omega=16$ MeV. The corresponding exact lowest eigenvalue $E_0$ is also shown (gray, dotted).}
    \label{fig:choosing-truncation-k-local}
\end{figure}

The local potential $V(r)$ for $\hbar \omega=16$ MeV (Fig.~\ref{fig:choosing-truncation-k-local}c):
\begin{eqnarray}
V(r)&=&-59.571 + 6.448 r^2 + 0.133 r^4 - 0.007 r^6  \\ \nonumber &-& 
 0.001 r^8 - 4.8 \times 10^{-5} r^{10} + 4.1\times 10^{-7} r^{12}  \\ \nonumber &+&  
 2.3\times 10^{-7} r^{14} + 1.5\times10^{-8} r^{16} + 3.0\times10^{-10} r^{18}  \\ \nonumber &-&  
 3.6\times10^{-11} r^{20} - 3.8\times10^{-12} r^{22} \\ \nonumber &-& 1.5\times10^{-13} r^{24}+ \mathcal{O}(r^{26}) \,\,\,\,\,\,\, (r \lesssim 3.5 {\rm fm}),
 \label{eq:Vhw16}
\end{eqnarray}
with the corresponding lowest eigenvalue of the original nonlocal potential $E_0= -20.84$ MeV. Hence, across the range of $\hbar\omega=12$-16 MeV, the \textit{ab initio} potential yields the lowest $1/2^+$ orbit at energy $E_0=-19.8 \pm 1.0$ MeV, which closely agrees with the corresponding $^3$He experimental energy of $-20.58$ MeV associated with a neutron removal from the $^4$He target.

\begin{figure*}[th]
    \begin{subfigure}[b]{0.49\textwidth}
    \centering
    \includegraphics[width=\textwidth]{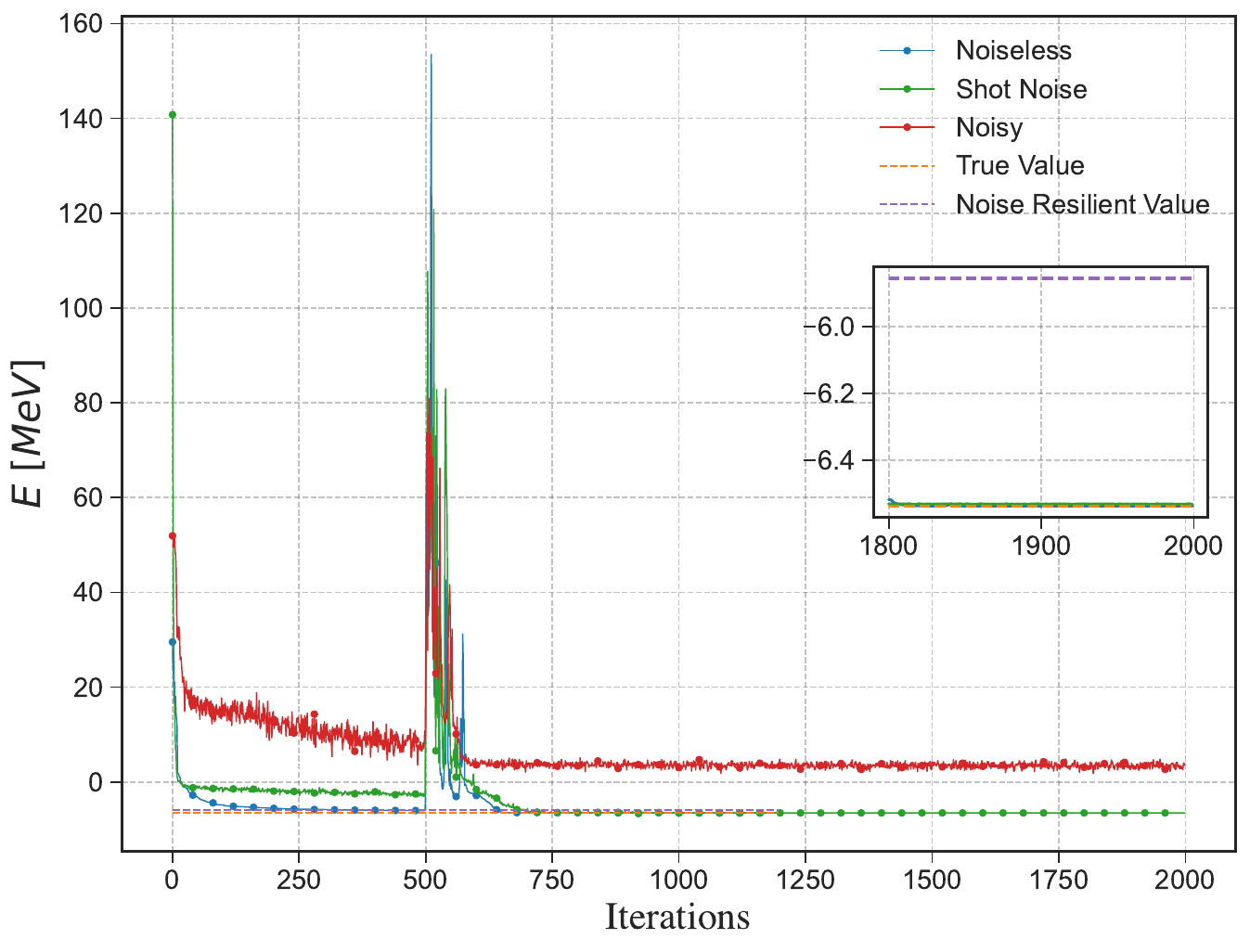}\\ (a)
    \end{subfigure} 
    \hfill
    \begin{subfigure}[b]{0.49\textwidth}
    \includegraphics[width=\textwidth]{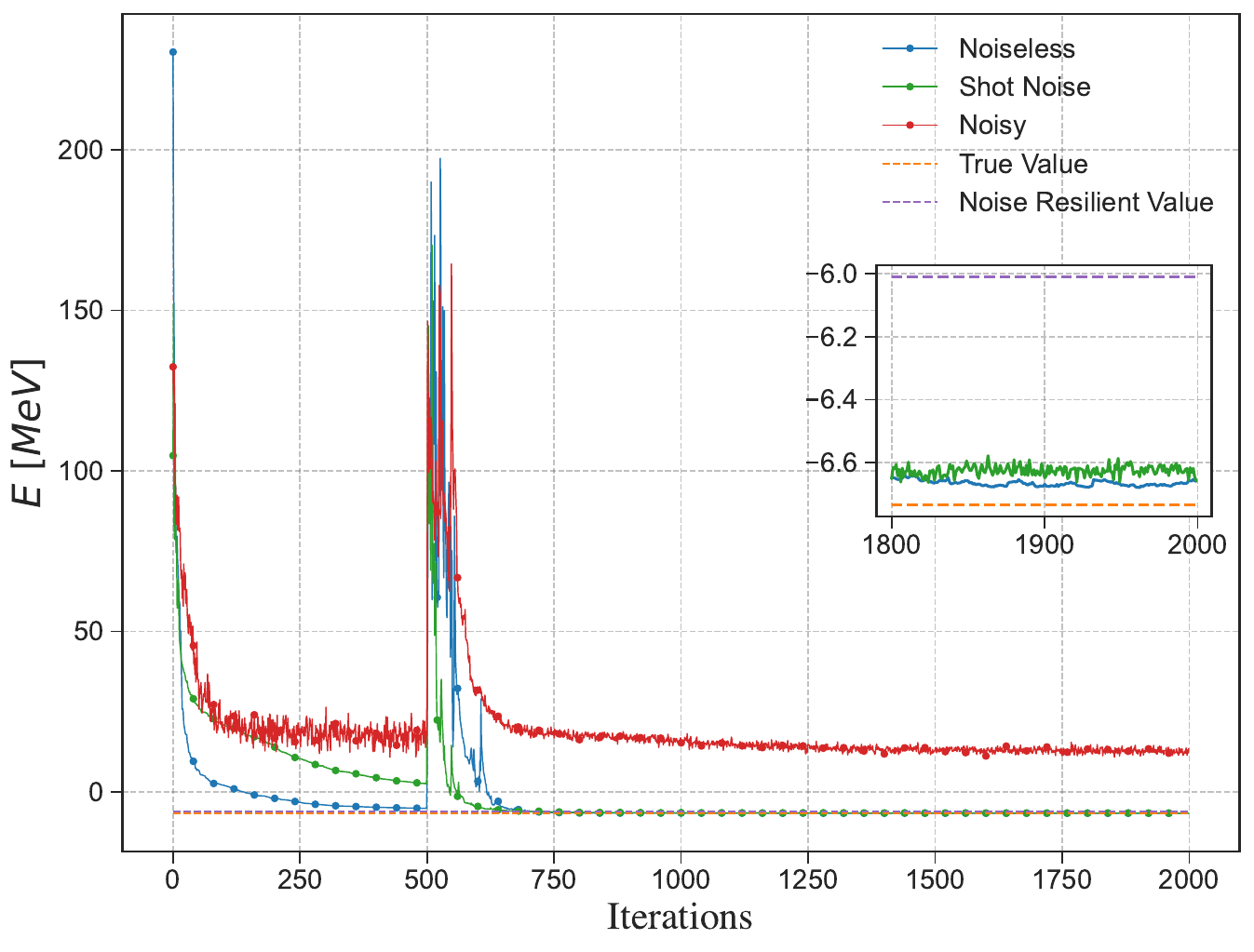}\\ (b)
    \end{subfigure} 
    \caption{Quantum simulation using the Gray encoding for the energy of the lowest $\half^+$ state for n$+^{10}$C modeled by the $V_{\rm E}$ exponential potential Eq.~\eqref{eq:Vexp} with (a) $N=8$ ($n=3$ qubits) and (b) $N=16$ ($n=4$ qubits), with $K=3$ and for different types of simulations detailed in Sec. \ref{sec:simtype}, as compared to the theoretical energy $E_{\rm th}^{K=3}$ labeled as ``True Value" (for the Hamiltonian parameters and $E_{\rm th}^{K=3}$, see Table~\ref{tab:Cdata}).  
    The inset plots are the last 200 iterations.}
    \label{fig:n_10C_plot}
\end{figure*}

\begin{figure*}[bht]
    \begin{subfigure}[b]{0.49\textwidth}
    \centering
    \includegraphics[width=\textwidth]{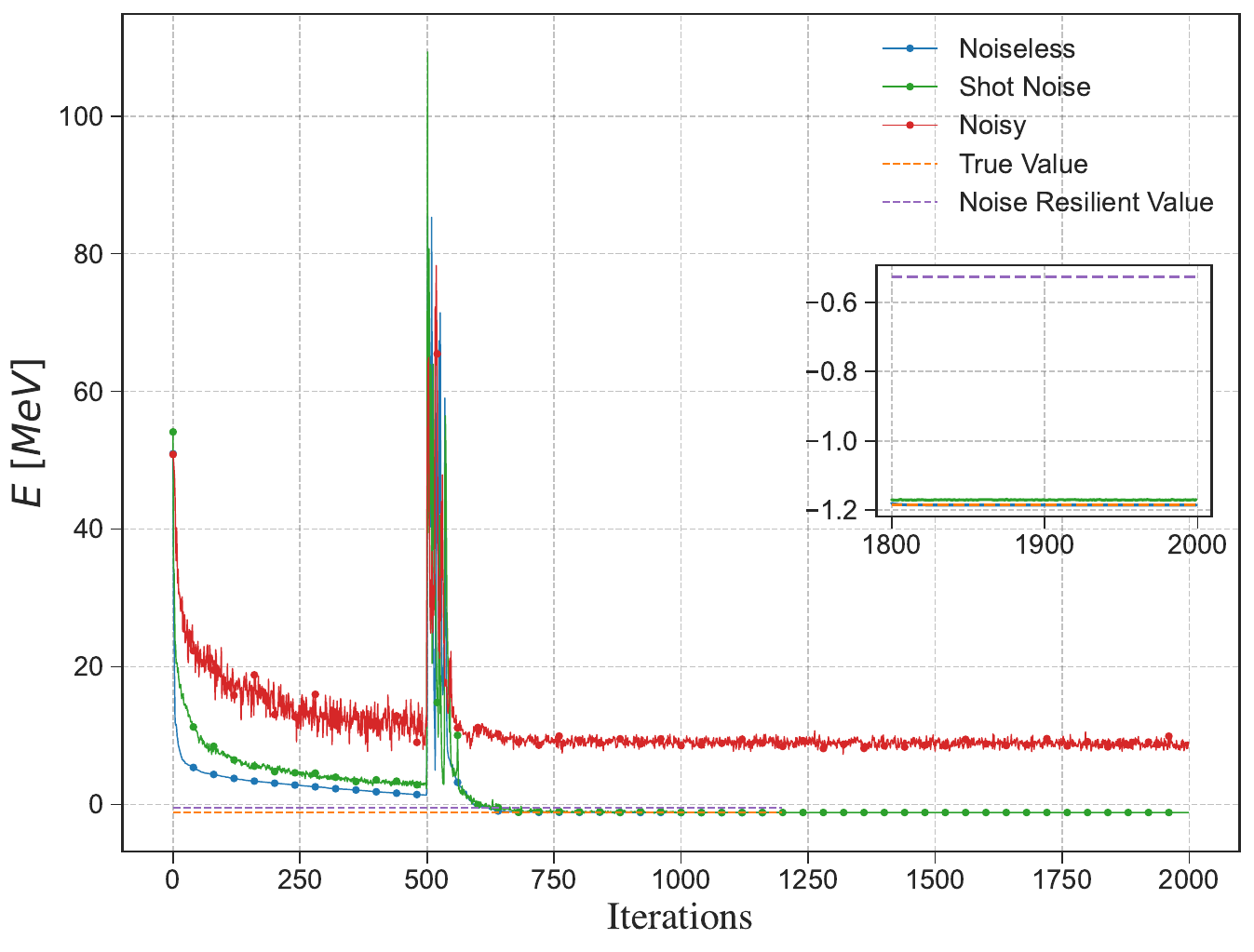}\\ (a)
    \end{subfigure} 
    \hfill
    \begin{subfigure}[b]{0.49\textwidth}
    \includegraphics[width=\textwidth]{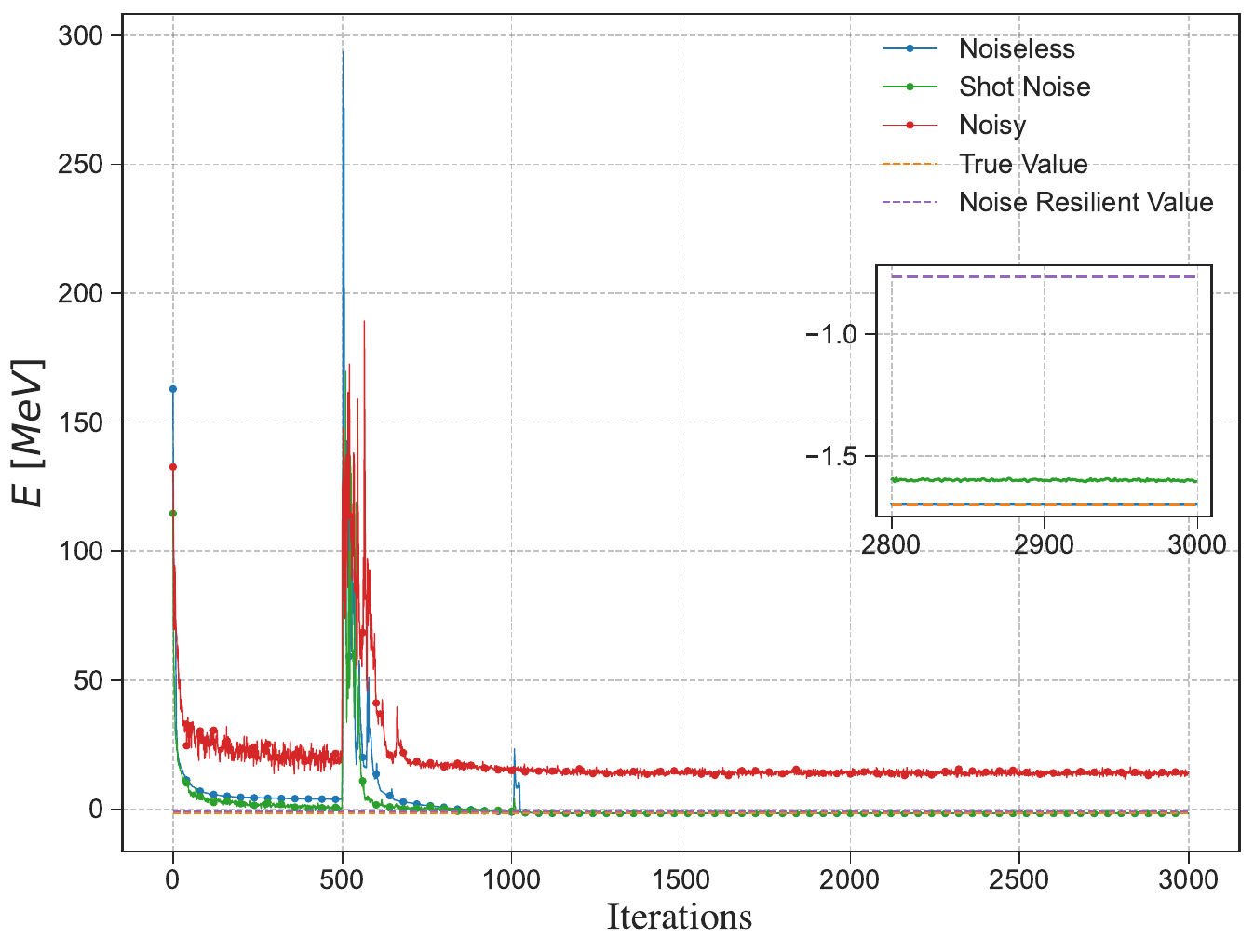}\\ (b)
    \end{subfigure} 
    \caption{
    Quantum simulation using the Gray encoding for the energy of the lowest $\half^+$ state for n$+^{12}$C modeled by the $V_{\rm E}$ exponential potential Eq.~\eqref{eq:Vexp} with (a) $N=8$ ($n=3$ qubits) and (b) $N=16$ ($n=4$ qubits), with $K=3$ and for different types of simulations detailed in Sec. \ref{sec:simtype}, as compared to the theoretical energy $E_{\rm th}^{K=3}$ labeled as ``True Value" (for the Hamiltonian parameters and $E_{\rm th}^{K=3}$, see Table~\ref{tab:Cdata}).  
    The inset plots are the last 200 iterations.
    }
    \label{fig:n_12C_plot}
\end{figure*}

Before running simulations, the hyperparameter $K$ of Eq.~\eqref{eq:Vcut} needs to be chosen. From Figs.~\ref{fig:choosing-truncation-k-local}b \& c, one determines that $K=1,\, 2$, and $7(8)$ represent reasonable approximations to the local potential with $\hbar \omega  =12(16) $ MeV at small distances, without the possibility of introducing spurious bound states. 
Indeed, as seen in Figs.~\ref{fig:choosing-truncation-k-local}b \& c, truncations at other $K$ values yield potentials that are attractive around 3.5 fm (the $K=3$ approximation becomes largely attractive beyond 5 fm). The quantum simulations for n+$\alpha$ discussed below are performed for $K=1$ and $2$, since the low-$K$ regime with the Gray encoding  is expected to  benefit largely from the use of commuting sets (see Sec. \ref{sec:tradeoffs}) on existing and far-term quantum devices.

\begin{figure*}[bht]
    \begin{subfigure}[b]{0.49\textwidth}
    \centering
    \includegraphics[width=\textwidth]{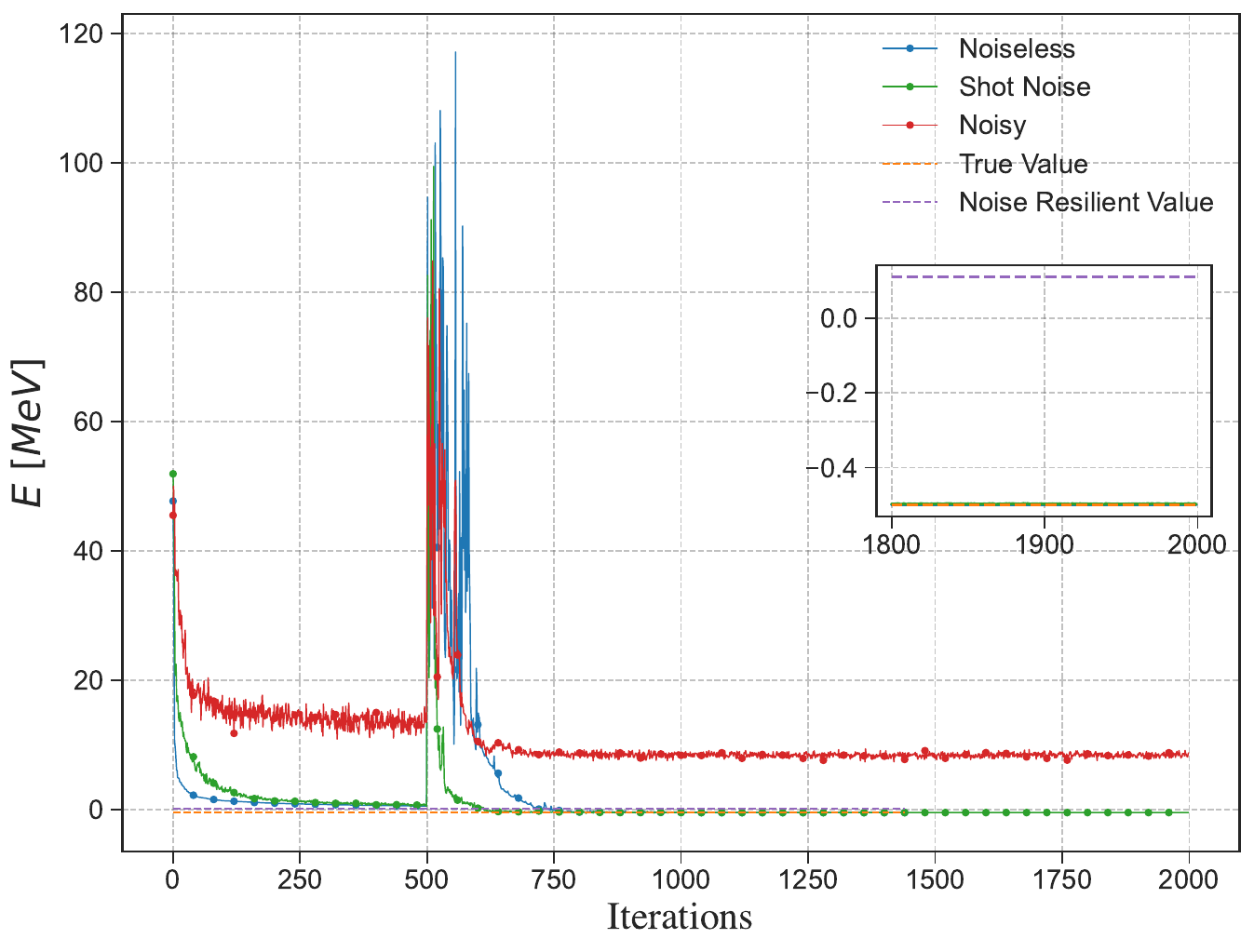}\\ (a)
    \end{subfigure} 
    \hfill
    \begin{subfigure}[b]{0.49\textwidth}
    \includegraphics[width=\textwidth]{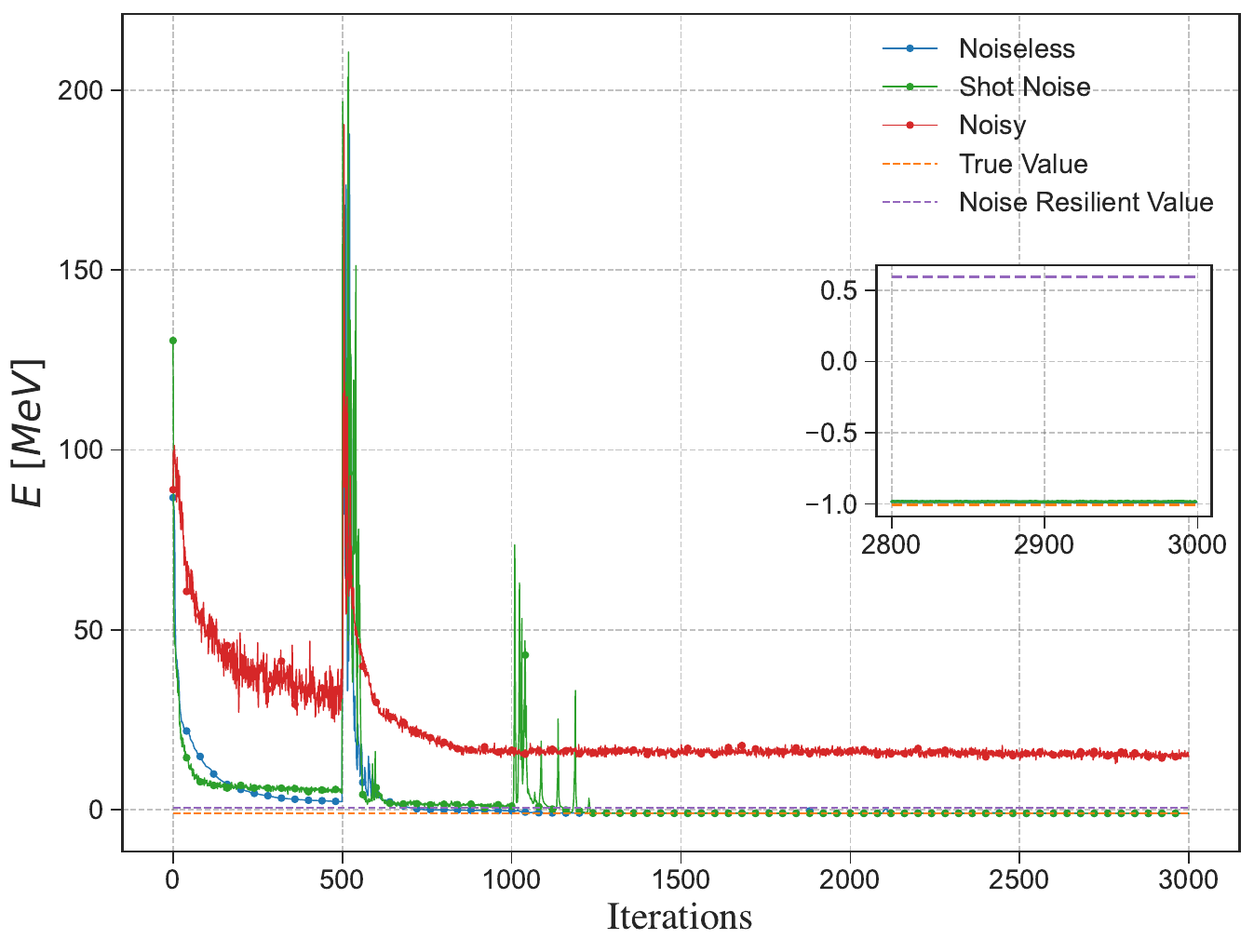}\\ (b)
    \end{subfigure} 
    \caption{Quantum simulation using the Gray encoding for the energy of the lowest $\half^+$ state for n$+^{14}$C modeled by the $V_{\rm E}$ exponential potential Eq.~\eqref{eq:Vexp} with (a) $N=8$ ($n=3$ qubits) and (b) $N=16$ ($n=4$ qubits), with $K=3$ and for different types of simulations detailed in Sec. \ref{sec:simtype}, as compared to the theoretical energy $E_{\rm th}^{K=3}$ labeled as ``True Value" (for the Hamiltonian parameters and $E_{\rm th}^{K=3}$, see Table~\ref{tab:Cdata}).  
    The inset plots are the last 200 iterations.}
    \label{fig:n_14C_plot}
\end{figure*}

\subsection{Description of the quantum simulations}
\label{sec:simtype}
The ansatz choice for the different simulations is given in Sec.~\ref{sec:ansatz}. For the n+C systems with an exponential potential, we perform simulations  with $N=8$ and $N=16$, with $K=3$ (7-diagonal matrix). For the n+$\alpha$ systems, we use $N=8$, with $K=1$ and $K=2$. The type of the simulations are as follows:

\begin{itemize}
    \item Exact diagonalization (referred to as ``true value"): We diagonalize the matrix exactly (using classical methods) and find the minimum eigenvalue, which is used to validate the
    quantum simulation outcomes.
    
    \item Noiseless simulation: We use a perfect noiseless state-vector simulator to simulate ideal behavior.
    
    \item Shot-noise simulation: We perform a shot-noise simulation using a variable number of shots (see Appendices \ref{app:sim_details_carbon} and \ref{app:sim_details_abi}). The shot-noise simulation is useful in the scenario where the quantum device is perfect and noiseless. However, the shot noise from the final measurements is unavoidable. This method gives a good estimate of performance in the far-term error-corrected regime. 
    
    \item Noisy simulation: We perform a noisy quantum simulation with noise models from existing quantum devices. Specifically, in this study, the cost function estimates use a fake IBMQ backend {\tt ibm\_manila}. This method gives a good estimate of performance in the near-term NISQ regime. 
    
    \item Noise-resilient estimation: We use the final parameters $\btheta$ 
    from the noisy simulation and calculate the expectation value of the Hamiltonian on a classical machine to find a noiseless estimate for the minimum energy, which we refer to as the noise-resilient (NR) value. The idea of noise-resilient training was first introduced in \cite{Sharma_2020}. We find that the algorithm exhibits noise resilience; i.e., the NR value is more accurate than the noisy simulation outcome. In other words, training is still possible in a noisy scenario.
\end{itemize}

In all examples, we use a combination of the Simultaneous Perturbation Stochastic Approximation (SPSA) algorithm \cite{Spall1992}, and Gradient Descent to estimate the gradient (see Appendices \ref{app:sim_details_carbon} and \ref{app:sim_details_abi} for specific details). The SPSA method provides an unbiased estimator of the gradient with a runtime independent of the number of parameters. We find that deviation from the true value caused the noise introduced by SPSA to increase with the problem size. However, we find that it is useful to quickly obtain a solution close to the optimal value. From here, we use Gradient Descent to improve the accuracy of
the optimal solution.

For the n+C system, we take as input $N$, $K$, $\hbar\omega$, $V_0$, and $c$. For the n+$\alpha$ system, we take as input $N$, $K$, $\hbar\omega$, and the list of coefficients $\{v_k\}_{k=0}^K$. 

We note that in all simulations, the overall Hamiltonian matrix is constructed and then the expectation value is calculated as an inner product.  Nonetheless, for runs on a quantum device, further advantages will stem from using qubit-wise or distance-grouped commuting sets. For example, in the case of the Gray encoding, for $N=16$ (four qubits) and $K=3$, there are $88$ Pauli strings, $19$ QC sets, and $10$ DGC sets (e.g., see Table~\ref{tab:GSubsetsVsk}).

In our simulations, we use a warm-start initialization for the ansatz parameters, inspired by perturbation theory. For the exponential potential, see Eq.~\eqref{eq:Vexp}, and for the local potentials deduced \textit{ab initio}, see Eqs. \eqref{eq:Vhw12} and \eqref{eq:Vhw16}, each progressive diagonal is scaled by an increasingly smaller coefficient. As a result, we expect each additional diagonal to alter the eigenvalue to a lesser extent than the previous diagonals. Inspired by this, we run the simulations for a lower $K$ value using fewer iterations and shots. This enables the optimization to quickly reach an approximate solution. Then, we use the endpoint of this simulation as the start for the full-scale simulation. We find that, in practice, this leads to finding the optimal solution in fewer iterations. In Figs.~\ref{fig:n_10C_plot}-\ref{fig:abi-n-3-16}, the large peaks in energy during the optimization occur due to this switch from a lower $K$ to the required value. We note that switching from SPSA to gradient descent can also introduce these peaks. However, the initial learning rate of the gradient descent stage can be tuned to remove this source.

\subsection{Quantum simulations for neutron-Carbon dynamics}

In this section we present the results of the different simulations for n-Carbon dynamics. The specific details, including the number of iterations, type of gradient estimator, and number of shots, can be found in Appendix~\ref{app:sim_details_carbon}. We provide simulations that, for the first time, show the efficacy of the Gray code for a Hamiltonian matrix beyond the tridiagonal case, that is, for $K=1$ (bandwidth of 3). With the Gray code, we can utilize only three and four qubits to simulate model spaces of $N=8$ and 16 basis states, respectively. This provides acceptable results that closely agree with the energy for the $K=3$ approximation, which lies near the exact theoretical energy as discussed above (see Fig.~\ref{fig:nC_EvsN} for $N=8$ and $N=16$). The quantum simulations for n$+^{10}$C, n$+^{12}$C, and n$+^{14}$C  are shown in Figs.~\ref{fig:n_10C_plot}-\ref{fig:n_14C_plot}, and the case of  n$+^{14}$C is compared to the one-hot encoding in Fig.~\ref{fig:n_14C_plot_OH}. 

We summarize the energy estimates from the various quantum simulations in Table~\ref{tab:CdataSimulations}.

\begin{table*}[th]
\centering
\begin{tabular}{c|c|c|c|c|c|c|c|c|c}
    & $E_{\rm th}$ & Encoding & Figure & $N$ & $E_{{\rm th}}^{K=3}$ & $E^{K=3}_{\rm NL}$ & $E^{K=3}_{\rm shot}$ & $E^{K=3}_{\rm noisy}$ & $E^{K=3}_{\rm NR}$\\
    & [MeV] & & & & [MeV] & [MeV] & [MeV] & [MeV]& [MeV]\\
    \hline
    \hline
    \multirow{2}{*}{n+$^{10}$C} & \multirow{2}{*}{$-6.78$} & \multirow{2}{*}{Gray} & \multirow{2}{*}{\ref{fig:n_10C_plot}} & $8$ & $-6.5364$ & $-6.5364 \pm 2 \times 10^{-10}$ & $-6.5305 \pm 0.0005$ & $3.5 \pm 0.5$ & $-5.9$ \\
    \cline{5-10}
    & & & & $16$ & $-6.7346$ & $-6.668 \pm 0.007$ & $-6.62 \pm 0.02$ & $12.6 \pm 0.7$ &  $-6.0$\\
    \hline
    
    \multirow{2}{*}{n+$^{12}$C} & \multirow{2}{*}{$-1.86$} & \multirow{2}{*}{Gray} & \multirow{2}{*}{\ref{fig:n_12C_plot}} & $8$ & $-1.18495$ & $-1.184993 \pm 1 \times 10^{-6}$ & $-1.1701 \pm 0.0006$ & $8.7 \pm 0.5$ & $-0.5$ \\
    \cline{5-10}
    & & & & $16$ & $-1.70020$ & $-1.6973 \pm 0.0001$ & $-1.599 \pm 0.003$ & $13.9 \pm 0.7$ & $-0.8$\\
    \hline
    
    \multirow{3}{*}{n+$^{14}$C} & \multirow{3}{*}{$-1.22$} & \multirow{2}{*}{Gray} & \multirow{2}{*}{\ref{fig:n_14C_plot}} & $8$ & $-0.49963$ & $-0.4996 \pm 6 \times 10^{-10}$ & $-0.4966 \pm 0.0005$ & $8.4 \pm 0.3$ & $0.1$ \\
    \cline{5-10}
    & & & & $16$ & $-1.0070$ & $-0.9860 \pm 0.0005$ & $-0.982 \pm 0.002$ & $15.2 \pm 0.5$ & $0.6$\\
    \cline{3-10}

    & & One-hot & \ref{fig:n_14C_plot_OH} & $8$ & $-0.49963$ & $-0.4987 \pm 0.0001$ & $-0.49 \pm 0.03$ & $9.0 \pm 0.6$ & $0.7$ \\
    \end{tabular}
    \caption{Simulation energy for the lowest $\half^+$ state for each neutron-Carbon system, for the $K=3$ approximation with $N=8$ or $N=16$, corresponding to Figs. \ref{fig:n_10C_plot}-\ref{fig:n_14C_plot_OH}: the true value $E_{\rm th}^{K=3}$;  $E^{K=3}_{\rm NL}$ from the noiseless simulations, $E^{K=3}_{\rm shot}$ from the shot-noise simulations, and $E^{K=3}_{\rm noisy}$ from the noisy simulations, reported as the mean and standard deviation of the last 100 iterations; and the NR value $E^{K=3}_{\rm NR}$. The theoretical eigenenergy $E_{\rm th}$ of the exact potential is also shown (cf. Table~\ref{tab:Cdata}).}
    \label{tab:CdataSimulations}
\end{table*}

In particular, in all cases the noiseless and shot-noise simulations with the Gray code yield results with very small errors, mostly, $10^{-10}$-$10^{-3}$.
Importantly, they reproduce the true value within $0$-$2$\%, with the only exception being the $6$\% deviation in $E^{K=3}_{\rm shot}$ for $^{12}$C and $N=16$.  

A very significant result is that the noisy simulations, which are expected to simulate runs on current NISQ processors, provide very reasonable energy estimates through the use of the NR value. In particular, we find that the algorithm exhibits noise resilience; i.e., the NR value $E^{K=3}_{\rm NR}$ is more accurate than the noisy simulation outcome and training of the wavefunctions parameters is still possible in a noisy scenario. Furthermore, for these n-C systems, we show that the $E^{K=3}_{\rm NR}$ energies differ from the corresponding true value only by $600$-$950$ keV, which surpasses the accuracy of many nuclear models. The only exception is the four-qubit simulation for n$+^{14}$C (Fig.~\ref{fig:n_14C_plot}), which predicts an unbound $\half^+$ state at $0.59$ MeV for the weakly bound state at $-1.0$ MeV. While this is still a reasonable estimate, weakly bound states require larger model spaces to achieve convergence, as discussed above, which implies the need for larger number of qubits. This means that special care needs to be taken in the quantum simulations of such systems, including n$+^{16}$C and n$+^{18}$C, when performed on NISQ devices. One way to improve this is to try different ansatz structures, or a larger model space, which remains to be shown in future work. 

We note here that for the example of n+$^{10}$C, with $N=8$ and $K=3$, in addition to the noise resilient final estimate, we report the mean and standard deviation of the noiseless estimates of the last 100 noisy iterations. In other words, we use the parameters from last 100 noisy iterations, calculate the noiseless estimate, and report the mean and standard deviation to be $-6.05 \pm 0.07$. This is $\sim 2\sigma$ away from the $E^{K=3}_{NR}$ value reported in Table~\ref{tab:CdataSimulations} calculated from the last iteration only, which is very reasonable. This further solidifies the idea that the noise resilient method offers a much better estimate as compared to the noisy estimate. However, in any large-scale experiment, we would not recommend running the noiseless simulation for hundreds of iterations. This is because each noiseless simulation is prohibitively expensive. Instead, in this study, we use the standard deviation for n+$^{10}$C as a guidance to the number of the significant digits we report for the noise resilient estimates in all simulations.

Furthermore, even in the case of weakly bound states, the Gray code with three qubits is superior to the one-hot encoding with eight qubits, as illustrated in Fig.~\ref{fig:n_14C_plot_OH} and Table~\ref{tab:CdataSimulations}. Indeed, compared to the Gray-code case, the one-hot simulations are much slower and show worse performance in the presence of noise. In particular, the one-hot simulation yields larger errors [e.g., by six (two) orders of magnitude for the noiseless (shot noise) simulations], as well as a deviation for $E^{K=3}_{\rm NR}$ from the true value that is twice as large as the corresponding Gray-code estimate. In addition, as shown in Figs. \ref{fig:OneHotAnsatz} and \ref{fig:BR_GR_Ansatz} as well as in Tables \ref{tab:details_OHencodings} and \ref{tab:details_BGencodings}, the one-hot ansatz utilizes 13 two-qubit gates compared to only 8 two-qubit (CNOT) gates in the Gray encoding ansatz with $L=4$ used in our simulations. This suggests that the Gray encoding and the use of fewer qubits are indeed highly advantageous for nuclear problems that achieve convergence in larger model spaces, with the case of weakly bound systems presented here being an illustrative example. 

\begin{figure}
\includegraphics[width=\columnwidth]{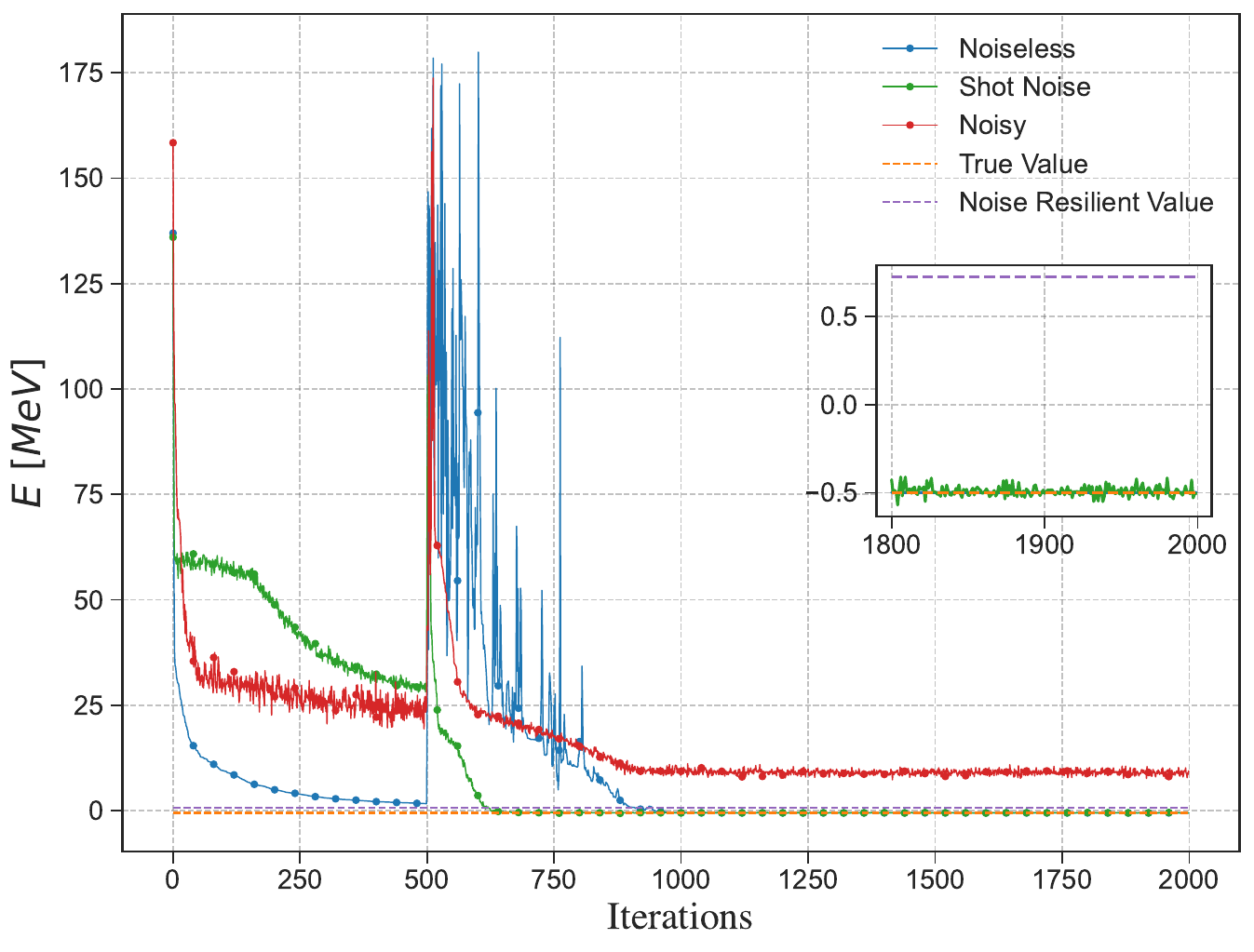}
    \caption{Simulation using the one-hot encoding for the energy of the lowest $\half^+$ state for n$+^{14}$C modeled by the exponential potential Eq.~\eqref{eq:Vexp} with $N=8$ ($n=8$ qubits) and $K=3$. The different types of simulations are detailed in Sec. \ref{sec:simtype}. The inset plots the last $200$ iterations.}
    \label{fig:n_14C_plot_OH}
\end{figure}

\begin{figure*}[th]
    \begin{subfigure}[b]{0.49\textwidth}
    \centering
    \includegraphics[width=\textwidth]{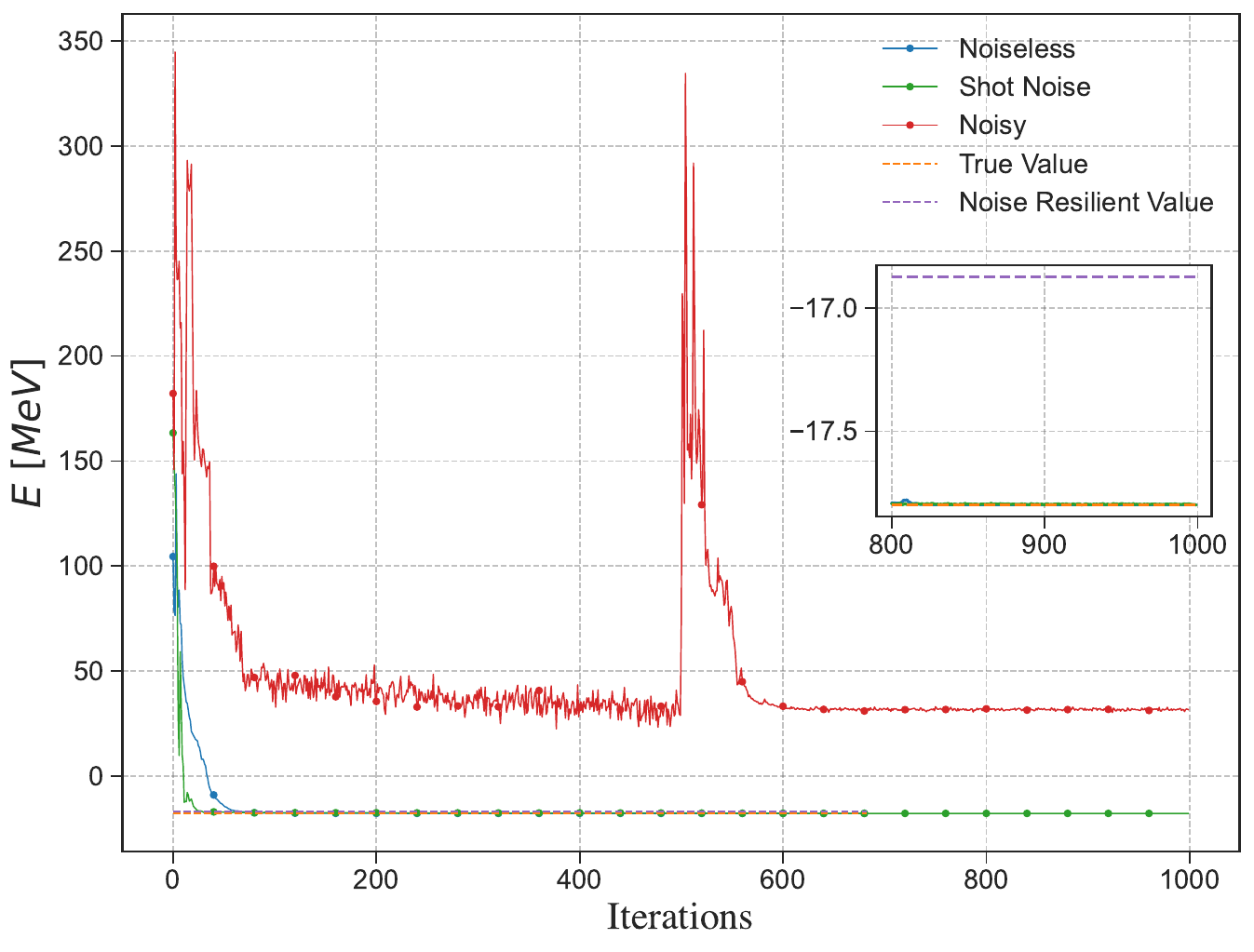}\\ (a)
    \end{subfigure} 
    \hfill
    \begin{subfigure}[b]{0.49\textwidth}
    \includegraphics[width=\textwidth]{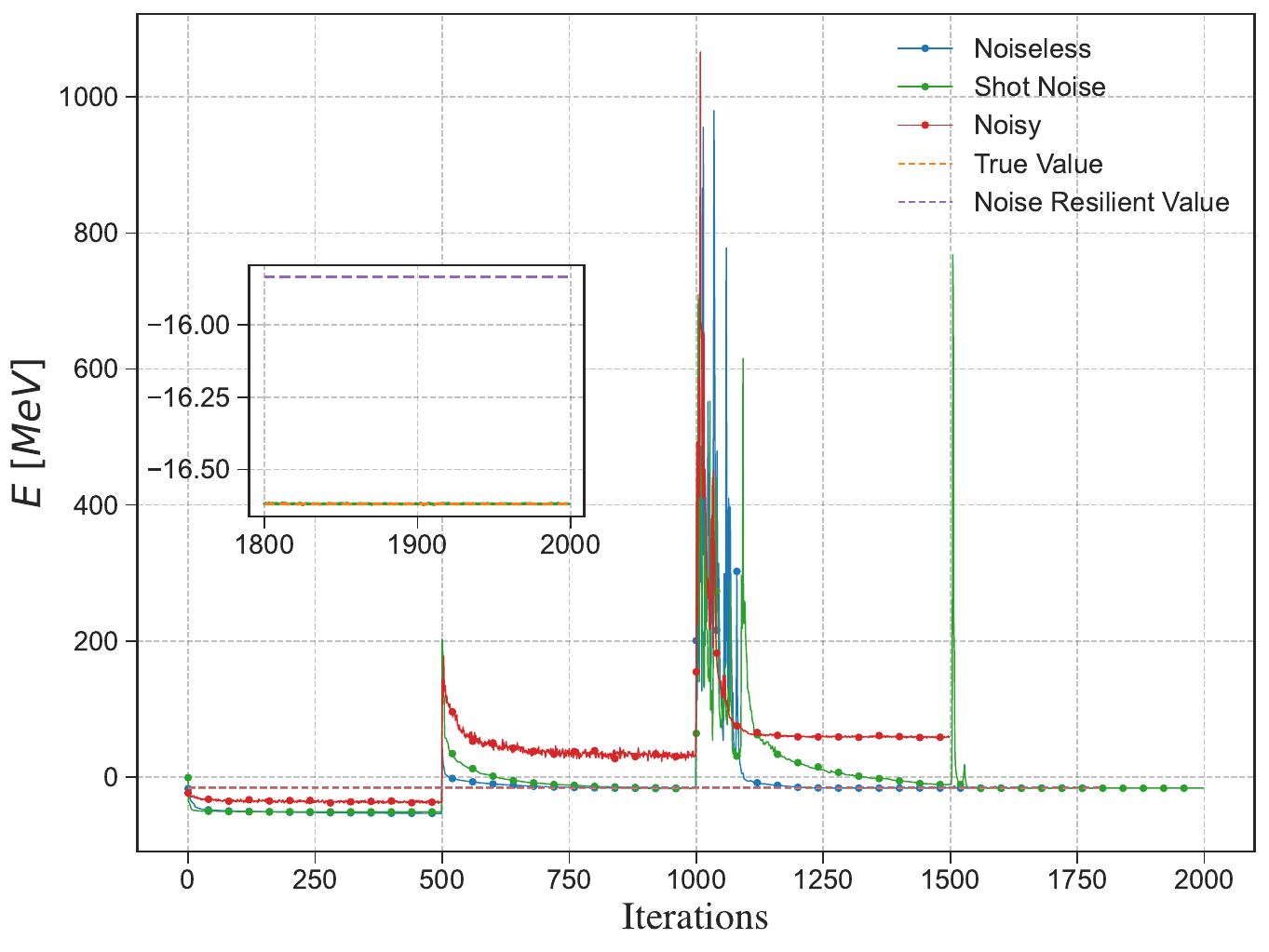}\\ (b)
    \end{subfigure} 
    \caption{Simulation with the Gray encoding for the n-$\alpha$ \textit{ab initio} potential for $\hbar \omega=12$ MeV Eq.~\eqref{eq:Vhw12} with $N=8$ ($n=3$), and (a) $K=1$ and (b) $K=2$. The different types of simulations are detailed in Sec. \ref{sec:simtype}. The inset plots the last 200 iterations.}
    \label{fig:abi-n-3-12}
\end{figure*}

\subsection{Quantum simulations with the Gray encoding for n\texorpdfstring{$+\alpha$}{+alpha} using \textit{ab initio} optical potentials}

In this section, we present the results of the  simulations for the n+$\alpha$ dynamics. The specific details, including the number of iterations, type of gradient estimator, and number of shots, can be found in Appendix~\ref{app:sim_details_abi}. The quantum simulations for n$+\alpha$ are carried out for model spaces of $N=8$ (Fig.~\ref{fig:abi-n-3-12}) and $N=16$ (Fig.~\ref{fig:abi-n-4-12}) for $\hbar\omega=12$ MeV and of $N=8$ (Fig.~\ref{fig:abi-n-3-16}) for $\hbar\omega=16$ MeV. We show the cases of $K=2$ for a potential of $\mathcal{O}(r^4)$ and $K=3$ for a potential of $\mathcal{O}(r^6$).

We summarize the energy estimates from the various quantum simulations in Table~\ref{tab:ab_initio_dataSimulations}. 
The energies $E_{\rm th}^{K}$ from the exact diagonalization for $K=1$ and $2$  converge within the first four to five digits with the increasing model space from $N=8$ to $N=16$, as illustrated in Table~\ref{tab:ab_initio_dataSimulations} for $\hbar\omega=12$ MeV (for $\hbar\omega=16$ MeV, $E_{\rm th}^{K=1}=-20.77$ MeV and $E_{\rm th}^{K=2}=-18.95$ MeV for $N=16$). Hence, the final estimates are based on the $N=8$ results, where the uncertainties are estimated for a 14\% variation of the \hw~values ($\hbar\omega=12$-$16$ MeV), the same range that yields a total cross section for the neutron scattering on $^4$He that is converged and in agreement with experiment, as discussed in Ref.~\cite{BurrowsL21}. 
In addition, the $K=1$ and $K=2$ approximation energies are very close to the corresponding $E_0$ eigenenergy of the original non-local \textit{ab initio} optical potential. In fact, across the \hw~range, the  energy of the lowest $\half^+$ orbit is estimated at $-19.3(1.5)$ MeV for $K=1$ and at $-17.8(1.2)$ MeV for $K=2$, both of which agree within the uncertainties with the non-local estimate of $-19.8(1.0)$ reported above. Exactly the same energy estimates are obtained by the noiseless and shot simulations ($-19.3\pm 1.5$ MeV for $K=1$ and at $-17.8\pm 1.2$ MeV for $K=2$). We note that the errors that stem from these simulations (of order of a few keV) are inconsequential compared to the ones associated with the \hw~variation ($\sim 1-1.5$ MeV). Importantly, while the noisy simulations yield unacceptably large energies, the corresponding NR energies agree with the shot-noise outcomes within $1\sigma$: $-18.6(1.7)$ MeV for $K=1$ and $-17.1(1.2)$ MeV for $K=2$, while leading to slightly larger or comparable error bars. This suggests that deep bound states described by an \textit{ab initio} deduced NA optical potential can be reasonably well approximated by tri- to five-diagonal potentials, and in turn, can be successfully simulated on far-term error-corrected devices (practically yielding the true result) and even on NISQ processors (yielding energies within $1\sigma$ across a 14\% \hw~variation).

\begin{table*}[t]
\centering
\begin{tabular}{c|c|c|c|c|c|c|c|c|c}
    $\hbar \omega$ & $E_0$ & $N$ & Figure & $K$ & $E_{\rm th}^K$ & $E_{\rm NL}$ & $E_{\rm shot}$ & $E_{\rm noisy}$ & $E_{\rm NR}$\\
    & [MeV] & & & & [MeV] & [MeV] & [MeV] & [MeV] & [MeV]\\
    \hline
    \hline
    \multirow{4}{*}{$12$} & \multirow{4}{*}{$-18.85$} & \multirow{2}{*}{$8$} & \multirow{2}{*}{\ref{fig:abi-n-3-12}} & $1$ & $-17.7986$ & $-17.7986 \pm 4 \times 10^{-10}$ & $-17.796 \pm 0.001$ & $31.6 \pm 0.4$ & $-16.9$ \\
    \cline{5-10}
    & & & & $2$ & $-16.6190$  & $-16.6189 \pm 1.7 \times 10^{-13}$ & $-16.618 \pm 0.001$ & $58.7 \pm 0.7$ & $-15.8$ \\
    \cline{3-10}

    & & \multirow{2}{*}{$16$} & \multirow{2}{*}{\ref{fig:abi-n-4-12}} & $1$ & $-17.7987$ & $-17.7852 \pm 0.0008$ & $-17.791 \pm 0.007$ & $68.8 \pm 0.7$ & $-16.8$ \\
    \cline{5-10}
    & & & & $2$ & $-16.6191$  & $-16.615 \pm 0.005$ & $-16.60 \pm 0.02$ & $162 \pm 2$ & $-14.9$ \\
    \hline

    \multirow{2}{*}{$16$} & \multirow{2}{*}{$-20.84$} & \multirow{2}{*}{$8$} & \multirow{2}{*}{\ref{fig:abi-n-3-16}} & $1$ & $-20.7735$ & $-20.7735 \pm 3.2 \times 10^{-11}$ & $-20.7733 \pm 0.0005$ & $8.8 \pm 0.3$ & $-20.3$ \\
    \cline{5-10} 
    
    & & & & $2$ & $-18.9470$ & $-18.9470 \pm 1.8 \times 10^{-13}$ & $-18.9470 \pm 0.0006$ & $43.2 \pm 0.6$ & $-18.3$ \\
    
    \end{tabular}
    \caption{Simulation energy for the lowest $\half^+$ orbit of the n$-\alpha$ \textit{ab initio} deduced local potential, using the Gray code with $N=8$ or $N=16$, corresponding to Figs. \ref{fig:abi-n-3-12}-\ref{fig:abi-n-3-16}: the true value $E_{\rm th}^{K}$;  $E_{\rm NL}$ from the noiseless simulations, $E_{\rm shot}$ from the shot-noise simulations, and $E_{\rm noisy}$ from the noisy simulations, reported as the mean and standard deviation of the last 100 iterations; and the NR value $E_{\rm NR}$. The exact theoretical eigenenergy $E_{0}$ of original non-local \textit{ab initio} potential is also shown.}
    \label{tab:ab_initio_dataSimulations}
\end{table*}

\begin{figure*}[bht]
    \begin{subfigure}[b]{0.49\textwidth}
    \centering
    \includegraphics[width=\textwidth]{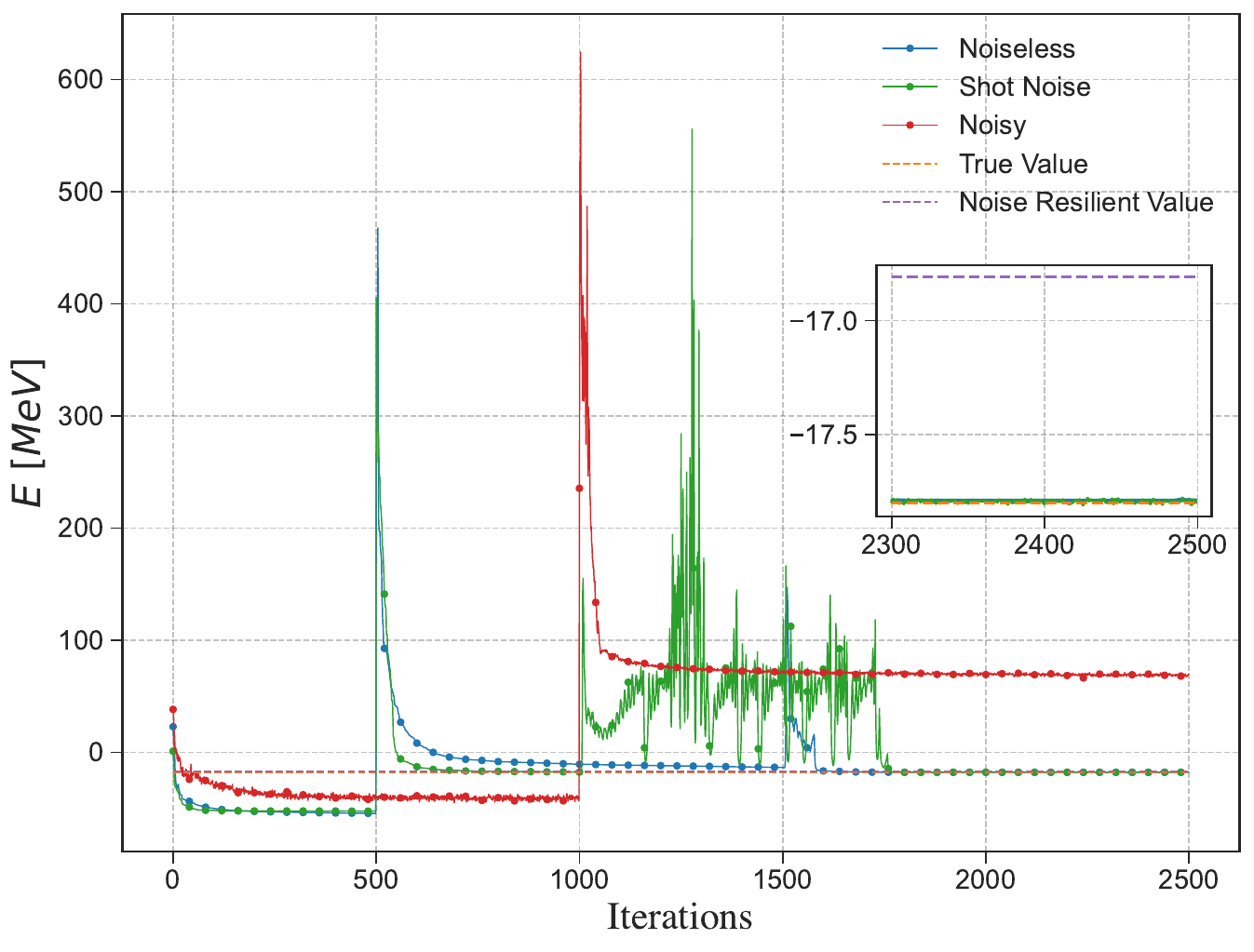}\\ (a)
    \end{subfigure} 
    \hfill
    \begin{subfigure}[b]{0.49\textwidth}
    \includegraphics[width=\textwidth]{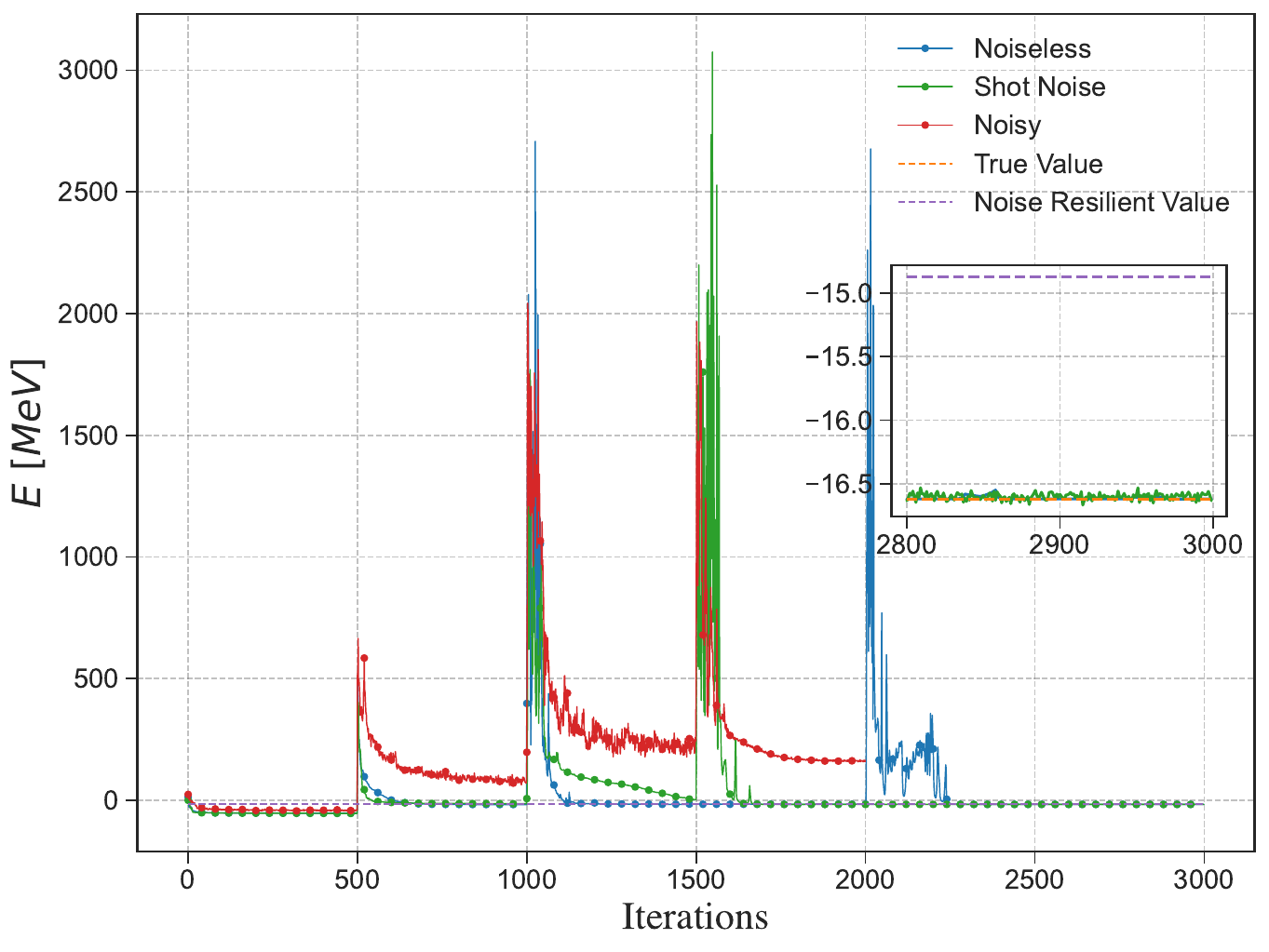}\\ (b)
    \end{subfigure} 
    \caption{Simulation with the Gray encoding for the n-$\alpha$ \textit{ab initio} potential for $\hbar \omega=12$ MeV Eq.~\eqref{eq:Vhw12} with $N=16$ ($n=4$), and (a) $K=1$ and (b) $K=2$. The different types of simulations are detailed in Sec. \ref{sec:simtype}. The inset plots the last 200 iterations.}
    \label{fig:abi-n-4-12}
\end{figure*}

\begin{figure*}[bht]
    \begin{subfigure}[b]{0.49\textwidth}
    \centering
    \includegraphics[width=\textwidth]{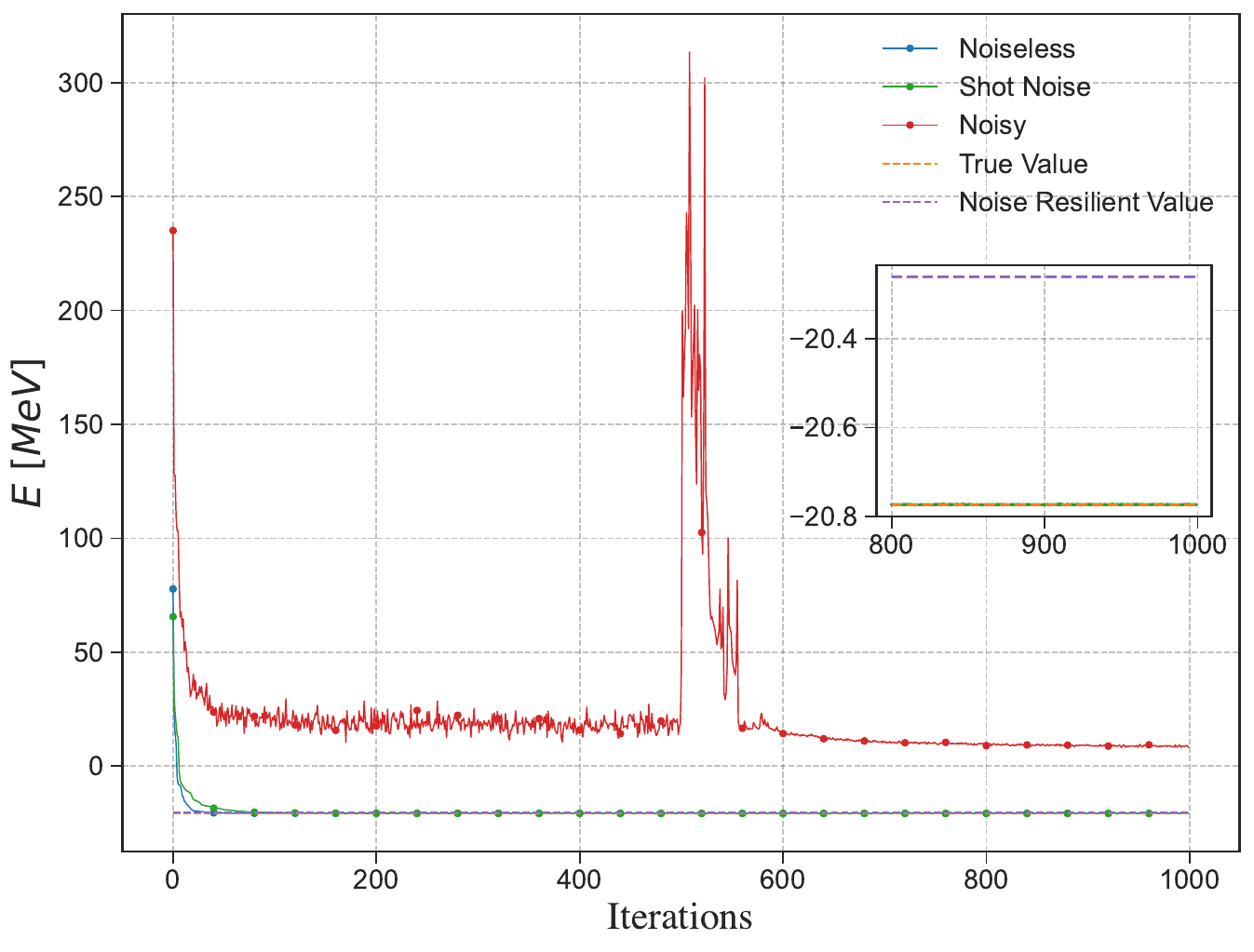}\\ (a)
    \phantomsubcaption
    \end{subfigure} 
    \hfill
    \begin{subfigure}[b]{0.49\textwidth}
    \includegraphics[width=\textwidth]{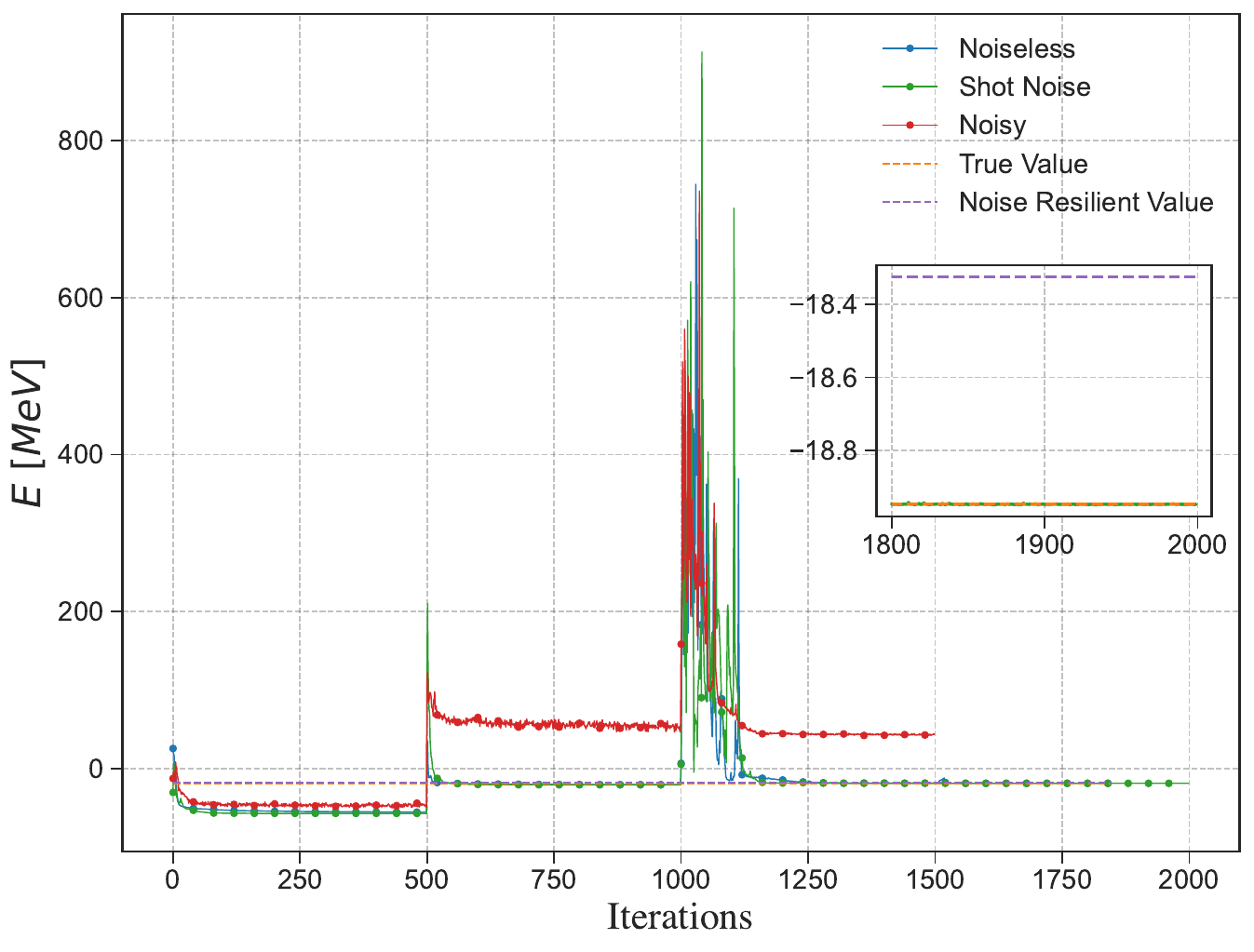}\\ (b)
    \end{subfigure} 
    \caption{Simulation for the n-$\alpha$ \textit{ab initio} potential for $\hbar \omega=16$ MeV Eq.~\eqref{eq:Vhw16} with $N=8$, and (a) $K=1$ and (b) $K=2$. The different types of simulations are detailed in Sec. \ref{sec:simtype}. The inset plots the last 200 iterations.}
    \label{fig:abi-n-3-16}
\end{figure*}

\subsection{Comparing QC and DGC schemes}
In this section we compare the efficacy of QC and DGC schemes using both the Gray and binary encodings. We perform both noiseless and shot noise simulations for all cases. In the shot noise runs, we keep the total number of shots per variational run to be fixed. This is to simulate a fixed amount of ``quantum" resources. Furthermore, in contrast with the preceding simulations in this work, here we only use gradient descent with a varying learning rate scheme to estimate the gradient. This removes the effects of the SPSA noise and gives a clearer picture of the inherent shot noise. However, this also leads to results of the noiseless simulations in this section being slightly different from those reported in Table \ref{tab:ab_initio_dataSimulations}, since the noise introduced by SPSA can take the system out of a local minima and converge to a better solution.

We first compare QC and DGC using the Gray encoding (Fig.~\ref{fig:qc_dgc_abi}). We perform simulations for the n$+\alpha$ system using a model space of $N=8$, $\hbar \omega = 12$MeV, and $K=1$ for a potential of $\mathcal{O}(r^2)$. We expect the noiseless QC and DGC to be exactly the same, and the plots reflect this. However, since DGC leads to a more optimal grouping of Pauli terms, each group is allocated more shots. Thus, we expect DGC to outperform QC in the shot noise runs, which is what the data in Fig.~\ref{fig:qc_dgc_abi} indicates. The results are summarized in Table~\ref{tab:QCvs.DGC_Gray}.

\begin{table}[th]
\centering
\begin{tabular}{c|c}
    Simulation & Mean and Standard Deviation \\
    \hline
    Noiseless + QC & $-17.78384 \pm 0.00009$ \\
    Noiseless + DGC & $-17.78384 \pm 0.00009$ \\
    Shot Noise + QC & $-17.67 \pm 0.14$ \\
    Shot Noise + DGC & $-17.72 \pm 0.13$
\end{tabular}
    \caption{Comparing the QC and DGC schemes using the Gray encoding, in the case of the n$+\alpha$ system for $N=8$ (three qubits), $\hbar \omega = 12$, and $K=1$, corresponding to Figure~\ref{fig:qc_dgc_abi}. The total number of shots per variational run is $10^6$. The mean and standard deviation reported are calculated using the last $100$ iterations.}
    \label{tab:QCvs.DGC_Gray}
\end{table}

\begin{figure*}[bht]
    \includegraphics[width=0.75\textwidth]{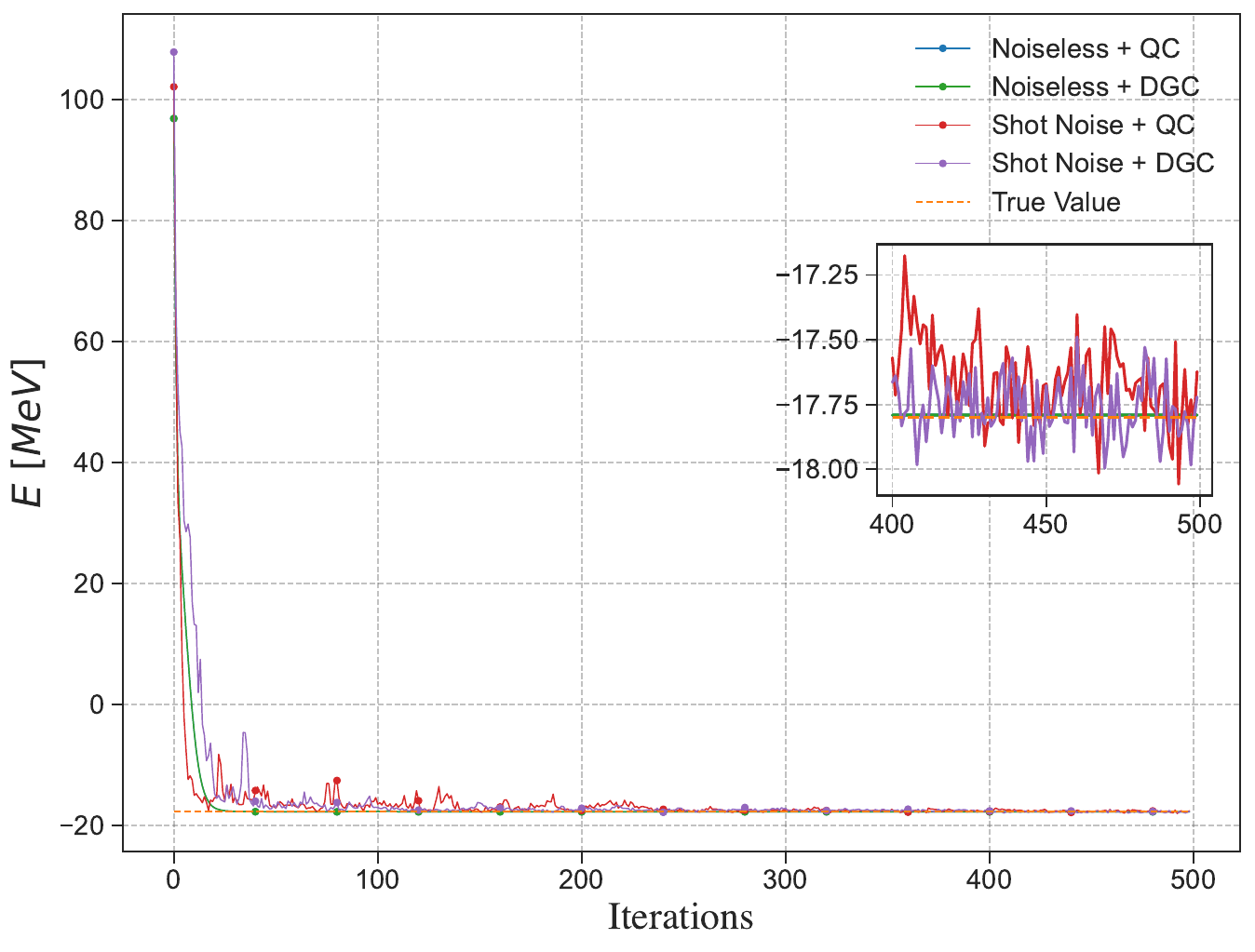}
    \caption{Comparison of the QC and DGC schemes for simulations with the Gray encoding of the n$+\alpha$ system using a model space of $N=8$, $K=1$, and $\hbar \omega = 12$MeV. The different types of simulations are detailed in Sec. \ref{sec:simtype} (the curves for the two noiseless simulations are indistinguishable). The inset plots the last 100 iterations.}
    \label{fig:qc_dgc_abi}
\end{figure*}

Next, we compare the Gray and binary encoding using the QC scheme. From the results of Sec. \ref{sec:numCG}, we expect the Gray encoding to outperform the binary encoding, and the data in Fig.~\ref{fig:qc_gvb_abi} clearly corroborates this. The outcomes of these simulations are summarized in Table~\ref{tab:QC_Gray_vs_binary}. While the Gray encoding agrees with the noiseless result within $1\sigma$, the binary encoding leads to a larger standard deviation and  at least a $2\sigma$ agreement.

\begin{table}[th]
\centering
\begin{tabular}{c|c}
    Simulation & Mean and Standard Deviation \\
    \hline
    Noiseless & $-16.5990  \pm 8 \times 10^{5}$ \\
    Shot Noise + QC + Gray & $-16.31 \pm 0.34$ \\
    Shot Noise + QC + binary & $-15.80 \pm 0.57$ \\
\end{tabular}
    \caption{Comparing the Gray and binary encodings using the QC scheme, in the case of the n$+\alpha$ system for, $N=8$ (three qubits), $\hbar \omega = 12$, and $K=2$, corresponding to Figure~\ref{fig:qc_gvb_abi}. The total number of shots per variational run is $10^6$. The mean and standard deviation reported are calculated using the last $100$ iterations. }
    \label{tab:QC_Gray_vs_binary}
\end{table}

\begin{figure*}[bht]
    \includegraphics[width=0.75\textwidth]{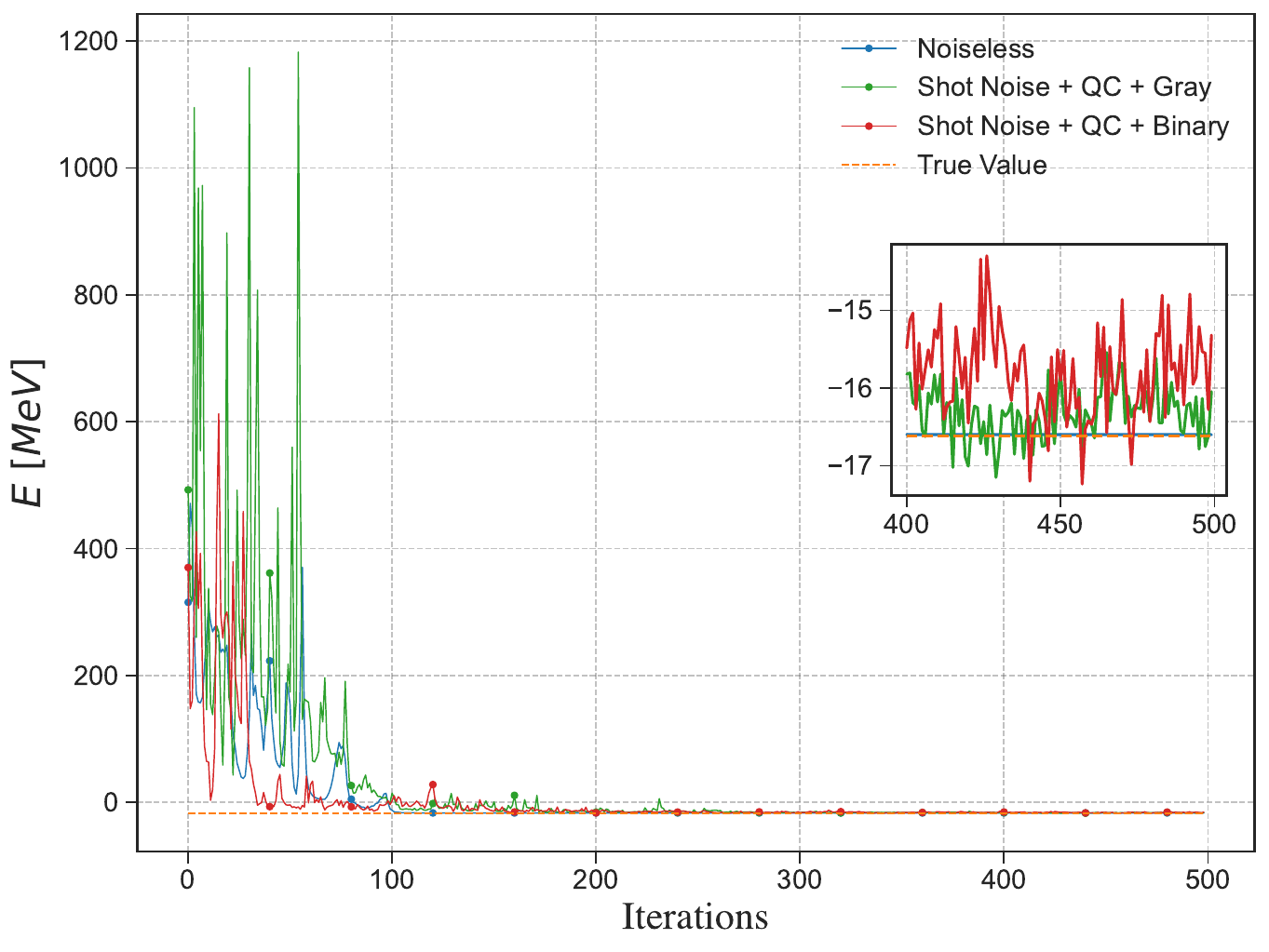}
    \caption{Comparison of the Gray and binary encodings using the QC scheme for simulations of the n$+\alpha$ system using a model space of $N=8$, $K=2$, and $\hbar \omega = 12$MeV. The different types of simulations are detailed in Sec. \ref{sec:simtype}. The inset plots the last 100 iterations.}
    \label{fig:qc_gvb_abi}
\end{figure*}

\section{Conclusions and future directions}

In this work, we developed a quantum algorithm to simulate neutron-nucleus dynamics on a quantum processor. We generalize the form of the nuclear Hamiltonian to any band-diagonal to full  matrices, which can accommodate a general central potential and a complete form of the chiral NN potential. We compare and contrast three encoding schemes, namely, the one-hot, binary, and Gray encodings. We show that the Gray encoding remains more resource efficient beyond the tridiagonal case, resolving an open problem posed in Ref.~\cite{PhysRevA.103.042405}. 

To estimate the measurement statistics, which is of key importance to successful simulations on quantum devices, we provide an extensive numerical analysis of the number of Pauli terms and qubit-wise commuting sets in the Hamiltonian as a function of the matrix size and the number of off-diagonals. We show that 
when the number of off-diagonals $2K+1$ exceeds the size of the matrix $N$, the number of Pauli terms and commuting sets saturate, with the Gray/binary encoding having fewer Pauli terms than the one-hot encoding. Beyond this point, more off-diagonals can be added to the problem, improving the approximation, without further load to the quantum device. 

We also introduce a new commutativity scheme, DGC, that allows for a more optimal grouping of Pauli strings at the cost of a more complex diagonalizing unitary, as compared to the qubit-commutativity scheme. We show that for small bandwidths ($K<N/2$), the Gray encodings leads to the same number of Pauli strings as compared to the binary encoding, a lower number of QC sets, and the same number of DGC sets but with a lower number of two-qubit gates. While the number of commuting
sets for the one-hot encoding is always three, we note that these measurements are on $N$ qubits, as opposed to $\operatorname{log}_2(N)$ qubits.

To demonstrate the efficacy of the Gray encoding, the one-hot and Gray encodings are compared in quantum simulations of the lowest $\half^+$ state of n+$^{14}$C, only bound by $1.04$ MeV in the $K=3$ approximate potential, for $N=8$ basis states. Indeed, compared to the Gray-code simulations with three qubits, the one-hot simulations with eight qubits are much slower and show worse performance in the presence of noise. 
In addition, we perform quantum simulations using a shot-noise and noisy simulator to inform the suitability of these simulations in the far-term error-corrected and NISQ regimes, respectively. It is remarkable that for  n+$^{10,12}$C and n+$\alpha$, we find that the shot-noise energies practically reproduce the corresponding exact value, whereas the noisy simulations coupled with the noise-resilient training method yield energies that deviate by less than one MeV. Finally, for the bound $\half^+$ orbit of the neutron-alpha optical potential deduced \textit{ab initio}, we report energy from the noisy quantum simulation that lies only within $1\sigma$ of the shot-noise outcome across a 14\% \hw\ variation. For this case, we also show that simulations with the DGC scheme outperform those using the QC scheme, and that the Gray encoding (for the QC scheme) leads to better precision and agreement with the noiseless outcome compared to the binary encoding.

Going forward from here, one could simulate systems of three clusters or larger (multi-channel reaction descriptions or multiple pairs of nucleons for reaching heavy nuclei) to give a better understanding of the scaling of the problem and whether the Gray efficacy shown here for two clusters propagates to more complex systems. It is also important to explore the dependence on the energy scale, and in particular, to seek further improvements to manage very weakly bound states and resonances. Another open theoretical question is about the optimality of the Gray code. In this work, we use the binary reflective Gray code, but any Gray code will lead to the same performance. 

\section{Code Availability}

The code for the simulations can be found as a Zenodo repository \cite{rethinasamy_2024_13696326}.

\section{Acknowledgements}

SR thanks Prerna Agarwal, Stav Haldar, Kaiyuan Ji, Dhrumil Patel, Aby Philip, and Vishal Singh for helpful and illuminating discussions. SR specifically thanks Kaiyuan Ji for the idea behind the proof of Lemma~\ref{lem:subending2}. This work was supported by the U.S.~Department of Energy under award DE-SC0023694. We are grateful to M.~Burrows for calculating the \textit{ab initio} n+$\alpha$ non-local optical potentials, which were enabled by high performance computational resources provided by LSU, the National Energy Research Scientific Computing Center (under the DOE award DE-AC02-05CH11231), as well as the Frontera computing project (under the NSF award OAC-1818253). SR and MMW acknowledge support
from the School of Electrical and Computer Engineering
at Cornell University, from the National Science Foundation under
Grant No.~2315398, and from AFRL under agreement no.~FA8750-23-2-0031.

This material is based on research sponsored by Air Force Research Laboratory under agreement number FA8750-23-2-0031. The U.S.~Government is authorized to reproduce and distribute reprints for Governmental purposes notwithstanding any copyright notation thereon. The views and conclusions contained herein are those of the authors and should not be interpreted as necessarily representing the official policies or endorsements, either expressed or implied,
of Air Force Research Laboratory or the U.S. Government.

\bibliography{Ref,lsu_latest}
\appendix

\setcounter{secnumdepth}{1}
\section{Definitions and Lemmas}
\label{app:LemDef}

In this section, we introduce new notation that we use throughout the paper. We also provide a wide range of lemmas and proofs for the results from the main text. 

\begin{lemma}
\label{lem:OpCOB}
Let $O_i$ denote the projector onto the computational basis element $i$:
\begin{equation}
    O_i \coloneqq \outerprod{i}{i},
\end{equation}
where the right-hand side is understood to be the binary representation of $i$. For example, $O_4 = \outerprod{100}{100}$. Furthermore, define $\tilde{O}_j$ to be the Pauli string composed of $I$ and  $Z$  operators such that the bits of~$j$ determine if the operator at each position is $Z$ or $I$. For a concrete example,  consider that for $j=6_{10} = 110_{2}$, 
each $1$ is represented by $Z$, and each $0$ by $I$. The operator is then given by
\begin{equation}
    \tilde{O}_6 = ZZI.
\end{equation}
Then the following equality holds:
\begin{equation}
    \label{eq:OpCOB}
    O_i = \frac{1}{2^n} \sum_{j=0}^{2^n-1} (-1)^{i \cdot j} \tilde{O}_j,
\end{equation}
where $n$ is the size of the register for $O$ and $\tilde{O}$, and $i \cdot j$ denotes the bit-wise dot-product modulo 2 of the binary representations of $i$ and $j$.
\end{lemma}

\begin{proof}
    Let $a = a_1a_2\cdots a_n$. Then
    \begin{align}
        O_a &= \outerprod{a}{a} \notag \\
        &= \outerprod{a_1}{a_1} \otimes \cdots \otimes \outerprod{a_n}{a_n} \notag \\
        &= \bigotimes_{i=1}^n \frac{I + (-1)^{a_i}Z}{2} \notag \\
        &= \frac{1}{2^n} \sum_{c=0}^{2^n-1} (-1)^{a \cdot c} \tilde{O}_c,
    \end{align}
    concluding the proof.
\end{proof} 

\begin{definition}[Projection Operators]
\label{def:projOp}
The projector onto a basis element $a$ on qubits $f$ is defined as
\begin{align*}
    P^a_{f} &\coloneqq \outerprod{a}{a}_{f} \\ 
    &= \bigotimes_{i=1} \outerprod{a_i}{a_i}_{f_i}.
\end{align*}
We also define a projector onto a subspace labelled by $b$ and $\bar{b}$ on qubits $f$ as follows:
\begin{equation}
    P^{b, \bar{b}}_{f} \coloneqq \left( \outerprod{b}{b} + \outerprod{\bar{b}}{\bar{b}} \right)_{f}.
\end{equation}
\end{definition}
\begin{definition}[Distance-$k$ operators]
\label{def:disOp}
In an $n$-qubit Gray or binary code (see Sec.~\ref{sec:encodings} for definitions), the set of distance-$k$ operators consists of all operators that connect two bitstrings that differ on $k$ bits:
\begin{equation}
    D(n, k) \coloneqq \left\{\outerprod{a}{b} + \outerprod{b}{a} : w(a \oplus b) = k\right\},
\end{equation}
where $a, b$ are bitstrings of length $n$, the variable $w$ denotes the Hamming weight, and $\oplus$ denotes bit-wise addition modulo 2.
For a two-qubit Gray code, an example of a distance-$1$ operator is
\begin{equation}
    \outerprod{01}{11} + \outerprod{11}{01}.
\end{equation}
We see that the bitstrings $01$ and $11$ differ in one position. The Hermitian operator above connects the two basis elements. The number of distinct distance-$k$ operators is given by
\begin{equation}
    \left\vert D(n, k) \right\vert = \binom{n}{k} 2^{n - 1}.
\end{equation}
To establish this equality, we first pick the $k$ bits that are different, which can be done in $\binom{n}{k}$ ways, the unflipped bits are in one of $2^{n-k}$ bitstrings, and the flipped bits can be chosen in $2^{k-1}$ ways since the operators connect two bitstrings with $k$-bits flipped. 
\end{definition}

\begin{remark}
\label{lem:dist-k-split}
    The set $D(n, k)$ of distance-$k$ operators  can be split into subsets depending on the set of $k$ flipped qubits. For example, consider $n = 3$ and $k=2$. The set $D(3, 2)$ consists of the following operators:
\begin{align}
    \outerprod{000}{011} &+ \outerprod{011}{000} ,\notag \\
    \outerprod{001}{010} &+ \outerprod{010}{001} ,\notag \\
    \outerprod{100}{111} &+ \outerprod{111}{100} ,\notag \\
    \outerprod{101}{110} &+ \outerprod{110}{101} ,\notag \\
    \outerprod{000}{110} &+ \outerprod{110}{000} ,\notag \\
    \outerprod{010}{100} &+ \outerprod{100}{010} ,\notag \\
    \outerprod{001}{111} &+ \outerprod{111}{001} ,\notag \\
    \outerprod{011}{101} &+ \outerprod{101}{011} ,\notag \\
    \outerprod{000}{101} &+ \outerprod{101}{000}, \notag \\
    \outerprod{001}{100} &+ \outerprod{100}{001}, \notag \\
    \outerprod{010}{111} &+ \outerprod{111}{010} ,\notag \\
    \outerprod{011}{110} &+ \outerprod{110}{011}.
\end{align}

As seen in Definition~\ref{def:disOp}, $\vert D(3, 2) \vert = 12$. We can split up the set of operators into three sets based on which two qubits are flipped -- $\{2, 3\}$, $\{1, 3\}$, $\{1, 2\}$. Thus, the operators of $D(3, 2)$ are then split as 
\begin{align}
    \{2, 3\} &= \bigg\{ \outerprod{000}{011} + \outerprod{011}{000}, 
    \outerprod{001}{010} + \outerprod{010}{001}, \notag \\
    &\outerprod{100}{111} + \outerprod{111}{100}, 
    \outerprod{101}{110} + \outerprod{110}{101}.\bigg\}, \notag \\
    \{1, 2\} &= \bigg\{ \outerprod{000}{110} + \outerprod{110}{000} , 
    \outerprod{010}{100} + \outerprod{100}{010} , \notag \\
    &\outerprod{001}{111} + \outerprod{111}{001} , 
    \outerprod{011}{101} + \outerprod{101}{011}. \bigg\}, \notag \\
    \{1, 3\} &= \bigg\{ \outerprod{000}{101} + \outerprod{101}{000},
    \outerprod{001}{100} + \outerprod{100}{001}, \notag \\
    &\outerprod{010}{111} + \outerprod{111}{010} ,
    \outerprod{011}{110} + \outerprod{110}{011}. \bigg\}.
\end{align}

Since $k=2$, each of the above operators can be thought of as two-qubit flips in a particular subspace. For example, $\outerprod{000}{011} + \outerprod{011}{000}$ acts like $IXX$ in the subspace spanned by $\{\ket{000}, \ket{011}\}$ since its action on any state in this subspace is the same as the action of $IXX$. This is a unifying feature---any distance-$\ell$ operator can be thought of as a product of $I$ (on the unchanged qubits) and $X$ (on the $\ell$ flipped qubits) acting on a particular subspace. Thus, an alternate way to represent the distance operators is as a composition of a projection onto the subspace of interest, followed by a string consisting of $X$ and $I$. For example,
\begin{align}
    \outerprod{000}{011} &+ \outerprod{011}{000} = (IXX) \circ (P^0_{\{1\}} \otimes P^{00, 11}_{\{2, 3\}}) , \notag \\
    \outerprod{001}{010} &+ \outerprod{010}{001} = (IXX) \circ (P^0_{\{1\}} \otimes P^{01, 10}_{\{2, 3\}}),
\end{align}
and the remaining operators can be constructed similarly.
\end{remark}

\begin{definition}
\label{def:FlDisOp}
In Remark~\ref{lem:dist-k-split}, we saw that $D(n, k)$ can be split into subsets depending on the set of $k$ flipped qubits. We label these subsets with a set $f$ of flipped qubits  of size $k$. Furthermore, we saw that the distance-$k$ operators can be written as a composition of a projector onto a particular subspace and a  Pauli string consisting of $X$ and $I$. Thus, we define the following sets of distance-$k$ operators for a fixed set of $k$ flipped qubits labelled by $f$:
\begin{multline}
    D(n, k, f) \coloneqq  \left\{ \left(I_{\bar{f}} \otimes X_f \right) P^a_{\bar{f}} \otimes P^{b, \bar{b}}_f \right. \\
    \left. : \forall a \in \{0, 1\}^{\vert \bar{f} \vert },\ \forall b \in \{0, 1\}^{\vert f \vert} \right\},
\end{multline}
where $P^a_{\bar{f}}$ and $P^{b, \bar{b}}_{f}$ are defined in Definition~\ref{def:projOp}.

For example, $D(3, 2, \{2, 3\})$ is the set consisting of the following operators:
\begin{align}
    \outerprod{000}{011} &+ \outerprod{011}{000} = IXX(P^0_{\{1\}} \otimes P^{00, 11}_{\{2, 3\}}) , \notag \\
    \outerprod{001}{010} &+ \outerprod{001}{010} = IXX(P^0_{\{1\}} \otimes P^{01, 10}_{\{2, 3\}}) , \notag \\
    \outerprod{100}{111} &+ \outerprod{111}{100} = IXX(P^1_{\{1\}} \otimes P^{00, 11}_{\{2, 3\}}) , \notag \\
    \outerprod{101}{110} &+ \outerprod{101}{110} =  IXX(P^1_{\{1\}} \otimes P^{01, 10}_{\{2, 3\}}) .
\end{align}
Motivated by the example above, we also use a shorthand to refer to a particular set $D(n, k, f)$ of operators  -- we define a string of I and X such that for all flipped qubits in~$f$, we label them by X. Therefore, $D(3, 2, \{2, 3\}) \equiv IXX$. We use the notations interchangeably.
\end{definition}

\begin{lemma} 
\label{lem:DisOpPosDep}
For an $n$-qubit Gray or binary code, the distance-$k$ operators are expressed as a linear combination of a set of Pauli strings that only depends on the set of the $k$ flipped  qubits. Alternatively, a set $f$ of flipped qubits  completely determines the set of Pauli strings.
\end{lemma}

\begin{proof}
Consider the set $D(n, k, f)$ of operators where $f$ is the set of flipped qubits. From Definition~\ref{def:FlDisOp}, we know that the operators in this set are of the form
\begin{equation}
    \left(I_{\bar{f}} \otimes X_f \right) P^a_{\bar{f}} \otimes P^{b, \bar{b}}_f,
\end{equation}
where $a$ and $b$ are bitstrings of length $\sof = k$ and $\sofbar = n-k$, respectively. On the flipped qubits, the operator is of the form
\begin{equation}
    X^{\otimes k} (\outerprod{b_1 \cdots b_k}{b_1 \cdots b_k} + \outerprod{\overline{b_1} \cdots \overline{b_k}}{\overline{b_1} \cdots \overline{b_k}}),
\end{equation}
where $b_i \in \{0, 1\}$. We now show that all these operators lead to the same set of Pauli strings, independent of the values of $\{b_i\}_i$. To show this, consider that
\begin{align}
    &X^{\otimes k}(\outerprod{b_1 \cdots b_k}{b_1 \cdots b_k} + \outerprod{\overline{b_1} \cdots \overline{b_k}}{\overline{b_1} \cdots \overline{b_k}}) \notag \\
    &= X^{\otimes k}(O_b + O_{\overline{b}}) \notag \\
    &= \frac{1}{2^k} \sum_j \left( (-1)^{b\cdot j} + (-1)^{\overline{b}\cdot j} \right) X^{\otimes k} \tilde{O}_j,
\end{align}
where the second equality follows from  Lemma~\ref{lem:OpCOB}. Upon expanding, we see that the coefficient of any $\tilde{O}_j$ is non-zero if and only if the binary representation of $j$ has even parity. The signs of the different $\tilde{O}_j$ depend on $b$, but the set of surviving $\tilde{O}_j$ is independent of $b$. The number of terms left is $2^{k-1}$.

Next, if we consider the unflipped qubits, the operators are of the form
\begin{equation}
    \outerprod{a}{a} = \bigotimes_i \outerprod{a_i}{a_i}.
\end{equation}
Since the two possible cases $\outerprod{0}{0}$ and $\outerprod{1}{1}$ are both linear combinations of $I$ and $Z$, and only differ by a negative sign, the set of Pauli strings is independent of $a$.

To summarize, for a given $n$ and $k$, all the operators in the set $D(n, k, f)$ are composed of the same Pauli strings and different sets $f$ leads to different Pauli strings. Thus, the set of Pauli strings depends only on the set $f$, i.e., on the position of the flipped qubits.
\end{proof} 

\begin{corollary}
\label{cor:samePauliSubset}
As a result of Lemma~\ref{lem:DisOpPosDep}, we see that if a Hamiltonian contains an operator from the set $D(n, k, f)$, any other operator from the same set can be added to the Hamiltonian without an increase in the number of Pauli terms. Since every operator from the set $D(n, k, f)$ consists of the same Pauli strings, adding another operator from the same set changes only the coefficients, and not the set of Pauli strings themselves.
\end{corollary}

\begin{lemma}
    The set of Pauli strings corresponding to $D(n, k, f)$ is of size $2^{n-1}$.
    \label{lem:numPauli_f_subset}
\end{lemma}
\begin{proof}
    Consider the set of operators $D(n, k, f)$, where $f$ is the set of flipped qubits. From Definition~\ref{def:FlDisOp}, we know that the operators in this set are of the form
    \begin{equation}
        \left(I_{\bar{f}} \otimes X_f \right) P^a_{\bar{f}} \otimes P^{b, \bar{b}}_f,
    \end{equation}
    where $a$ and $b$ are bitstrings of length $\sof = k$ and $\sofbar = n-k$, respectively. There are $2^{n - k}$ possible choices for $a$ and the flipped operator is of size $2^{k-1}$, as discussed in Lemma~\ref{lem:DisOpPosDep}. Thus, the total number of operators is 
    \begin{equation}
        2^{n-k} \times 2^{k-1} = 2^{n-1},
    \end{equation}
    concluding the proof.
\end{proof} 

\begin{definition}[Alternate Representation]
     An encoding represents a one-to-one mapping between Fock basis elements and computational basis elements. As seen in the main text, a Gray basis $\mathcal{G}_n$ on $n$ bits  is a list of $2^n$ basis elements:
    \begin{equation}
        \mathcal{G}_n = (g_0, g_1, \ldots, g_{2^n - 1}),
    \end{equation}
    where each $g_i$ differs from its neighbors by a single bit. Another example considered in the main text is the binary encoding. The binary basis $\mathcal{B}_n$ on $n$ bits is a list of $2^n$ basis elements:
    \begin{equation}
        \mathcal{B}_n = (b_0, b_1, \ldots, b_{2^n - 1}),
    \end{equation}
    where $b_i$ is the binary representation of the integer $i$. 
    
    A basis encoding can alternatively be represented using a sequence of flipped bits. This alternate representation is defined as $S_n$. Thus, for an encoding of size $2^n$, the alternate representation is of size $2^n-1$. 
    
    The alternate representation for the Gray code is straightforward. Since any two neighboring entries only have a single bit flipped, the entries of the alternate representation are the flipped bits. For example, the entry that connects $010$ and $011$ is $3$. 
    
    For the binary code, we use a different notation. Each entry of the alternate representation 
    is the decimal equivalent of the bit-wise addition modulo 2 of the two entries of the  encoding. For example, the entry that connects $1011$ and $1100$ is $0111_2 \equiv 7_{10}$. 
    
    For $N=8$, the alternate representation for the two encodings can be seen in Table~\ref{tab:AltRep}.
    \begin{table}[h]
        \centering
        \begin{tabular}{S|S|>{\centering\arraybackslash}p{0.1\textwidth}} \hline
        Basis & Gray & Binary \\ \hline
        $\ket{0} \leftrightarrow \ket{1}$ & $1$ & $1$ \\
        $\ket{1} \leftrightarrow \ket{2}$ & $2$ & $3$ \\
        $\ket{2} \leftrightarrow \ket{3}$ & $1$ & $1$ \\
        $\ket{3} \leftrightarrow \ket{4}$ & $3$ & $7$ \\
        $\ket{4} \leftrightarrow \ket{5}$ & $1$ & $1$ \\
        $\ket{5} \leftrightarrow \ket{6}$ & $2$ & $3$ \\
        $\ket{6} \leftrightarrow \ket{7}$ & $1$ & $1$ \\
        \hline
        \end{tabular}
        \caption{Alternate representations for the Gray and binary code on three qubits.}
        \label{tab:AltRep}
\end{table}
\end{definition}

\begin{lemma}
    \label{lem:reverse_alt_rep}
    The alternate representation for $\overline{\mathcal{G}_n}$ and $\overline{\mathcal{B}_n}$ is the same as $\mathcal{G}_n$ and $\mathcal{B}_n$, respectively.
\end{lemma}
\begin{proof}
    The binary reflective Gray code on $n$ bits is given by
    \begin{equation}
        \mathcal{G}_{n} = (\mathcal{G}_{n-1} \cdot 0, \overline{\mathcal{G}_{n-1}}, \cdot 1),
    \end{equation}
    where $\overline{G}_{n}$ is the Gray code on $n$ bits with the entries in reverse order. Reversing the entire code, we find that
    \begin{equation}
        \overline{\mathcal{G}_n} = (\mathcal{G}_{n-1} \cdot 1, \overline{\mathcal{G}_{n-1}}, \cdot 0).
    \end{equation}
    Thus, the reversed code has the same structure as the existing code with the first half entries ending with $1$, and the second half ending with $0$. The entries of alternate representation indicate which bits are flipped and therefore, is unaffected if $0 \leftrightarrow 1$. Thus, the alternate representation for the reversed Gray code is the same as the original.

    For the binary code $\mathcal{B}_n$, the reversed binary code is the same as the original with all zeros and ones flipped. Again, the entries of alternate representation indicate which bits are flipped and therefore, is unaffected if $0 \leftrightarrow 1$. Thus, the alternate representation for the reversed binary code is the same as the original.
\end{proof}

\begin{definition}[Subsequences]
    An entry in the alternate representation of an encoding represents the flipped qubits between the two corresponding entries of the encoding. We now define a subsequence of the alternate representation as an ordered subset that connects the two corresponding entries of the encoding. For example, in the Gray encoding, if $3$ connects $010$ and $011$, and $1$ connects $011$ and $111$, then $\underline{3, 1}$ connects $010$ and $111$. In this work, we represent subsequences using an underscore. For example, in the Gray encoding, the following subsequence
    \begin{equation}
        (1, 2, \underline{1, 3, 1, 2}, 1)
    \end{equation}
    connects the basis states $110$ and $101$. 
\end{definition}

\begin{lemma}
\label{rmk:equivalance}
    There exists an equivalence between subsequences of the alternate representation of an $n$-qubit encoding and Pauli strings of the form $\{I, X\}^{\otimes n} \setminus I^{\otimes n}$.
\end{lemma}    
\begin{proof}
    In the alternate representation of the Gray code, each entry $i$ indicates that the operator connecting the corresponding basis state has Pauli $X$ acting on qubit $i$. For example, in an $n=3$ qubit Gray code, an entry $2$ indicates that the two basis elements differ on the second qubit; i.e., the operator connecting them is $IXI$. Subsequences, therefore, encode a string from $\{I, X\}^{\otimes n}$ that connect the endpoints of the subsequence. For example,   
    \begin{equation}
        (1, 2, \underline{1, 3, 1, 2}, 1)
    \end{equation}
    corresponds to $X_1 X_3 X_1 X_2 = IXX$.
    
    Similarly, for the binary code, we act with a bit-wise addition modulo 2 between the binary representation of every element in the subsequence. The resulting binary string is then translated into a Pauli string of $\{I, X\}$ -- each $1$ is mapped to $X$, and each $0$ is mapped to $I$. For example,   
    \begin{equation}
        (1, 3, \underline{1, 7, 1, 3}, 1) 
    \end{equation}
    corresponds to $100$ and ultimately, $XII$.
\end{proof} \medskip

\begin{lemma}
    \label{lem:pivot_def}
    The alternate representation for the Gray and binary code on $n$ qubits can be expressed in the form 
    \begin{equation}
        S_{n} = (S_{n-1}, P_{n}, S_{n-1}),
    \end{equation}
    where $P_i$ stands for $i$ and $2^i-1$ in the Gray and binary codes, respectively, and we refer to each $P_i$ as a pivot. The term pivot refers to the fact that about $P_{n}$, the alternate representation $S_{n}$ is symmetric.
\end{lemma}
\begin{proof}
    Consider the form of the binary reflective Gray code on $n$ qubits
    \begin{equation}
        \mathcal{G}_{n} = (\mathcal{G}_{n-1} \cdot 0, \overline{\mathcal{G}_{n-1}} \cdot 1),
    \end{equation}
    where $\overline{G}_{n}$ is the Gray code on $n$ qubits with the entries in reverse order. Thus the alternate representation for the code is given by
    \begin{equation}
        S_{n} = (S_{n-1}, n, \tilde{S}_{n-1}),
    \end{equation}
    where $\tilde{S}_n$ is the alternate representation of the reversed code $\mathcal{G}_n$. Using Lemma~\ref{lem:reverse_alt_rep}, the alternate representation is
    \begin{equation}
        S_{n} = (S_{n-1}, n, S_{n-1}).
    \end{equation}

    The binary code on $n$ bits is given by 
    \begin{equation}
        \mathcal{B}_{n} = (0 \cdot \mathcal{B}_{n-1}, 1 \cdot \mathcal{B}_{n-1}).
    \end{equation}
    The last entry of $\mathcal{B}_{n-1}$ is $1^{n-1}$, and the first entry of $\mathcal{B}_{n-1}$ is $0^{n-1}$. In the first and second half of the overall code, the first bit is never flipped, and between the halves all bits are flipped. Since the decimal  representation of $1^n$ is $2^n -1$, the alternate representation is then given by
    \begin{equation}
        S_{n} = (S_{n-1}, 2^{n}-1, S_{n-1}),
    \end{equation}
    concluding the proof.
\end{proof}

\medskip
In the next two lemmas, we prove that we only need to consider subsequences that end at a power of two index. We show that those subsequences are optimal, and for all other subsequences, there exists a shorter (or equal length) subsequence ending at a power of two index that maps to the same Pauli string of $\{I, X\}$ as defined in  Lemma~\ref{rmk:equivalance}.

\begin{lemma}
\label{lem:subending2}
Consider any subsequence of the alternate representation $S_n$ on $n$ qubits. Let $H$ be the largest entry of the subsequence. Using Lemma~\ref{rmk:equivalance}, we know that the subsequence corresponds to a Pauli string consisting of $X$ and~$I$.

Then the corresponding Pauli string can be formed by another subsequence that ends at the first instance of $H$ (which is guaranteed to occur at an index that is a power of two) and has a length less than or equal to the original length. 
\end{lemma}

Let us consider a few examples before we go into the proof. For the Gray code,
\begin{align}
    (1, 2, \underline{1, 3, 1, 2}, 1) \equiv (\underline{1, 2, 1, 3}, 1, 2, 1), \notag \\
    (\underline{1, 2, 1, 3, 1}, 2, 1) \equiv (1, \underline{2, 1, 3}, 1, 2, 1), \notag \\
    (\underline{1, 2, 1}, 3, 1, 2, 1) \equiv (1, \underline{2}, 1, 3, 1, 2, 1).
\end{align}
Similarly, for the binary code
\begin{align}
    (1, 3, \underline{1, 7, 1, 3}, 1) \equiv (\underline{1, 3, 1, 7}, 1, 3, 1), \notag \\
    (\underline{1, 3, 1, 7, 1}, 3, 1) \equiv (1, \underline{3, 1, 7}, 1, 3, 1), \notag \\
    (\underline{1, 3, 1}, 7, 1, 3, 1) \equiv (1, \underline{3}, 1, 7, 1, 3, 1).
\end{align}

\begin{figure*}
     \centering
     \begin{subfigure}[b]{0.49\textwidth}
         \centering
         \includegraphics[width=.95\linewidth]{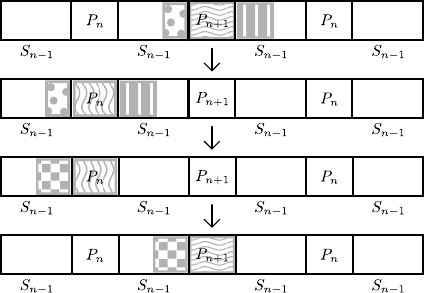}
         \caption{Case~$3$: We construct a new subsequence with the dotted section, $P_n$, and the striped section. Using the induction hypothesis, a new subsequence ending at $P_n$ can be found. Lastly, we replace $P_n$ with $P_{n+1}$ again.}
         \label{fig:case3}
     \end{subfigure}
     \vspace{1cm}
     \begin{subfigure}[b]{0.49\textwidth}
         \centering
         \includegraphics[width=.95\linewidth]{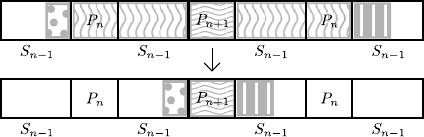}
         \caption{Case~$4$: The vertical wave section of the subsequence can be canceled out. We construct a new subsequence with the dotted section, $P_{n+1}$, and the striped section, which reduces to Case~$3$.}
         \label{fig:case4}
     \end{subfigure}
     \begin{subfigure}[b]{0.85\textwidth}
         \centering
         \includegraphics[width=.55\linewidth]{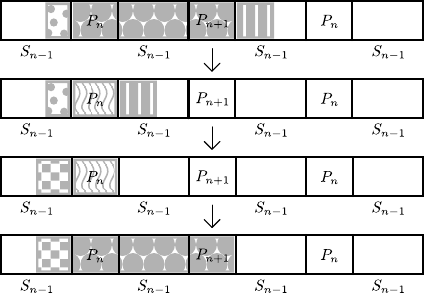}
         \caption{Case~$5$: We preserve the filled circle section of the subsequence and construct a new subsequence with the dotted section, $P_n$, and the striped section. Using the induction hypothesis, a new subsequence ending at $P_n$ can be found. The checkered section can now be appended to the preserved filled circle section.}
         \label{fig:case5}
     \end{subfigure}
    \caption{Proof of Lemma~\ref{lem:subending2} can be broken down into multiple cases. We label sections of the subsequence with different colors. $P_i$ stands for $i$ and $2^i-1$ in the Gray and binary codes respectively.}
    \label{fig:proofLemsubending2}
\end{figure*}

\begin{proof}[Proof of Lemma~\ref{lem:subending2}]
    We prove the result using an inductive approach on $n$, where $n$ is the number of qubits. We have two base cases, $n=1$ and $n=2$.
    For $n=1$, $S_1 = \{1\}$. There exists only one possible subsequence and it ends at a power of two. For $n=2$, $S_2 = \{1, 2, 1\}$ or $\{1, 3, 1\}$ for the Gray and binary code, respectively. There are six possible subsequences:
    \begin{align}
        (\underline{1}, 2, 1) &\equiv (\underline{1}, 2, 1), &\quad (\underline{1}, 3, 1) &\equiv (\underline{1}, 3, 1) , \notag \\
        (\underline{1, 2}, 1) &\equiv (\underline{1, 2}, 1) , &\quad (\underline{1, 3}, 1) &\equiv (\underline{1, 3}, 1) , \notag \\
        (\underline{1, 2, 1}) &\equiv (1, \underline{2}, 1) , &\quad (\underline{1, 3, 1}) &\equiv (1, \underline{3}, 1) , \notag \\
        (1, \underline{2}, 1) &\equiv (1, \underline{2}, 1) , &\quad (1, \underline{3}, 1) &\equiv (1, \underline{3}, 1) , \notag \\
        (1, \underline{2, 1}) &\equiv (\underline{1, 2}, 1) , &\quad (1, \underline{3, 1}) &\equiv (\underline{1, 3}, 1) , \notag \\
        (1, 2, \underline{1}) &\equiv (\underline{1}, 2, 1) , &\quad (1, 3, \underline{1}) &\equiv (\underline{1}, 3, 1),
    \end{align}
    where the columns represent a Gray code and binary code, respectively. We see that there always exists a subsequence with length less than or equal to the original length, ending at a power of two index.
    
    Next, we state the induction hypothesis: Lemma~\ref{lem:subending2} holds for some positive integer $n$. The induction step is to now show that it holds for $n+1$. From Lemma~\ref{lem:pivot_def}, we see that
    \begin{align}
        S_n & = (S_{n-1}, P_n, S_{n-1}),\\
        S_{n+1} & =  (S_{n}, P_{n+1}, S_{n})\notag \\
        &  = (S_{n-1}, P_n, S_{n-1}, P_{n+1}, S_{n-1}, P_n, S_{n-1}),
    \end{align}
    where $P_i$ stands for $i$ and $2^i-1$ in the Gray and binary codes, respectively, and we refer to each $P_i$ as a pivot. The pivots $P_n$, $P_{n+1}$, and $P_n$ occur at indices $2^{n-1}$, $2^n$, and $3 \cdot 2^{n-1}$, respectively. 
    Consider an arbitrary subsequence
    \begin{equation}
        a\coloneqq (a[L], \ldots, a[R])
    \end{equation}
    with endpoint indices $L$ and $R$, where $L \leq R$. Let $\len \coloneqq R - L + 1$ be the length of this subsequence. Based on the values of $L$ and $R$, we split the problem into multiple cases.

    Note that in the following cases, we use the different brackets to indicate if the endpoint of the interval is included -- square brackets indicate that the endpoint is included and regular parentheses indicate that the endpoint is not included. For example, $(5, 11]$ is the set $\{6, 7, 8, 9, 10, 11\}$.
    
    \begin{itemize}
        \item Case~$1$: $L, R \in [0, 2^n)$. Since both $L$ and $R$ are in the left half of the sequence, the subsequence must exist in $S_n = (S_{n-1}, P_n, S_{n-1})$. Using the induction hypothesis, there exists an alternate subsequence that lies in $S_n$ ending at a power of two index. The same subsequence thus exists in the left half of $S_{n+1}$.
        
        \item Case~$2$: $L, R \in (2^n, 2^{n+1})$. Since both $L$ and $R$ are in the right half of the sequence, the subsequence must exist in $S_n = (S_{n-1}, P_n, S_{n-1})$. Thus, this case is similarly covered by the induction hypothesis; i.e., there exists an alternate subsequence that lies in $S_n$ ending at a power of two index. The same subsequence thus exists in the left half of $S_{n+1}$.
        
        \item Case~$3$: $L \in (2^{n-1}, 2^n)$, $R \in (2^n, 3 \cdot 2^{n-1})$. A pictorial representation of the proof can be found in Fig.~\ref{fig:case3}. Intuitively, this subsequence, with $P_{n+1}$ replaced with $P_n$, must already exist in the left half of the overall sequence. To formalize this notion, we now construct a new subsequence
        \begin{equation*}
            \quad \qquad (a[L], \ldots, a[2^n-1], P_n, a[2^n+1], \ldots, a[R]).
        \end{equation*}
        In essence, we have replaced $P_{n+1}$ with $P_n$ in the original subsequence. This subsequence is now guaranteed to exist in $S_n = (S_{n-1}, P_n, S_{n-1})$. By the inductive hypothesis, there exists a subsequence $(a[k], \ldots, a[2^{n-1}-1], P_n)$ of length $\leq \len$, where $k \in [0, 2^{n-1})$. By symmetry, the following subsequence 
        \begin{equation*}
            \qquad (a[k+2^{n-1}], \ldots, a[2^n-1], P_{n+1}),
        \end{equation*}
        where we have replaced $P_n$ with $P_{n+1}$ again, must exist, has length $\leq \len$, and ends at $n+1$.

        \item Case~$4$: $L \in [0, 2^{n-1}]$, $R \in [3 \cdot 2^{n-1}, 2^{n+1})$. A pictorial representation of the proof can be found in Fig.~\ref{fig:case4}. Intuitively, this subsequence contains $(P_n, S_{n-1}, P_{n+1}, S_{n-1}, P_n)$ as a part of it. This part can be effectively reduced to just $P_{n+1}$ since every other entry occurs an even number of times. Thus, we create a new subsequence
        \begin{multline*}
            \qquad (a[L], \ldots, a[2^{n-1}-1], P_{n+1}, \\ a[3 \cdot 2^{n-1} + 1], \ldots, a[R]).
        \end{multline*}
        By symmetry, this subsequence is the same as 
        \begin{multline*}
            \qquad (a[L + 2^{n-1}], \ldots, a[2^n-1], P_{n+1}, \\ a[2^n + 1], \ldots, a[R - 2^{n-1}]),
        \end{multline*}
        where we have shifted the indices of the left and right half by $2^{n-1}$ up and down respectively. This then reduces to Case~$3$, handled earlier.

        \item Case~$5$: $L \in [0, 2^{n-1}]$, $R \in [2^n, 3 \cdot 2^{n-1})$. A pictorial representation of the proof can be found in Fig.~\ref{fig:case5}. In this case, we preserve the subsequence $(P_n, S_{n-1}, P_{n+1})$ and create a new subsequence of the form 
        \begin{multline*}
            \qquad (a[L], \ldots, a[2^{n-1}-1], P_n, \\ a[2^{n-1} + 1], \ldots, a[R - 2^{n-1}]).
        \end{multline*}
        This subsequence is guaranteed to exist in $S_n = (S_{n-1}, P_n, S_{n-1})$. By the inductive hypothesis, there exists a subsequence $(a[k], \ldots, a[2^{n-1}-1], P_n)$. We now reconstruct the original subsequence as follows
        \begin{equation*}
            \qquad (a[k], \ldots, a[2^{n-1}-1], P_n, S_{n-1}, P_{n+1}).
        \end{equation*}
        The length of this subsequence is $\leq \len$ and it ends at a power of two.

        \item Case~$6$: $L \in (2^{n-1}, 2^n]$, $R \in [3 \cdot 2^{n-1}, 2^{n+1}]$. By symmetry, this subsequence can be reflected about the midpoint $n+1$. This then reduces to Case~$5$, handled earlier.
    \end{itemize}
    Thus, in all possible cases, we have shown that the inductive step holds if we assume the inductive hypothesis to be true.
\end{proof}

\begin{lemma}
\label{lem:shorterSubseq}
Given a subsequence that ends at an index which is a power of two, there does not exist another subsequence of a shorter length  that leads to the same Pauli string. 
\end{lemma}
\begin{proof}
Given a subsequence ending at index $2^m$, the corresponding Pauli string must have $X$ acting on qubit $m+1$. We now provide a proof by contradiction. We first assume that a shorter subsequence exists. The two possible cases are:
\begin{itemize}
    \item Case~$1$: The shorter subsequence ends before index $2^m$. More concretely, the shorter subsequence exists within $[0, 2^m)$ and using Lemma~\ref{lem:subending2}, there exists an equivalent subsequence that ends at index $2^{m-1}$. No possible subsequence in this region can have its corresponding Pauli string with $X$ acting on qubit $m+1$.
    
    \item Case~$2$: A shorter subsequence ends at index $2^m$. This subsequence has $X$ acting on qubit $m+1$. There are $2^m$ subsequences of the form $(k, \ldots, 2^m)$ for $k \in \{1, \ldots, 2^m\}$. These subsequences map to all possible Pauli strings of length~$m$. Thus, two different subsequences ending at $2^m$ cannot lead to the same Pauli string. 
\end{itemize}
Thus, we see that there cannot exist any other subsequence of a shorter length that leads to the same Pauli string.
\end{proof} 

\medskip
We note that in the present study, the Hamiltonian for $K=0$ (diagonal potential) is tridiagonal because of the kinetic energy term. As a result, the entries for $K=1$ in the following lemmas and remarks are used for $K=0$.

\begin{lemma}
\label{lem:zero_diag_terms}
    Let $i, j$ be two $n$-bit binary strings, and let $N=2^n$. Then, the following statement is true
    \begin{equation}
        P(n, K) \colon \summ{n} \dopr{i}{j} i^K = 0 \iff \vert j \vert > K,
    \end{equation}
    for all $0 \leq K \leq N$ and $n>0$.
\end{lemma}
\begin{proof}
    We prove the statement using induction. More concretely, we use simple induction on $n$ and strong induction on $K$. Thus, we assume $P(n-1, K), P(n-1, K-1), \ldots, P(n-1, 0)$ to be true, and use them to prove $P(n, K)$. We first show the sets of base cases $P(1, K)$ and $P(n, 0)$.\\
    
    \textbf{Base Case 1 : }$P(1, K)$. The left hand side now becomes
    \begin{align}
        &\dopr{0}{j} 0^K + \dopr{1}{j} 1^K \notag \\
        &= 0^K + (-1)^j = 0 \iff K=0, j=1,
    \end{align}
    where we use the fact that $0^0 = 1$. Thus, $\vert j \vert > K$.\\
    
    \textbf{Base Case 2 : } $P(n, 0)$. The left hand side now becomes
    \begin{equation}
        \summ{n} \dopr{i}{j} i^0 = 0,
    \end{equation}
    which is true $\forall j \neq 0$.
    Thus, $\vert j \vert > 0$.\\
    
    To prove the equivalence $P(n, K)$, we start from the right hand side $\vert j \vert > K$, and show that it is equivalent to the left hand side. Let the binary expansion of $j$ be given by $j_1 j_2 \ldots j_n$.\\

    \textbf{Case 1:} $j_1 = 0$. Thus, $\vert j_2 \ldots j_n \vert > K$. Using the inductive hypotheses and the fact that $\vert j_2 \ldots j_n \vert > K$ satisfies the right hand side of all the hypotheses, we have the following equalities:
    \begin{align}
        \summ{n-1} \dopr{i}{(j_2 \ldots j_n)} i^K &= 0 \qquad \text{Using }P(n-1, K) \label{eq:P_n-1_k_1}\\
        \summ{n-1} \dopr{i}{(j_2 \ldots j_n)} i^{K-1} &= 0 \qquad \text{Using }P(n-1, K-1) \label{eq:P_n-1_k-1_1}\\
        &\shortvdotswithin{=}
        \summ{n-1} \dopr{i}{(j_2 \ldots j_n)} i^{0} &= 0 \qquad \text{Using }P(n-1, 0) \label{eq:P_n-1_0_1}.
    \end{align}
    In each of the statements above, $i$ is an $(n-1)$-bit binary string. Now consider the following equation:
    \begin{align}
        &  \binom{K}{0} (2^n)^0 \left[ \summ{n-1} \dopr{i}{(j_2 \ldots j_n)} i^K \right] \notag \\
        &+ \binom{K}{1} 2^n \left[ \summ{n-1} \dopr{i}{(j_2 \ldots j_n)} i^{K-1} \right] \notag \\
        &+ \binom{K}{2} (2^n)^2 \left[ \summ{n-1} \dopr{i}{(j_2 \ldots j_n)} i^{K-2} \right] \notag \\
        &+ \cdots + \binom{K}{K} (2^n)^K \left[ \summ{n-1} \dopr{i}{(j_2 \ldots j_n)} i^{0} \right] \notag \\
        &= 0.
    \end{align}
    The equality is because each term within square brackets is zero (using \eqref{eq:P_n-1_k_1}-\eqref{eq:P_n-1_0_1}). Thus, using the binomial theorem, 
    \begin{equation}
        \left[ \summ{n-1} \dopr{i}{(j_2 \ldots j_n)} (2^n + i)^K \right] = 0.
    \end{equation}
    Combining with the first step of the inductive tree \eqref{eq:P_n-1_k_1},
    \begin{multline}
    \left[ \summ{n-1} \dopr{i}{(j_2 \ldots j_n)} (i)^K \right] \\ 
    + \left[ \summ{n-1} \dopr{i}{(j_2 \ldots j_n)} (2^n + i)^K \right] = 0.
    \end{multline}

    In the above two summations, $i$ is an $n-1$ bit binary string. We now create an $n$ bit binary string $i'$ by appending either $0$ or $1$ to the front of $i$ and append $j_1$ to the front of $j_2 \ldots j_n$. Since $j_1=0$, this extra bit has no effect. Thus, the equation becomes 
    \begin{multline}
    \left[ \summ{n-1} \dopr{0i}{(j_1 j_2 \ldots j_n)} (i)^K \right] \\
    + \left[ \summ{n-1} \dopr{1i}{(j_1 j_2 \ldots j_n)} (2^n + i)^K \right] = 0.
    \end{multline}
    Now, we notice that $1i$ is the binary expansion of $2^n + i$. Thus, the summation can be written as 
    \begin{multline}
    \left[ \summ{n-1} \dopr{0i}{(j_1 j_2 \ldots j_n)} (i)^K \right] \\
    + \left[ \sum_{i' = 2^{n-1}}^{2^n-1} \dopr{i'}{(j_1 j_2 \ldots j_n)} (i')^K \right] = 0.
    \end{multline}
    Since $i'$ is just a dummy index, we replace it with $i$ and combine it with the first term, finally leading to
    \begin{equation}
        \summ{n} \dopr{i}{j} i^K = 0.
    \end{equation}
    \\
    
    \textbf{Case 2:} $j_1 = 1$. Thus, $\vert j_2 \ldots j_n \vert > K-1$. Using the inductive hypotheses and the fact that $\vert j_2 \ldots j_n \vert > K-1$ satisfies the right hand side of all the hypotheses, we have the following equalities:
    \begin{align}
        \summ{n-1} \dopr{i}{(j_2 \ldots j_n)} i^{K-1} &= 0 \qquad \text{Using }P(n-1, K-1)\label{eq:P_n-1_k-1_2} \\
        \summ{n-1} \dopr{i}{(j_2 \ldots j_n)} i^{K-2} &= 0 \qquad \text{Using }P(n-1, K-2)\label{eq:P_n-1_k-2_2} \\
        &\shortvdotswithin{=}
        \summ{n-1} \dopr{i}{(j_2 \ldots j_n)} i^{0} &= 0 \qquad \text{Using }P(n-1, 0) \label{eq:P_n-1_0_2}.
    \end{align}
    In each of the statements above, $i$ is an $(n-1)$-bit binary string. Now consider the following equation:
    \begin{align}
        &-\binom{K}{1} 2^n \left[ \summ{n-1} \dopr{i}{(j_2 \ldots j_n)} i^{K-1} \right] \notag \\
        &-\binom{K}{2} (2^n)^2 \left[ \summ{n-1} \dopr{i}{(j_2 \ldots j_n)} i^{K-2} \right] \notag \\
        &- \cdots - \binom{K}{K} (2^n)^K \left[ \summ{n-1} \dopr{i}{(j_2 \ldots j_n)} i^{0} \right] \notag \\
        = 0.
    \end{align}
    The equality is because each term within square brackets is zero (using \eqref{eq:P_n-1_k-1_2}-\eqref{eq:P_n-1_0_2}).  Adding and subtracting the following term, 
    \begin{equation}
        \left[ \summ{n-1} \dopr{i}{(j_2 \ldots j_n)} i^{K} \right],
    \end{equation}
    we get 
    \begin{align}
        &\left[ \summ{n-1} \dopr{i}{(j_2 \ldots j_n)} i^{K} \right] \notag \\
        &- \left\{ \binom{K}{0} (2^n)^0 \left[ \summ{n-1} \dopr{i}{(j_2 \ldots j_n)} i^{K} \right] \right.\notag \\
        &+ \binom{K}{1} 2^n \left[ \summ{n-1} \dopr{i}{(j_2 \ldots j_n)} i^{K-1} \right] \notag \\
        &+ \binom{K}{2} (2^n)^2 \left[ \summ{n-1} \dopr{i}{(j_2 \ldots j_n)} i^{K-2} \right] \notag \\
        &+ \left. \cdots + \binom{K}{K} (2^n)^K \left[ \summ{n-1} \dopr{i}{(j_2 \ldots j_n)} i^{0} \right] \right\}= 0.
    \end{align}
    Combining the terms within the curly brackets using the binomial theorem, we get
    \begin{multline}
        \left[ \summ{n-1} \dopr{i}{(j_2 \ldots j_n)} i^{K} \right] \\
        - \left[ \summ{n-1} \dopr{i}{(j_2 \ldots j_n)} (2^n + i)^{K} \right] = 0.
    \end{multline}
    In the above two summations, $i$ is an $n-1$ bit binary string. We now create an $n$ bit binary string $i'$ by appending either $0$ or $1$ to the front of $i$ and append $j_1$ to the front of $j_2 \ldots j_n$. Since $j_1=1$, this extra bit accounts for the negative sign. Thus, the equation becomes  
    \begin{multline}
    \left[ \summ{n-1} \dopr{0i}{(j_1 j_2 \ldots j_n)} (i)^K \right] \\
    + \left[ \summ{n-1} \dopr{1i}{(j_1 j_2 \ldots j_n)} (2^n + i)^K \right] = 0.
    \end{multline}
    Now, we notice that $1i$ is the binary expansion of $2^n + i$. Thus, the summation can be written as 
    \begin{multline}
    \left[ \summ{n-1} \dopr{0i}{(j_1 j_2 \ldots j_n)} (i)^K \right] \\
    + \left[ \sum_{i' = 2^{n-1}}^{2^n-1} \dopr{i'}{(j_1 j_2 \ldots j_n)} (i')^K \right].
    \end{multline}
    Since $i'$ is just a dummy index, we replace it with $i$ and combine it with the first term, finally leading to
    \begin{equation}
        \summ{n} \dopr{i}{j} i^K = 0.
    \end{equation}
    Thus, for both cases $j_1=0, 1$, we showed that $P(n, K)$ follows from the inductive hypotheses, concluding the proof.
\end{proof} 

\begin{lemma}
\label{lem:diag_ham_expansion}
    Consider an $N \times N$ matrix $H$ with diagonal entries of the form 
    \begin{equation}
        H_{i, i} = \sum_{k=0}^K a_k i^k,
    \end{equation}
    for some set  $\{a_k\}_{k=0}^K$ of real constants and for all $i \in \left\{0, \ldots, N-1\right\}$. Considering only the diagonal part of the matrix, and using the binary encoding, we get the following equality:
    \begin{align}
        H_{\operatorname{diag}} &= \sum\limits_{i=0}^{N-1} \sum_{k=0}^K a_k i^k \outerprod{i}{i} \notag \\
        &= \sum\limits_{i=0}^{N-1} \sum_{k=0}^K a_k i^k O_{i} \notag \\
        &= \frac{1}{2^n} \sum\limits_{j=0}^{N-1} \sum_{k=0}^K a_k \left[ \sum\limits_{i=0}^{N-1} (-1)^{i \cdot j} i^k\right] \tilde{O}_{j},
        \label{eq:diag_ham_expansion}
    \end{align}
    where the third equality follows from Lemma~\ref{lem:OpCOB}. Then, the number of Pauli terms accounted for by the diagonal part of the matrix is given by 
    \begin{equation}
        \operatorname{d}(n, K) \coloneqq \sum_{m=0}^K \binom{n}{m},
    \end{equation}
    where $n = \log_2(N).$
\end{lemma}
\begin{proof}
    In \eqref{eq:diag_ham_expansion}, using Lemma~\ref{lem:zero_diag_terms}, the term in square brackets is non-zero if and only if $\vert j \vert \leq K$ and all terms with $\vert j \vert > K$ get cancelled. Thus, the surviving terms have $\vert j \vert \leq K$ and the number of such terms is given by
    \begin{equation}
        \sum_{m=0}^K \binom{n}{m},
    \end{equation}
    which is the number of ways to pick $j$ such that $\vert j \vert \leq K$.
\end{proof}

\begin{remark}
    Lemma~\ref{lem:diag_ham_expansion} shows that the number of Pauli terms for the diagonal part of the Hamiltonian, using the binary encoding, is given by $\operatorname{d}(n, K)$. Instead, if we use the Gray encoding, the diagonal part of the Hamiltonian is given by
    \begin{align}
        H_{\operatorname{diag}} &= \sum\limits_{i=0}^{N-1} \sum_{k=0}^K a_k i^k \outerprod{i}{i} \notag \\
        &= \sum\limits_{i=0}^{N-1} \sum_{k=0}^K a_k i^k O_{i} \notag \\
        &= \frac{1}{2^n} \sum\limits_{j=0}^{N-1} \sum_{k=0}^K a_k \left[ \sum\limits_{i=0}^{N-1} (-1)^{e(i) \cdot j} i^k\right] \tilde{O}_{j}.
    \end{align}
    Replacing $i$ with $e(i)$ in the exponent changes the relative signs of each term. Thus, the set of surviving $j$ is different for the binary and Gray encodings. However, since we consider a sum of all possible $j$, the size of the set of surviving terms is the same and is equal to $\operatorname{d}(n, K)$.
\end{remark}

\begin{lemma}
\label{lem:NumPauliGrayBinary}
    The number of Pauli terms in the Hamiltonian $H_{N, K}$ for the Gray and binary code is given by ($N = 2^n$),
    \begin{equation}
        \vert H(N, K) \vert
            =  \begin{cases} 
                \operatorname{d}(n, 1) + n 2^{n-1} & K = 0  \\
                \operatorname{d}(n, K) + 2^{n-1} \sum\limits_{k=1}^K  \overline{n}_k & 1 \leq K \leq 2^{n-1} \\
                2^{n-1}(1+2^n) & K>2^{n-1},
            \end{cases}
    \end{equation}
    where $\overline{n}_k \coloneqq n - \lceil \operatorname{log}_2(k) \rceil $.
\end{lemma}
\begin{proof}
From Lemmas~\ref{lem:subending2} and \ref{lem:shorterSubseq}, we see that we are interested in subsequences ending at powers of two only. For any subsequence not ending at a power of two, Lemma~\ref{lem:subending2} states that there exists an equivalent subsequence ending at a power of two of shorter (or equal) length. Furthermore, Lemma~\ref{lem:shorterSubseq} states that for any subsequence ending at a power of two, no shorter subsequence gives the same Pauli string. Lastly, the set of subsequences ending at a power of two covers all Pauli strings. Thus, the set of subsequences ending at a power of two forms a complete and optimal set. 

To prove the current lemma, we establish a correspondence between subsequences and Pauli strings of the Hamiltonian. As seen in Remark~\ref{rmk:equivalance}, for every subsequence, we can associate a Pauli string composed of $\{I, X\}$. The Pauli string corresponds to the flipped bits in the two encoded basis states at the end of the subsequence. For example, consider an $n=3$ Gray code as seen in Table~\ref{tab:AltRep}. The subsequence
\begin{equation}
    (\underline{1, 2, 1, 3}, 1, 2, 1) \equiv IXX,
\end{equation}
indicates that $\ket{0} \rightarrow \ket{000}$ and $\ket{4} \rightarrow \ket{011}$ differ on qubits $2$ and $ 3$, as the equivalent Pauli string has $X$ acting on qubits $2$ and $3$. 
Next, we saw in Definition~\ref{def:FlDisOp} that the set of flipped qubits specifies a particular set of operators, each composed of the same Pauli strings but differing only in their coefficients. Thus, the set can be specified by the set of Pauli strings that its operators are composed of. For example $IXX$ corresponds to the set 
\begin{equation}
    D(3, 2, \{2, 3\}) = \{ \left(I_1 \otimes X_{2, 3} \right) P^a_{1} \otimes P^{b, \bar{b}}_{2, 3} : \forall a, \forall b\}.
\end{equation}
Each operator in this example is made up of the following set of Pauli strings:
\begin{equation}
    \left\{(I \otimes XX)\left(\frac{I \pm Z}{2} \otimes \frac{II \pm ZZ}{2}\right) \right\}.
\end{equation}

Thus, there exists a one-to-one mapping between subsequences ending at a power of two and Pauli strings of the form $\{I, X\}^{\otimes n} \setminus I^{\otimes n}$. We also showed that there exists a one-to-one mapping between Pauli strings of the form $\{I, X\}^{\otimes n} \setminus I^{\otimes n}$ and sets of distance operators with fixed flipped qubits, completing the connection to subsequences. Since we are only interested in subsequences ending at a power of two, we need to find the number of such subsequences. 

Next, we provide an important connection that allows us to quantify the number of Pauli terms for a particular $N$ and $K$. The truncation parameter $K$ in $H_{N, K}$ specifies the maximum length of subsequences allowed; thus, to quantify the number of Pauli terms, we consider all subsequences of length $k \in \{1, \ldots, K\}$. 

We first provide an example. In the case of $K=3$ for a three-qubit Gray code, we need to consider subsequences of length $k\in \{1, 2, 3\}$. Furthermore, we need to only consider subsequences ending at a power of two index. Thus, the subsequences we need to consider are 
\begin{align}
    (\underline{1}, 2, 1, 3, 1, 2, 1) &\equiv XII,\\
    (1, \underline{2}, 1, 3, 1, 2, 1) &\equiv IXI,\\
    (1, 2, 1, \underline{3}, 1, 2, 1) &\equiv IIX,\\
    (\underline{1, 2}, 1, 3, 1, 2, 1) &\equiv XXI,\\
    (1, 2, \underline{1, 3}, 1, 2, 1) &\equiv XIX,\\
    (1, \underline{2, 1, 3}, 1, 2, 1) &\equiv XXX,
\end{align}
where the first three entries account for $k=1$, the next two account for $k=2$, and the last entry accounts for $k=3$.

A subsequence of length $k$ cannot end at index $2^m$ if $k > 2^m$. Since the highest power of two in an $n$-qubit Gray code is $2^{n-1}$, for $k>2^{n-1}$, we expect to see no new Pauli strings. Thus, for a fixed $k$ such that $1 < k < 2^{n-1}$, we need to count all powers of two greater than $k$ and less than $2^n$. There are 
\begin{equation}
    \overline{n}_k = n - \lceil \operatorname{log}_2(k) \rceil
\end{equation}
subsets, each contributing $2^{n-1}$ Pauli terms (see Lemma~\ref{lem:numPauli_f_subset}). For a concrete example, the subsets for different $k$ for a Gray code on $n=4$ qubits is shown in Table~\ref{tab:GSubsetsVsk}. The columns represent the index endpoints of the subsequences.
\begin{table}[h]
        \centering
        \begin{tabular}{s|s|s|s|>{\centering\arraybackslash}p{0.06\textwidth}} \hline
        $k$ & $1$  & $2$  & $4$  & $8$  \\ \hline
        $1$ & XIII & IXII & IIXI & IIIX \\
        $2$ &      & XXII & XIXI & XIIX \\
        $3$ &      &      & XXXI & XXIX \\ 
        $4$ &      &      & IXXI & IXIX \\
        $5$ &      &      &      & IXXX \\
        $6$ &      &      &      & XXXX \\ 
        $7$ &      &      &      & XIXX \\
        $8$ &      &      &      & IIXX \\
        \hline
        \end{tabular}
        \caption{Gray encoding subsequences as a function of $k$ for an $n=4$ code. For a fixed $k$, we see that there are $\overline{n}_k = n - \lceil \operatorname{log}_2(k) \rceil$ entries and for $k>2^{n-1}$, there are no new entries.}
        \label{tab:GSubsetsVsk}
\end{table} 

Similarly, the table for a binary code is shown in Table~\ref{tab:BSubsetsVsk}.
\begin{table}[h]
        \centering
        \begin{tabular}{s|s|s|s|>{\centering\arraybackslash}p{0.06\textwidth}} \hline
        $k$ & $1$  & $2$  & $4$  & $8$  \\ \hline
        $1$ & IIIX & IIXX & IXXX & XXXX \\
        $2$ &      & IIXI & IXXI & XXXI \\
        $3$ &      &      & IXIX & XXIX \\ 
        $4$ &      &      & IXII & XXII \\
        $5$ &      &      &      & XIXX \\
        $6$ &      &      &      & XIXI \\ 
        $7$ &      &      &      & XIIX \\
        $8$ &      &      &      & XIII \\
        \hline
        \end{tabular}
        \caption{Binary encoding subsequences as a function of $k$ for an $n=4$ code. For a fixed $k$, we see that there are $\overline{n}_k = n - \lceil \operatorname{log}_2(k)\rceil$ entries and for $k>2^{n-1}$, there are no new entries.}
        \label{tab:BSubsetsVsk}
\end{table} 

Thus, the total number of Pauli terms for $1 \leq K \leq 2^{n-1}$ is given by
\begin{equation}
    \operatorname{d}(n, K) + \sum\limits_{k=1}^K \overline{n}_k 2^{n-1} ,
\end{equation}
where the first term arises when $k=0$, i.e., the number operators of the form $\{I, Z\}^{\otimes n}$ using Lemma~\ref{lem:diag_ham_expansion}.

Thus, we finally see that 
\begin{equation}
        \vert H(N, K) \vert 
        =  \begin{cases} 
                \operatorname{d}(n, 1) + n 2^{n-1} & K = 0 \\
                \operatorname{d}(n, K) + 2^{n-1} \sum\limits_{k=1}^K  \overline{n}_k & 1 \leq K \leq 2^{n-1} \\
                2^{n-1}(1+2^n) & K>2^{n-1},
            \end{cases}
\end{equation}
concluding the proof.
\end{proof} 

\begin{lemma}
    The set of Pauli strings corresponding to $D(n, k, f)$ consists of $2^{\vert f \vert -1}$ qubit-wise commuting sets.
    \label{lem:subsetCG}
\end{lemma}
\begin{proof}
As a reminder, the set of operators $D(n, k, f)$ is defined as
\begin{multline}
    D(n, k, f) \coloneqq \left\{ \left(I_{\bar{f}} \otimes X_f \right) P^a_{\bar{f}} \otimes P^{b, \bar{b}}_f \right. \\
    \left. : \forall a \in \{0, 1\}^{\vert \bar{f} \vert },\ \forall b \in \{0, 1\}^{\vert f \vert}\right\},
\end{multline}
where $P^a_{\bar{f}}$ and 
and $P^{b, \bar{b}}_{f}$ are defined in Definition~\ref{def:projOp}. 
Since we are interested in qubit-wise commutativity, we consider the $f$ and $\bar{f}$ qubits separately. On the unflipped qubits $\bar{f}$, the action of every element in $D(n, k, f)$ is $P^a$ for every bitstring $a$. Expanding, we know that for all $a$, $P^a$ is a linear combination of Pauli strings composed of $\{I, Z\}$ only. Thus, all of them pairwise-qubit commute.

On the flipped qubits $f$, the action of every element in $D(n, k, f)$ is $X_f \circ P^{b, \bar{b}}_f$ for every bitstring $b$. As seen before, $P_f^{b, \bar{b}}$ is composed of $2^{\vert f \vert - 1}$ Pauli strings. None of the Pauli strings qubit-wise commute because of the overall composition with the Pauli $X$ string. Thus, every set $D(n, k, f)$ consists of $2^{\vert f \vert - 1}$ qubit-wise commuting Pauli terms.
\end{proof} 
\medskip

As an example related to Lemma~\ref{lem:subsetCG}, consider the set $D(4, 2, \{3, 4\})$, alternatively labeled $IIXX$,
\begin{equation}
    IIXX = \left\{ \left(I_{1, 2} \otimes X_{3, 4} \right) P^a_{\{1, 2\}} \otimes P^{b, \bar{b}}_{\{3, 4\}} \right\}.
\end{equation}
As mentioned earlier, the values of $a$ and $b$ change the coefficient of the different Pauli strings, but not the set of strings themselves. Thus, we consider $a=00$ and $b=00$. For this choice, the operator is given by
\begin{multline}
     \left(I_{1, 2} \otimes X_{3, 4} \right) ((II + IZ + ZI + ZZ) \otimes (II + ZZ)) \\
     = (II + IZ + ZI + ZZ) \otimes (XX + YY),
\end{multline}
up to a normalization constant. Thus, we see that the set of Pauli terms in the linear combination can be split into the following two qubit-wise commuting sets:
\begin{align}
 & \{IIXX, IZXX, ZIXX, ZZXX\}, \\
  & \{IIYY, IZYY, ZIYY, ZZYY\} . 
\end{align}
Note that different values of $a$ and $b$ will change the relative sign of some Pauli strings, but preserve the set of Pauli strings.

\begin{lemma}
    \label{lem:NumQCGBinary}
    The number of qubit-wise commuting sets for $H_{N, K}$ for the binary code is given by
    \begin{align}
        &\vert H(N, K) \vert_C \notag \\
        &=  \begin{cases} 
                2^n & K = 0 \\
                1 + \sum\limits_{k=1}^K 2^{\vert b(\bar{k}) \vert}\left[ 1- 2^{-\overline{n}_k} \right] & 1 \leq K \leq 2^{n-1} \\
                \frac{1}{2} \left( 1 + 3^n \right) & K > 2^{n-1},
            \end{cases}
    \end{align}
    where $\vert w \vert$ is the Hamming weight of the string $w$, $\overline{n}_k \coloneqq n - \lceil \operatorname{log}_2(k) \rceil$, and $\bar{k} \coloneqq 2^n-k$.
\end{lemma}
\begin{proof}
Consider $1 \leq K \leq 2^{n-1}$ and consider a column in Table~\ref{tab:BSubsetsVsk} labeled by $2^{j}$. For a particular $k$ value, the entry represents a subsequence of length $k$ ending at index $2^{j}$. From the ordering of the binary code basis elements, we see that the entry is given by $b(2^{j+1} - k)$. Thus, the total number is given by
     \begin{equation}
         \sum\limits_{j=\lceil \operatorname{log}_2(k) \rceil}^{n-1} 2^{ \vert b(2^{j+1} - k) \vert -1}.
     \end{equation}
     This can be simplified by noting that moving a column to the left reduces the weight of the string by one. Thus, we consider the weight of the last string in the row and move left. For the last column, the entry is given by
     $\bar{k} \coloneqq 2^{n} - k$. Thus, the number of QC sets for a given $k$ is given by
     \begin{equation}
         \sum\limits_{i=0}^{\overline{n}_k - 1} 2^{ \vert b(\bar{k}) \vert - i - 1} 
         = 2^{\vert b(\bar{k}) \vert}\left[ 1- 2^{-\overline{n}_k} \right] .
     \end{equation}
     Thus, for a given $1 \leq K \leq 2^{n-1}$, we get 
     \begin{equation}
         1 + \sum\limits_{k=1}^K 2^{\vert b(\bar{k}) \vert}\left[ 1- 2^{-\overline{n}_k} \right]
     \end{equation}
     qubit-wise commuting sets. The case~$K=0$ is equal to the value of $K=1$ since the kinetic energy term accounts for two off-diagonal terms anyway.
     
     Lastly, for $K>2^{n-1}$, we have considered all subsets, leading to
     \begin{equation}
         1 + \frac{1}{2} \sum\limits_{i=1}^n 2^i \binom{n}{i} = \frac{1}{2} \left( 1 + 3^n \right),
     \end{equation}
     concluding the proof.
\end{proof} 

\begin{lemma}
    \label{lem:NumQCGGray}
    The number of qubit-wise commuting sets for $H_{N, K}$ for the Gray code is given by
    \begin{equation}
        \vert H(N, K) \vert_C 
        =  \begin{cases} 
                1 + n & K = 0 \\
                1 + \sum\limits_{k=1}^K \overline{n}_k 2^{\vert g_{k-1} \vert} & 1 \leq K \leq 2^{n-1} \\
                \frac{1}{2} \left( 1 + 3^n \right) & K>2^{n-1},
            \end{cases}
    \end{equation}
    where $\vert w \vert$ is the Hamming weight of the string $w$, and $\overline{n}_k \coloneqq n - \lceil \operatorname{log}_2(k) \rceil $.
\end{lemma}
\begin{proof}
    Consider $1 \leq K \leq 2^{n-1}$ and consider a column in Table~\ref{tab:GSubsetsVsk} labeled by $2^{j}$. Since the columns represent subsequences ending at index $j$, we see that strings in a column have a fixed structure -- the strings end with $XI^{n - j - 1}$. The first $j$ bits are a Gray representation of the row index $k-1$. For example, for the cell with row index $3$ and column index $2^2$, the entry is of the form 
    \begin{equation}
        g_{3-1} \otimes X \otimes I = XX \otimes X \otimes I.
    \end{equation}
    As seen in Lemma~\ref{lem:subsetCG}, this set contributes $2^{3-1}$ Pauli terms. 

    Thus, for a particular $k$, the number of Pauli terms is 
    \begin{align}
        &\sum\limits_{j = \lceil \operatorname{log}_2(k) \rceil}^{n-1} 2^{\vert g_{k-1} X I^{n - j - 1} \vert -1 } \notag \\
        &= \sum\limits_{j = \lceil \operatorname{log}_2(k) \rceil}^{n-1} 2^{\vert g_{k-1} \vert} \notag \\
        &= \overline{n}_k 2^{\vert g_{k-1} \vert},
    \end{align}
where the first non-zero $j$ for a row $k$ is $\lceil \operatorname{log}_2(k) \rceil$. Thus, for a fixed $1 \leq K \leq 2^{n-1}$, we see that
\begin{equation}
    1 + \sum\limits_{k=1}^K \overline{n}_k 2^{\vert g_{k-1} \vert},
\end{equation}
where the first term arises from the all-$Z$ measurement for the diagonal terms. The case~$K=0$ is equal to the value of $K=1$ since the kinetic energy terms accounts for two off-diagonal terms anyway.

Lastly, for $K > 2^{n-1}$, we consider all possible subsets $D(n, k, f)$ for all $k \in \{1, \ldots, n\}$. Thus, the total is given by
\begin{equation}
    1 + \sum\limits_{i=1}^n 2^{i-1} \binom{n}{i} = \frac{1}{2} \left( 1 + 3^n \right).
\end{equation}
This concludes the proof.
\end{proof} 

\begin{lemma}
    The set of Pauli strings corresponding to $D(n, k, f)$ all pairwise commute with each other. Alternatively, the set of Pauli strings corresponding to $D(n, k, f)$ forms a distance-grouped commuting set.
    \label{lem:dist-k-commute}
\end{lemma}
\begin{proof}
The set of operators in $D(n, k, f)$ are
\begin{multline}
    D(n, k, f) \coloneqq \{ \left(I_{\bar{f}} \otimes X_f \right) P^a_{\bar{f}} \otimes P^{b, \bar{b}}_f \\
    : \forall a \in \{0, 1\}^{\vert \bar{f} \vert },\ \forall b \in \{0, 1\}^{\vert f \vert}\}.
\end{multline}
We already saw that when mapped to Pauli strings, all the operators from the set lead to combinations of the same Pauli strings with different coefficients only. In other words, independent of the value of $a$ and $b$, the operators all map to the same Pauli strings. Since they all map the same set of Pauli strings, without loss of generality, we consider the fixed operator $a = 0^{\vert \bar{f} \vert }$ and $b = 0^{\vert f \vert }$:
\begin{equation}
    \{ \left(I_{\bar{f}} \otimes X_f \right) P^0_{\bar{f}} \otimes P^{0, \bar{0}}_f.
\end{equation}

Now consider, using Lemma~\ref{lem:OpCOB}, we see that
\begin{equation}
    P^0_{\sofbar} = \frac{1}{2^{\sofbar}} \sum_{j=0}^{2^{\sofbar}-1} \tilde{O}_j.
\end{equation}
Next, we see that 
\begin{align}
    P^{0, \bar{0}}_f &= \frac{1}{2^{\sof}} \sum_l \left( (-1)^{0\cdot l} + (-1)^{1 \cdot l} \right) \tilde{O}_l \notag \\
    &= \frac{1}{2^{\sof}} \sum_l (-1)^{1\cdot l} \tilde{O}_l \notag \\
    &= \frac{1}{2^{\sof}} \sum_{l : p(l)\, \mathrm{mod}\,2 =0} \tilde{O}_l.
\end{align}
As mentioned earlier, only terms with even parity $l$ survive the summation. Thus the Pauli strings are of the form
\begin{equation}
    (I \otimes X)\tilde{O}_j \otimes  \tilde{O}_l,
\end{equation}
where we use the shorthand $X$ for $X_f$ and the constraint that $l$ has even parity. Finally, we see that the Pauli strings all commute since
\begin{align}
    &[(I \otimes X)\tilde{O}_j \otimes \tilde{O}_l, (I \otimes X)\tilde{O}_m \otimes \tilde{O}_n] \notag \\
    &= (I \otimes X)[\tilde{O}_j \otimes \tilde{O}_l, (I \otimes X)\tilde{O}_m \otimes \tilde{O}_n] \notag \\
    &\qquad + [(I \otimes X), (I \otimes X)\tilde{O}_m \otimes \tilde{O}_n]\tilde{O}_j \otimes \tilde{O}_l \notag \\
    &= [\tilde{O}_j \otimes \tilde{O}_l, \tilde{O}_m \otimes \tilde{O}_n] \notag \\
    &\qquad + (I \otimes X)[\tilde{O}_j \otimes \tilde{O}_l, (I \otimes X)] \tilde{O}_m \otimes \tilde{O}_n \notag \\
    &\qquad + (I \otimes X) [(I \otimes X), \tilde{O}_m \otimes \tilde{O}_n]\tilde{O}_j \otimes \tilde{O}_l \notag \\
    &\qquad + [(I \otimes X), (I \otimes X)](\tilde{O}_m \otimes \tilde{O}_n)(\tilde{O}_j \otimes \tilde{O}_l) \notag \\
    &= (I \otimes X)[\tilde{O}_j \otimes \tilde{O}_l, (I \otimes X)] \tilde{O}_m \otimes \tilde{O}_n \notag \\
    &\qquad + (I \otimes X) [(I \otimes X), \tilde{O}_m \otimes \tilde{O}_n]\tilde{O}_j \otimes \tilde{O}_l,
\end{align}
where the first term is zero since all $\tilde{O}$ are composed of $I, Z$ only. Next, consider that
\begin{align}
    [\tilde{O}_j \otimes \tilde{O}_l, (I \otimes X)] 
    &= \tilde{O}_j \otimes [\tilde{O}_l, X] \notag \\
    &=0,
\end{align}
where the last equality is because the parity of $l$ is even. Similarly, the other term is also zero since the parity of $n$ is even. Thus, the overall commutator is zero. Thus, the set of Pauli strings corresponding to the set $D(n, k, f)$ all commute. Since the operators in the set are linear combinations of these Pauli strings, they also commute. 
\end{proof}

\begin{lemma}
    The unitary transformation that rotates the computational basis to the common eigenbasis for set of Pauli strings corresponding to $D(n, k, f)$ is given by 
    \begin{equation}
        I_{\bar{f}} \otimes U^{\operatorname{GHZ}}_{f},
    \end{equation}
    where $U^{\operatorname{GHZ}}_f$ is the unitary operator
    \begin{equation}
        U^{\operatorname{GHZ}}_f \coloneqq \prod_{i=\sof}^2 \operatorname{CNOT}_{f_1 \rightarrow f_i} H_{f_1}.
    \end{equation} 
    The number of two-qubit gates in the diagonalizing unitary for $D(k, n, f)$ is thus given by $(\vert f \vert -1)$.
    \label{lem:diagUnitaryDGC}
\end{lemma}
\begin{proof}
    To show the above result, we show that the output basis when the unitary acts on the computational basis is the common eigenbasis of the Pauli operators $D(n, k, f)$. Consider the action on a computational basis state of the form $\ket{i}_{\bar{f}} \otimes \ket{j}_f$.
    \begin{align}
        &I_{\bar{f}} \otimes U^{\operatorname{GHZ}}_{f} (\ket{i}_{\bar{f}} \otimes \ket{j}_f) \notag \\
        &= I_{\bar{f}} \otimes U^{\operatorname{GHZ}}_{f} (\ket{i}_{\bar{f}} \otimes \ket{j_1 \cdots j_k}_f) \notag \\
        &= \ket{i}_{\bar{f}} \otimes \left( \prod_{t=\sof}^2 \operatorname{CNOT}_{1 \rightarrow t} H_1 \ket{j_1 \cdots j_k}_f \right) \notag \\
        &= \ket{i}_{\bar{f}} \otimes \notag \\
        & \left( \frac{1}{\sqrt{2}} \prod_{t=\sof}^2 \operatorname{CNOT}_{1 \rightarrow t} (\ket{0}_{j_1}+ (-1)^{j_1} \ket{1}_{j_1}) \ket{j_2 \cdots j_k}\right) \notag \\
        &= \ket{i}_{\bar{f}} \otimes \left( \frac{1}{\sqrt{2}} (\ket{0 j_2 \cdots j_k}_f + (-1)^{j_1} \ket{1 \bar{j_2} \cdots \bar{j_k}}_f) \right).
        \label{eq:mapped-basis-states}
    \end{align}
    Next, we show that these basis elements are eigenstates of the Pauli operators $D(n, k, f)$. From Lemma~\ref{lem:dist-k-commute}, we know that the Pauli strings associated with $D(n, k, f)$ are given by 
    \begin{multline}
        \left\{ (\tilde{O}_l)_{\bar{f}} \otimes X_f (\tilde{O}_m)_{f} :  \forall l \in \{0, \ldots, 2^{\sofbar}-1\} ,\ \right. \\
        \left. \forall m \in \{ 0, \ldots, 2^{\sof}-1 \}, \  p(m)\,\mathrm{mod}\, 2=0 \right\},
    \end{multline}
    where $p(m)$ is the parity of the bitstring $m$. 
    
    We first consider the action of the above Pauli strings on the individual subspaces $f$ and $\bar{f}$. On the unflipped subspace $\bar{f}$, the action is given by
    \begin{equation}
    \label{eq:unflipped-eigenstate}
        (\tilde{O}_l)_{\bar{f}} \ket{i}_{\bar{f}} = (-1)^{l \cdot i} \ket{i}_{\bar{f}}.
    \end{equation}    
    Thus, $\ket{i}$ is an eigenstate. Next, we consider the flipped subspace.
    \begin{align}
        &X_f (\tilde{O}_m)_{f} \left( \frac{1}{\sqrt{2}} (\ket{0 j_2 \cdots j_k}_f + (-1)^{j_1} \ket{1 \bar{j_2} \cdots \bar{j_k}}_f) \right) \notag \\
        &= X_f \left( \frac{1}{\sqrt{2}} ( (-1)^{m \cdot 0j_2 \cdots j_k} \ket{0 j_2 \cdots j_k}_f + \right. \notag \\
        &\qquad\qquad \left. (-1)^{j_1} (-1)^{m \cdot 1 \bar{j_2} \cdots \bar{j_k}} \ket{1 \bar{j_2} \cdots \bar{j_k}}_f) \right). \notag \\
        &= \frac{1}{\sqrt{2}} ( (-1)^{m \cdot 0j_2 \cdots j_k} \ket{1 \bar{j_2} \cdots \bar{j_k}}_f + \notag \\
        &\qquad\qquad (-1)^{j_1} (-1)^{m \cdot 1 \bar{j_2} \cdots \bar{j_k}} \ket{0 j_2 \cdots j_k}_f) \notag \\
        &= \frac{1}{\sqrt{2}} (-1)^{j_1 + m \cdot 1 \bar{j_2} \cdots \bar{j_k}} \big( \ket{0 j_2 \cdots j_k}_f + \notag \\
        &\qquad (-1)^{j_1 + m \cdot 0j_2 \cdots j_k + m \cdot 1 \bar{j_2} \cdots \bar{j_k}} \ket{1 \bar{j_2} \cdots \bar{j_k}}_f \big).
    \end{align}

    Consider the following equality:
    \begin{align}
         (-1)^{m \cdot 0j_2 \cdots j_k + m \cdot 1 \bar{j_2} \cdots \bar{j_k}} &= (-1)^{m_1 + \cdots + m_k} \notag \\
         &= (-1)^{\sum_i m_i} \notag \\
         &= 1,
    \end{align}
    where the last equality is because the parity of $m$ is even. Thus, the action on the flipped qubits is given by
     \begin{align}
     \label{eq:flipped-eigenstate}
        &X_f (\tilde{O}_m)_{f} \left( \frac{1}{\sqrt{2}} (\ket{0 j_2 \cdots j_k}_f + (-1)^{j_1} \ket{1 \bar{j_2} \cdots \bar{j_k}}_f) \right) \notag \\
        &= \frac{1}{\sqrt{2}} (-1)^{j_1 + m \cdot 1 \bar{j_2} \cdots \bar{j_k}} \bigg( \ket{0 j_2 \cdots j_k}_f +  \notag \\
        &\qquad \hspace{3cm} (-1)^{j_1} \ket{1 \bar{j_2} \cdots \bar{j_k}}_f \bigg).
    \end{align}
    Thus, using Eqs.~\eqref{eq:unflipped-eigenstate} and \eqref{eq:flipped-eigenstate}, the action of any of the Pauli strings on the rotated basis states Eq.~\eqref{eq:mapped-basis-states} is given by
    \begin{align}
        &\left[(\tilde{O}_l)_{\bar{f}} \otimes X_f (\tilde{O}_m)_{f}\right] I_{\bar{f}} \otimes U^{\operatorname{GHZ}}_{f} (\ket{i}_{\bar{f}} \otimes \ket{j}_f) \notag \\
        &= \left[(-1)^{l \cdot i + j_1 + m \cdot 1 \bar{j_2} \cdots \bar{j_k}}\right] I_{\bar{f}} \otimes U^{\operatorname{GHZ}}_{f} (\ket{i}_{\bar{f}} \otimes \ket{j}_f).
    \end{align}
    Thus, we see that the unitary transformation that rotates the computational basis to the common eigenbasis for the set of Pauli strings corresponding to $D(n, k, f)$ is given by
    \begin{equation}
        I_{\bar{f}} \otimes U^{\operatorname{GHZ}}_f,
    \end{equation}
    concluding the proof.
\end{proof} 

\begin{lemma}
    \label{lem:NumDGCBinGray}
    The number of distance-grouped commuting sets for $H_{N, K}$ for the binary and the Gray code is given by 
    \begin{equation}
        \vert H(N, K) \vert_C
        =  \begin{cases} 
                1+n & K = 0 \\
                1 + \sum\limits_{k=1}^{K} \overline{n}_k & 1 \leq K \leq 2^{n-1} \\
                2^n & K > 2^{n-1}.
            \end{cases}
    \end{equation}
\end{lemma}
\begin{proof}
    Consider $1 \leq  K \leq 2^{n-1}$ and a column in Table~\ref{tab:BSubsetsVsk} labeled by $2^{j}$. For a particular $k$ value, the entry represents a subsequence of length $k$ ending at index $2^{j}$. From the ordering of the binary code basis elements, we see that the entry is given by $b(2^{j+1} - k)$. Thus, the total number is given by
     \begin{equation}
         \sum\limits_{i=\lceil \operatorname{log}_2(k) \rceil}^{n-1} 1 = \overline{n}_k.
     \end{equation}
     Thus, for a given $1 \leq K \leq 2^{n-1}$, we get 
     \begin{equation}
        1 + \sum\limits_{k=1}^{K} \overline{n}_k 
     \end{equation}
     distance-grouped commuting sets. Lastly, for $K>2^{n-1}$, we have considered all subsets, leading to
     \begin{equation}
         1 + \sum\limits_{i=1}^n \binom{n}{i} = 2^n.
     \end{equation}
     
     We notice that the number of distance-grouped commuting sets depends only on the number of filled entries in the Table~\ref{tab:BSubsetsVsk}, and not the individual entries. This is because each entry contributes a single DGC set. Since the Gray code has the same number of filled entries per row as the binary code, the number of DGC sets are the same. 
\end{proof} 

\begin{lemma}
\label{lem:2QG_Gray}
The number of two-qubit gates in the diagonalizing unitary using the DGC scheme for the Gray code is given by
\begin{equation}
    \vert H(N, K) \vert_{DU} 
    =  \begin{cases} 
            0 & K = 0 \\
            \sum\limits_{k=1}^K \vert g_{k-1} \vert \overline{n}_k & 1 \leq K \leq 2^{n-1} \\
            1 + 2^{n-1}(n-2) & K > 2^{n-1}.
        \end{cases}
\end{equation}
\end{lemma}
\begin{proof}
As seen in Lemma~\ref{lem:diagUnitaryDGC}, the number of two-qubit gates for the diagonalizing unitary for the set $D(n, k, f)$ is $(\vert f \vert - 1)$. Similar to the proof of Lemma~\ref{lem:NumQCGGray}, consider $1 \leq K \leq 2^{n-1}$ and consider a column in Table~\ref{tab:GSubsetsVsk} labeled by $2^{j}$. Since the columns represent subsequences ending at $j$, we see that strings in a column have a fixed structure -- the strings end with $XI^{n - j - 1}$. The first $j$ bits are a Gray representation of the row index $k-1$. For example, the cell with row index $3$ and column index $2^2$, the entry is of the form 
\begin{equation}
    g_{3-1} \otimes X \otimes I = XX \otimes X \otimes I.
\end{equation}
As seen in Lemma~\ref{lem:diagUnitaryDGC}, the diagonalizing unitary for this set contains $2$ two-qubit gates.

Thus, for a particular $k$, the number of two-qubit gates for all the diagonalizing unitaries are is 
\begin{equation}
    \sum\limits_{j = \lceil \operatorname{log}_2(k) \rceil}^{n-1} \vert g_{k-1} \vert =   \overline{n}_k \vert g_{k-1} \vert,
\end{equation}
where the first non-zero $j$ for a row $k$ is $\lceil \operatorname{log}_2(k) \rceil$. Thus, for a fixed $1 \leq K \leq 2^{n-1}$, we see that
\begin{equation}
    \sum\limits_{k=1}^K  \overline{n}_k \vert g_{k-1} \vert.
\end{equation}
The case~$K=0$ is equal to the value of $K=1$ since the kinetic energy term accounts for two off-diagonal terms anyway.

Lastly, for $K > 2^{n-1}$, we consider all possible subsets $D(n, k, f)$ for all $k \in \{1, \ldots, n\}$. Thus, the total is given by
\begin{equation}
    \sum\limits_{i=1}^n (i-1) \binom{n}{i} = 1 + 2^{n-1}(n-2),
\end{equation}
concluding the proof.
\end{proof} 

\begin{lemma}
\label{lem:2QG_Binary}
The number of two-qubit gates in the diagonalizing unitary using the DGC scheme for the binary code is given by
\begin{align}
    &\vert H(N, K) \vert_{DU} \notag \\
    &=  \begin{cases} 
            0.5n(n-1) & K = 0 \\
            0.5\sum\limits_{k=1}^K \overline{n}_k \left[ 2 \vert b(\overline{k}) \vert - 1 - \overline{n}_k \right] & 1 \leq K \leq 2^{n-1} \\
            1 + 2^{n-1}(n-2) & K > 2^{n-1}.
        \end{cases}
\end{align}
\end{lemma}
\begin{proof}
As seen in Lemma~\ref{lem:diagUnitaryDGC}, the number of two-qubit gates for the diagonalizing unitary for the set $D(n, k, f)$ is $(\vert f \vert - 1)$. Similar to the proof of Lemma~\ref{lem:NumQCGBinary}, Consider $1 \leq K \leq 2^{n-1}$ and consider a column in Table~\ref{tab:BSubsetsVsk} labeled by $2^{j}$. For a particular $k$ value, the entry represents a subsequence of length $k$ ending at index $2^{j}$. From the ordering of the binary code basis elements, we see that the entry is given by $b(2^{j+1} - k)$. Thus, the total number is given by
\begin{equation}
    \sum\limits_{j=\lceil \operatorname{log}_2(k) \rceil}^{n-1} \vert b(2^{j+1} - k) \vert -1.
\end{equation}
This can be simplified by noting that moving a column to the left reduces the weight of the string by 1. Thus, we consider the weight of the last string in the row and move left. For the last column, the entry is given by $\bar{k} \coloneqq 2^{n} - k$. Thus, the number of QC sets for a given $k$ is given by
\begin{equation}
    \sum\limits_{i=0}^{\overline{n}_k - 1} \vert b(\bar{k}) \vert - i - 1 \\
    = 0.5\overline{n}_k \left[ 2 \vert b(\overline{k}) \vert - 1 - \overline{n}_k \right],
\end{equation}
where $\overline{n}_k \coloneqq n - \lceil \operatorname{log}_2(k) \rceil $. Thus, for a given $1 \leq K \leq 2^{n-1}$, we get 
\begin{equation}
    0.5\sum\limits_{k=1}^K \overline{n}_k \left[ 2 \vert b(\overline{k}) \vert - 1 - \overline{n}_k \right]
\end{equation}
two-qubit gates. The case~$K=0$ is equal to the value of $K=1$ since the kinetic energy term accounts for two off-diagonal terms anyway.
     
Lastly, for $K > 2^{n-1}$, we consider all possible subsets $D(n, k, f)$ for all $k \in \{1, \ldots, n\}$. Thus, the total is given by
\begin{equation}
    \sum\limits_{i=1}^n (i-1) \binom{n}{i} = 1 + 2^{n-1}(n-2),
\end{equation}
concluding the proof.
\end{proof}

\section{List of Operators}
\label{app:list_operators}
In this section, as an example, we give the complete list of number and ladder operators for all encodings with $N=4$. 

Table~\ref{tab:enc-op-OH} lists the encoded operators for the one-hot encoding. Table~\ref{tab:enc-op-BN} lists the encoded operators for the binary encoding. Lastly, Table~\ref{tab:enc-op-GR} lists the encoded operators for the Gray encoding.

\renewcommand{\arraystretch}{1.5}
\begin{table}
    \centering
    \begin{tabular}{P{3cm}|P{3cm}} \hline
    Fock operator   & Encoded operator \\ \hline
    $\outerprod{0}{0}$ & $0.5(I - Z_1)$ \\
    $\outerprod{1}{1}$ & $0.5(I - Z_2)$ \\
    $\outerprod{2}{2}$ & $0.5(I - Z_3)$ \\
    $\outerprod{3}{3}$ & $0.5(I - Z_4)$ \\
    
    $\outerprod{0}{1} + \outerprod{1}{0}$ & $0.5(X_1X_2 + Y_1Y_2)$ \\
    $\outerprod{1}{2} + \outerprod{2}{1}$ & $0.5(X_2X_3 + Y_2Y_3)$ \\
    $\outerprod{2}{3} + \outerprod{3}{2}$ & $0.5(X_3X_4 + Y_3Y_4)$ \\

    $\outerprod{0}{2} + \outerprod{2}{0}$ & $0.5(X_1X_3 + Y_1Y_3)$ \\
    $\outerprod{1}{3} + \outerprod{3}{1}$ & $0.5(X_2X_4 + Y_2Y_4)$ \\
    
    $\outerprod{0}{3} + \outerprod{3}{0}$ & $0.5(X_1X_4 + Y_1Y_4)$ \\
    \hline
    \end{tabular}
    \caption{Encoded operators for the one-hot encoding with $N=4$. For the ladder operators, each term has its Hermitian conjugate added as well. This is because the Hamiltonian is Hermitian, leading to them having the same coefficient.}
    \label{tab:enc-op-OH}
\end{table}

\begin{table}
    \centering
    \begin{tabular}{P{3cm}|P{4cm}} \hline
    Fock operator   & Encoded operator \\ \hline
    $\outerprod{0}{0}$ & $0.25(II + IZ + ZI + ZZ)$ \\
    $\outerprod{1}{1}$ & $0.25(II - IZ + ZI - ZZ)$ \\
    $\outerprod{2}{2}$ & $0.25(II + IZ - ZI - ZZ)$ \\
    $\outerprod{3}{3}$ & $0.25(II - IZ - ZI + ZZ)$ \\
    
    $\outerprod{0}{1} + \outerprod{1}{0}$ & $0.5 (IX + ZX)$ \\
    $\outerprod{1}{2} + \outerprod{2}{1}$ & $0.5 (XX + YY)$ \\
    $\outerprod{2}{3} + \outerprod{3}{2}$ & $0.5 (IX - ZX)$ \\

    $\outerprod{0}{2} + \outerprod{2}{0}$ & $0.5 (XI + XZ)$ \\
    $\outerprod{1}{3} + \outerprod{3}{1}$ & $0.5 (XI - XZ)$ \\
    
    $\outerprod{0}{3} + \outerprod{3}{0}$ & $0.5 (XX - YY)$ \\
    \hline
    \end{tabular}
    \caption{Encoded operators for the binary encoding with $N=4$. For the ladder operators, each term has its Hermitian conjugate added as well. This is because the Hamiltonian is Hermitian, leading to them having the same coefficient.}
    \label{tab:enc-op-BN}
\end{table}

\begin{table}
    \centering
    \begin{tabular}{P{3cm}|P{4cm}} \hline
    Fock operator   & Encoded operator \\ \hline
    $\outerprod{0}{0}$ & $0.25(II + IZ + ZI + ZZ)$ \\
    $\outerprod{1}{1}$ & $0.25(II + IZ - ZI - ZZ)$ \\
    $\outerprod{2}{2}$ & $0.25(II - IZ - ZI + ZZ)$ \\
    $\outerprod{3}{3}$ & $0.25(II - IZ + ZI - ZZ)$ \\
    
    $\outerprod{0}{1} + \outerprod{1}{0}$ & $0.5 (XI + XZ)$ \\
    $\outerprod{1}{2} + \outerprod{2}{1}$ & $0.5 (IX - ZX)$ \\
    $\outerprod{2}{3} + \outerprod{3}{2}$ & $0.5 (XI - XZ)$ \\

    $\outerprod{0}{2} + \outerprod{2}{0}$ & $0.5 (XX - YY)$ \\
    $\outerprod{1}{3} + \outerprod{3}{1}$ & $0.5 (XX + YY)$ \\
    
    $\outerprod{0}{3} + \outerprod{3}{0}$ & $0.5 (IX + ZX)$ \\
    \hline
    \end{tabular}
    \caption{Encoded operators for the Gray encoding with $N=4$. For the ladder operators, each term has its Hermitian conjugate added as well. This is because the Hamiltonian is Hermitian, leading to them having the same coefficient.}
    \label{tab:enc-op-GR}
\end{table}

\section{Simulation Details for n+C}
\label{app:sim_details_carbon}

We now provide specific simulation details in Table~\ref{tab:details-n+c} for the different simulations for the energy of the lowest $\frac{1}{2}^+$ state for the n+C systems under consideration. We mainly use the Gray encoding, with truncation parameter $N=8$ $(n=3)$ and  $N=16$ $(n=4)$, and $K=3$. For all Gray encoding simulations, we use $L=4$ layers.

For gradient descent, we use a adaptive learning rate scheme. Every $10$ iterations, fit a straight line of the last $10$ cost function values. If the slope is negative, increase learning rate $lr$ to $\operatorname{\min}(1.05lr, lr_{\max})$. If the slope is positive, reduce learning rate $lr$ to $\operatorname{\max}(0.8lr, lr_{\min})$.

\begin{table*}
\centering
    \begin{tabular}{P{2cm}| P{2cm} | P{2.5cm} | M{9.5cm}}
    & $N$ & Type & Details \\
    \hline
    \hline
    \multirow{3}{*}{n$+^{10}$C} & \multirow{3}{*}{$8, 16$ - Gray} & \multirow{1}{*}[1em]{Noiseless} & $K=1$ for $500$ iterations followed by $K=3$ for $500$ iterations, both using SPSA. Lastly, $K=3$ for $1000$ iterations using gradient descent and a varying learning rate scheme. \\
    \cline{3-4}
    & & Shot Noise & Same as noiseless but using a shot based estimator for $10^3$ shots.\\
    \cline{3-4}
    & & \multirow{1}{*}[1em]{Noisy} & $K=1$ for $500$ iterations using $10^3$ shots, followed by $K=3$ for $1500$ iterations using $2 \times 10^4$ shots, both using SPSA. Cost function estimated using a fake IBMQ backend {\tt ibm\_manila}. \\
    \hline

    \multirow{3}{*}{n$+^{12}$C} & \multirow{3}{*}{$8$ - Gray} & \multirow{1}{*}[1em]{Noiseless} & $K=1$ for $500$ iterations followed by $K=3$ for $500$ iterations, both using SPSA. Lastly, $K=3$ for $1000$ iterations using gradient descent and a varying learning rate scheme. \\
    \cline{3-4}
    & & Shot Noise & Same as noiseless but using a shot based estimator for $10^3$ shots.\\
    \cline{3-4}
    & & \multirow{1}{*}[1em]{Noisy} & $K=1$ for $500$ iterations using $10^3$ shots, followed by $K=3$ for $1500$ iterations using $2 \times 10^4$ shots, both using SPSA. Cost function estimated using a fake IBMQ backend {\tt ibm\_manila}. \\
    \hline

    \multirow{3}{*}{n$+^{12}$C} & \multirow{3}{*}{$16$ - Gray} & \multirow{1}{*}[1em]{Noiseless} & $K=1$ for $500$ iterations followed by $K=3$ for $1000$ iterations, both using SPSA. Lastly, $K=3$ for $1500$ iterations using gradient descent and a varying learning rate scheme. \\
    \cline{3-4}
    & & Shot Noise & Same as noiseless but using a shot based estimator for $10^3$ shots.\\
    \cline{3-4}
    & & \multirow{1}{*}[1em]{Noisy} & $K=1$ for $500$ iterations using $10^3$ shots, followed by $K=3$ for $1500$ iterations using $2 \times 10^4$ shots, both using SPSA. Cost function estimated using a fake IBMQ backend {\tt ibm\_manila}. \\
    \hline

    \multirow{3}{*}{n$+^{14}$C} & \multirow{3}{*}{$8$ - Gray} & \multirow{1}{*}[1em]{Noiseless} & $K=1$ for $500$ iterations followed by $K=3$ for $500$ iterations, both using SPSA. Lastly, $K=3$ for $1000$ iterations using gradient descent and a varying learning rate scheme. \\
    \cline{3-4}
    & & Shot Noise & Same as noiseless but using a shot based estimator for $10^3$ shots.\\
    \cline{3-4}
    & & \multirow{1}{*}[1em]{Noisy} & $K=1$ for $500$ iterations using $10^3$ shots, followed by $K=3$ for $1500$ iterations using $2 \times 10^4$ shots, both using SPSA. Cost function estimated using a fake IBMQ backend {\tt ibm\_manila}. \\
    \hline

    \multirow{3}{*}[-0.5em]{n$+^{14}$C} & \multirow{3}{*}[-0.5em]{$16$ - Gray} & \multirow{1}{*}[1em]{Noiseless} & $K=1$ for $500$ iterations followed by $K=3$ for $1000$ iterations, both using SPSA. Lastly, $K=3$ for $1500$ iterations using gradient descent and a varying learning rate scheme. \\
    \cline{3-4}
    & & Shot Noise & Same as noiseless but using a shot based estimator for $10^3$ shots.\\
    \cline{3-4}
    & & \multirow{1}{*}[1em]{Noisy} & $K=1$ for $500$ iterations using $10^3$ shots, followed by $K=3$ for $2500$ iterations using $2 \times 10^4$ shots, both using SPSA. Cost function estimated using a fake IBMQ backend {\tt ibm\_manila}. \\
    \hline

    \multirow{3}{*}[-0.5em]{n$+^{14}$C} & \multirow{3}{*}[-0.5em]{$8$ - OH} & \multirow{1}{*}[0.3em]{Noiseless} & $K=1$ for $500$ iterations followed by $K=3$ for $1500$ iterations, both using SPSA.\\
    \cline{3-4}
    & & Shot Noise & Same as noiseless but using a shot based estimator for $10^3$ shots.\\
    \cline{3-4}
    & & \multirow{1}{*}[1em]{Noisy} & $K=1$ for $500$ iterations followed by $K=3$ for $1500$ iterations, both using SPSA. Cost function estimated using a fake IBMQ backend {\tt ibm\_manila}.\\
    \hline
    
    \end{tabular}
    \caption{Simulation details for n$+$C. Plots shown in Fig.~\ref{fig:n_10C_plot} for n$+^{10}$C Gray, Fig.~\ref{fig:n_12C_plot} for n$+^{12}$C Gray, Fig.~\ref{fig:n_14C_plot} for n$+^{14}$C Gray, and Fig.~\ref{fig:n_14C_plot_OH} for n$+^{14}$C OH.}
    \label{tab:details-n+c}
\end{table*}

\section{Simulation Details for n\texorpdfstring{$+\alpha$}{+alpha}}
\label{app:sim_details_abi}

We now provide specific simulation details for the different simulations in Table~\ref{tab:details-abi}  for the energy of the lowest $\frac{1}{2}^+$ orbit for the n+$\alpha$ optical potential derived \textit{ab initio}. For all the simulations, we use the Gray encoding, with truncation parameter $N=8$ $(n=3)$ and  $N=16$ $(n=4)$, and $K=1, 2$. For all Gray encoding simulations, we use $L=5$ layers.

For gradient descent, we use an adaptive learning rate scheme. Every $10$ iterations, fit a straight line of the last $10$ cost function values. If the slope is negative, increase the learning rate $lr$ to $\operatorname{\min}(1.05lr, lr_{\max})$. If the slope is positive, reduce the learning rate $lr$ to $\operatorname{\max}(0.8lr, lr_{\min})$.

\begin{table*}
\centering
    \begin{tabular}{P{2cm}| P{2cm} | P{2.5cm} | M{10cm}}
    $\hbar \omega$ & $N, K$ & Type & Details \\
    \hline
    \hline
    \multirow{3}{*}[-0.5em]{$12$} & \multirow{3}{*}[-0.5em]{$8, 1$} & \multirow{1}{*}[0.5em]{Noiseless} & $K=0$ for $500$ iterations using SPSA, followed by $K=1$ for $500$ iterations using gradient descent and a varying learning rate scheme.\\
    \cline{3-4}
    & & Shot Noise & Same as noiseless but using a shot based estimator for $10^5$ shots.\\
    \cline{3-4}
    & & \multirow{1}{*}[1em]{Noisy} & $K=0$ for $500$ iterations using $10^3$ shots, followed by $K=1$ for $500$ iterations using $10^5$ shots, both using SPSA. Cost function estimated using a fake IBMQ backend {\tt ibm\_manila}.\\
    \hline

    \multirow{3}{*}[-0.5em]{$12$} & \multirow{3}{*}[-0.5em]{$8, 2$} & \multirow{1}{*}[1em]{Noiseless} & $K=0$ for $500$ iterations, $K=1$ for $500$ iterations, followed by $K=2$ for $500$ iterations, all using SPSA, followed by $K=2$ for $500$ iterations using gradient descent and a varying learning rate scheme. \\
    \cline{3-4}
    & & Shot Noise & Same as noiseless but using a shot based estimator for $10^5$ shots.\\
    \cline{3-4}
    & & \multirow{1}{*}[1.5em]{Noisy} & $K=0$ for $500$ iterations using $10^3$ shots, $K=1$ for $500$ iterations using $10^3$ shots, followed by $K=2$ for $500$ iterations using $10^5$ shots, all using SPSA. Cost function estimated using a fake IBMQ backend {\tt ibm\_manila}. \\
    \hline

    \multirow{3}{*}[-1em]{$12$} & \multirow{3}{*}[-1em]{$16, 1$} & \multirow{1}{*}[1em]{Noiseless} & $K=0$ for $500$ iterations, $K=1$ for $1000$ iterations, both using SPSA, followed by $K=1$ for $1000$ iterations using gradient descent and a varying learning rate scheme. \\
    \cline{3-4}
    & & \multirow{1}{*}[1em]{Shot Noise} & $K=0$ for $500$ iterations, $K=1$ for $500$ iterations, both using SPSA, following by $K=1$ for $1500$ iterations using gradient descent and a varying learning rate scheme. All estimations use $10^5$ shots.\\
    \cline{3-4}
    & & \multirow{1}{*}[1em]{Noisy} & $K=0$ for $1000$ iterations using $10^3$ shots, followed by $K=1$ for $1500$ iterations using $10^5$ shots, both using SPSA. Cost function estimated using a fake IBMQ backend {\tt ibm\_manila}. \\
    \hline

    \multirow{3}{*}[-2em]{$12$} & \multirow{3}{*}[-2em]{$16, 2$} & \multirow{1}{*}[1em]{Noiseless} & $K=0$ for $500$ iterations, $K=1$ for $500$ iterations, $K=2$ for $1000$ iterations, all using SPSA, followed by $K=2$ for $1000$ iterations using gradient descent and a varying learning rate scheme.\\
    \cline{3-4}
    & & \multirow{1}{*}[1.5em]{Shot Noise} & $K=0$ for $500$ iterations, $K=1$ for $500$ iterations, $K=2$ for $500$ iterations, all using SPSA, following by $K=2$ for $1500$ iterations using gradient descent and a varying learning rate scheme. All estimations use $10^5$ shots.\\
    \cline{3-4}
    & & \multirow{1}{*}[1.5em]{Noisy} & $K=0$ for $500$ iterations using $10^3$ shots, $K=1$ for $500$ iterations using $10^3$ shots, $K=2$ for $500$ iterations using $10^3$ shots, followed by $K=2$ for $500$ iterations using $10^5$ shots, all using SPSA. Cost function estimated using a fake IBMQ backend {\tt ibm\_manila}. \\
    \hline

    \multirow{3}{*}[-0.5em]{$16$} & \multirow{3}{*}[-0.5em]{$8, 1$} & \multirow{1}{*}[0.5em]{Noiseless} & $K=0$ for $500$ iterations using SPSA, followed by $K=1$ for $500$ iterations using gradient descent and a varying learning rate scheme.\\
    \cline{3-4}
    & & Shot Noise & Same as noiseless but using a shot based estimator for $10^5$ shots.\\
    \cline{3-4}
    & & \multirow{1}{*}[1em]{Noisy} & $K=0$ for $500$ iterations using $10^3$ shots, followed by $K=1$ for $500$ iterations using $10^5$ shots, both using SPSA. Cost function estimated using a fake IBMQ backend {\tt ibm\_manila}.\\
    \hline

    \multirow{3}{*}[-1em]{$16$} & \multirow{3}{*}[-1em]{$8, 2$} & \multirow{1}{*}[0.5em]{Noiseless} & $K=0$ for $750$ iterations, $K=1$ for $750$ iterations, followed by $K=2$ for $750$ iterations, all using SPSA.\\
    \cline{3-4}
    & & Shot Noise & Same as noiseless but using a shot based estimator for $10^5$ shots.\\
    \cline{3-4}
    & & \multirow{1}{*}[1.5em]{Noisy} & Same as noiseless, but using a fake IBMQ backend {\tt ibm\_manila} with $10^3$, $10^3$, and $10^5$ shots for $K=0, 1, 2$, respectively, to estimate expectation values. Cost function estimated using a fake IBMQ backend {\tt ibm\_manila}.\\
    \hline
    
    \end{tabular}
    \caption{Simulation details for n$+\alpha$. Plots shown in Fig.~\ref{fig:abi-n-3-12} for $\hbar \omega = 12, N=8$, Fig.~\ref{fig:abi-n-4-12} for $\hbar \omega = 12, N=16$, and  Fig.~\ref{fig:abi-n-3-16} for $\hbar \omega = 16, N=8$.}
    \label{tab:details-abi}
\end{table*}

\end{document}